\documentclass
  [
    DIV=14,
    fontsize=12,
    abstract=true,
    bibliography=totoc
  ]
  {scrartcl}

\setkomafont{date}{\small\rmfamily}
\setkomafont{author}{\large\sffamily}

\usepackage{authblk}

\newcommand
{\NameEmail}[2]
{%
  \stackunder%
  {#1}%
  {%
   \clap{%
   \normalfont%
   \footnotesize%
   \texttt{#2}%
  }}%
  \hspace{-3.5pt}%
}

\newcommand\keywords[1]{%
  \begingroup
  \renewcommand\thefootnote{}
  \footnote{{Keywords:} #1}%
  \addtocounter{footnote}{-1}
  \endgroup
}
\newcommand\MSCcodes[1]{%
  \begingroup
  \renewcommand\thefootnote{}
  \footnote{{MSC codes:} #1}%
  \addtocounter{footnote}{-1}
  \endgroup
}
\newcommand\funding[1]{%
  \begingroup
  \renewcommand\thefootnote{}
  \footnote{{Funding:} #1}%
  \addtocounter{footnote}{-1}
  \endgroup
}

\usepackage{scrlayer-scrpage}
\clearpairofpagestyles 
\ihead{} 
\ohead{\headmark}
\automark{section}

\NewCommandCopy{\oldparagraph}{\paragraph}
\RenewDocumentCommand{\paragraph}{s t{!} O{#4} m}{%
  \IfBooleanTF{#1}
    {\IfBooleanTF{#2}{\oldparagraph*{#4}}{\oldparagraph*{#4.}}}
    {\IfBooleanTF{#2}{\oldparagraph[#3]{#4}}{\oldparagraph[#3]{#4.}}}%
}

\usepackage{caption}
\usepackage
  [leftcaption]
  {sidecap}

\usepackage{enumitem}
\setlist[enumerate,1]{label=\textbf{\textup{(\roman*)}}}
\setlist[enumerate,2]{label=\textbf{\textup{(\alph*)}}}
\setlist{
  topsep=1pt,
  itemsep=2pt,
  leftmargin=1cm,
  before=\leavevmode
}

\usepackage[
  backend=biber,
  style=alphabetic,
  giveninits=true, 
  backref=true
]{biblatex}
\usepackage[T1]{fontenc}
\usepackage{lmodern}
\usepackage{anyfontsize}

\usepackage{multirow}
\usepackage{hhline}

\usepackage{tabularray}
\UseTblrLibrary{booktabs}

\usepackage{adjustbox}
\usepackage{tcolorbox}

\newtcolorbox{standout}{
  colback=gray!15,
  boxrule=0pt,
  left=.3cm,
  right=.3cm,
  top=.18cm,
  bottom=.18cm,
  boxsep=0pt
}

\usepackage{amsmath}
\allowdisplaybreaks
\usepackage{mathtools}
\usepackage{amssymb}
\usepackage{amsthm}
\usepackage{xfrac}
\usepackage{stmaryrd} 
\usepackage{scalerel} 
\usepackage{stackengine} 
\usepackage[safe]{tipa} 
\usepackage{centernot}
\usepackage[
  cal=euler,
  scr=dutchcal
]
 {mathalpha} 

 \newcommand{\bracket}[3]{%
  \stretchleftright
    {#1}
    {%
      \ensurestackMath{\addstackgap[1pt]{#2}}%
      \vrule width 0pt depth 2pt height 0pt
    }
    {#3}%
} 
\newcommand{\scaledbracket}[3]{%
  \ThisStyle{%
    \stretchleftright
      {#1}
      {
        \ensurestackMath{\addstackgap[1pt]{\SavedStyle #2}}%
        \vrule width 0pt depth 1.5pt height 0pt
      }
      {#3}%
  }%
}
\newcommand{\bracketmid}[4]{%
  \stretchleftright{#1}{%
    \ensurestackMath{%
      \addstackgap[1pt]{#2}%
      \,\stretchrel*{|}{\addstackgap[1pt]{#2#3}}\,%
      \addstackgap[1pt]{#3}%
    }%
  }{#4}%
}

\theoremstyle{plain}
\newtheorem{theorem}{Theorem}[section]
\newtheorem{lemma}[theorem]{Lemma}
\newtheorem{proposition}[theorem]{Proposition}
\newtheorem{corollary}[theorem]{Corollary}
\theoremstyle{definition}
\newtheorem{definition}[theorem]{Definition}
\newtheorem{example}[theorem]{Example}

\theoremstyle{remark}
\newtheorem{remark}[theorem]{Remark}
\usepackage[
  colorlinks=true,
  linkcolor=darkgreen,
  citecolor=darkgreen,
  urlcolor=darkgreen
]{hyperref}
\usepackage{xurl} 

\DeclareSourcemap{
  \maps[datatype=bibtex]{
    \map{
      \step[fieldsource=url, final=true]
      \step[fieldset=verba, origfieldval, final=true]
      \step[fieldsource=verba, match=\regexp{\Ahttps?://}, replace={}]
    }
  }
}
\DeclareFieldFormat{url}{%
  \mkbibacro{URL}\addcolon\space
  \href{#1}{\nolinkurl{\thefield{verba}}}%
}

\usepackage{cleveref}
\crefname{equation}{}{}
\crefname{section}{\S}{\S\S}
\crefname{subsection}{\S}{\S\S}
\crefname{subsubsection}{\S}{\S\S}
\crefname{definition}{Def.}{Defs.}
\crefname{theorem}{Thm.}{Thms.}
\crefname{corollary}{Cor.}{Cors.}
\crefname{lemma}{Lem.}{Lems.}
\crefname{proposition}{Prop.}{Props.}
\crefname{remark}{Rem.}{Rems.}
\crefname{notation}{Ntn.}{Ntns.}
\crefname{fact}{Fact}{Facts}
\crefname{example}{Ex.}{Exs.}
\crefname{figure}{Fig.}{Figs.}
\crefname{table}{Tab.}{Tabs.}
\crefname{footnote}{ftn.}{ftns.}
\Crefname{footnote}{Ftn.}{Ftns.}

\usepackage{tikz}
\usetikzlibrary{
  cd,
  calc,
  arrows.meta,
  backgrounds,
  decorations,
  decorations.pathmorphing, 
}

\tikzcdset{
  arrow style=tikz, 
  diagrams={
    cramped,
    baseline=
    {([yshift=-2pt]current bounding box.center)},
    >={
      Computer Modern Rightarrow[
        length=4pt, width=4pt
      ]
    },
    hook/.style={
      {Hooks[right, length=2pt, width=8pt]}->
    },
    hook'/.style={
      {Hooks[left, length=2pt, width=8pt]}->
    },
    shorten=0pt,
  }
}

\usepackage{xcolor}
\definecolor{darkblue}{rgb}{0.05,0.25,0.65}
\definecolor{darkgreen}{RGB}{20,140,10}
\definecolor{lightgray}{rgb}{0.9,0.9,0.9}
\definecolor{darkorange}{RGB}{200,100,5}
\definecolor{darkyellow}{rgb}{.91,.91,0}
\definecolor{lightolive}{RGB}{225, 220, 185}

\let\originalsslash\sslash
\renewcommand{\sslash}{\mathord{\originalsslash}}

\makeatletter
\newcommand{\cpt}{\mathpalette\cpt@inner\relax}
\newcommand{\cpt@inner}[2]{%
  \scalebox{0.5}[0.9]{$#1\cup$}
  #1\{\infty\}
}
\makeatother

\makeatletter
\newcommand{\plus}{\mathpalette\sqcpt@inner\relax}
\newcommand{\sqcpt@inner}[2]{%
  \scalebox{0.5}[0.9]{$#1\sqcup$}
  #1\{\infty\}
}
\makeatother

\newcommand{\grayunderbrace}[2]{\mathcolor{gray}{\underbrace{\mathcolor{black}{#1}}}_{\mathcolor{gray}{#2}}}
\newcommand{\grayoverbrace}[2]{\mathcolor{gray}{\overbrace{\mathcolor{black}{#1}}}^{\mathcolor{gray}{#2}}}

\tikzset{
  snake left/.style={
    rounded corners,
    to path={
      let \p1 = (\tikztostart.east),
          \p2 = (\tikztotarget.west),
          \p3 = ($(\p1)!0.5!(\p2)$),
          \n1 = {8pt} 
      in
      (\p1)
      -- (\x1 + \n1, \y1)
      -- (\x1 + \n1, \y3)
      -- (\x2 - \n1, \y3) \tikztonodes
      -- (\x2 - \n1, \y2)
      -- (\p2)
    }
  }
}

\tikzset{
  uphordown/.style={
    rounded corners,
    to path={
      let \p1 = (\tikztostart.north),
          \p2 = (\tikztotarget.north),
          \n1 = {max(\y1,\y2) + 8pt}
      in
      (\p1)
      -- (\x1, \n1)
      -- (\x2, \n1) \tikztonodes 
      -- (\p2)
    }
  }
}

\tikzset{
  downhorup/.style={
    rounded corners,
    to path={
      let \p1 = (\tikztostart.south),
          \p2 = (\tikztotarget.south),
          \n1 = {min(\y1,\y2) - 8pt}
      in
      (\p1)
      -- (\x1, \n1)
      -- (\x2, \n1) \tikztonodes 
      -- (\p2)
    }
  }
}

\tikzset{
  rightvertleft/.style={
    rounded corners,
    to path={
      let \p1 = (\tikztostart.east),
          \p2 = (\tikztotarget.east),
          \n1 = {max(\x1,\x2) + 8pt}
      in
      (\p1)
      -- (\n1, \y1)
      -- (\n1, \y2) \tikztonodes 
      -- (\p2)
    }
  }
}

\tikzset{
  leftvertright/.style={
    rounded corners,
    to path={
      let \p1 = (\tikztostart.west),
          \p2 = (\tikztotarget.west),
          \n1 = {min(\x1,\x2) - 8pt}
      in
      (\p1)
      -- (\n1, \y1)
      -- (\n1, \y2) \tikztonodes 
      -- (\p2)
    }
  }
}

\newcommand{\inlinetikzcd}[1]{\begin{tikzcd}[sep=small, ampersand replacement=\&]#1\end{tikzcd}}

\renewcommand{\setminus}{-}

\newcommand{\defneq}{\mathrel{\equiv}}

\newcommand{\HilbertSpace}{%
  \mathcal{H}%
}

\newcommand{\truncation}[2]
  {\bracket[{#2}]_{#1}}

\colorlet{IntBraneColor}{darkblue}
\colorlet{FracBraneColor}{darkorange}

\newcommand
  {\toroidification}
  {\mathrm{Trd}}

\begin{document}

\setlength{\abovedisplayskip}{3.5pt}
\setlength{\belowdisplayskip}{3.5pt}
\setlength{\abovedisplayshortskip}{-5pt}
\setlength{\belowdisplayshortskip}{3pt}


\title{On the Derivation of \\ Twisted K-Theory from M-Theory}

\author[1]
{\NameEmail
  {Nikita Golub}
  {ngolub@nyu.edu}
}

\author[1,2]
{\NameEmail
  {Hisham Sati}
  {hsati@nyu.edu}
}

\author[1]
{\NameEmail
  {Urs Schreiber}
  {us13@nyu.edu}
}

\affil[1]{Mathematics Program and Center for Quantum and Topological Systems (CQTS), \newline New York University Abu Dhabi, UAE}
\affil[2]{The Courant Institute of Mathematical Sciences, \newline New York University, New York, USA}

\maketitle

\begin{abstract}
In a restricted untwisted sector,
Diaconescu, Moore \& Witten \cite{DMW2000,DMW2003} gave a quantum M-theoretic consistency check, at the level of partition functions, of what one may call \emph{Hypothesis K}, namely that magnetic and electric type IIA brane charges are jointly quantized in \emph{K}-theory.

We show that the more recently discussed \emph{Hypothesis H} --- that magnetic and electric charges in 11D are jointly
quantized in co\emph{H}omotopy --- together with a topological smallness condition on the M-theory circle, systematically implies that IIA charges are quantized in a nonabelian deformation of twisted K-theory by the previously neglected non-linear NS1-brane charge. 

This is essentially a derivation of twisted K-theory from IR-completed 11D SuGra, and conversely a further consistency check on \emph{Hypothesis H} from \emph{Hypothesis K}. 

We motivate and precisely state this result, discuss its structural implications, and outline the elaborate algebro-topological proof.
\end{abstract}


\MSCcodes{
  Primary:
 81T30, 
 81T70, 
 19L50, 
 55N20, 
 55P35, 
 Secondary:
 55Q55, 
 55S25, 
 53C08, 
 55P62. 
}

\funding{
  by \textit{Tamkeen UAE} under the \textit{Abu Dhabi Research Institute Grant} \texttt{CG008}
}

\newpage

\tableofcontents

\section
{Motivation}

\subsection
{IR-Completion of SuGra}
\label
{IRCompletionOfSuGra}

It is an open secret that higher gauge fields (in the sense of \parencites{Szabo2013}[\S~2.1]{SS24-Phase}[\S~II]{MooreSaxena2025}{SS26-HigherGauge}, cf. also  \parencites{Alfonsi2025HigherGeometry}{BorstenEtAl2025}) tend to be discussed in the (typically: Lagrangian) field and string theory literature (such as \parencites{HenneauxTeitelboim1992}{Costello2011}{Rejzner2016} and \parencites{Polchinski1998-Vol2}{BlumenhagenTheisenLuest2013}{Ortin2015}) only locally on charts of spacetime, where gauge potentials may be regarded as plain differential forms (cf. \cref{TransitionData}). 

This is the case notably in 
\begin{enumerate}
\item
the usual presentation of (the higher gauge sectors of) higher-dimensional supergravity theories (SuGra, cf. \parencites
{DuffNilssonPope1986}
{SalamSezgin1989}
{CDF1991}
{FreedmanVanProeyen2012}
{Ortin2015}
{Sezgin2023}%
), 
\item
such as that of 11D  SuGra (\parencites{CremmerJuliaScherk1978}, cf. \parencites[\S~I]{Duff1999World}{MiemiecSchnakenburg2006}[\S~10]{FreedmanVanProeyen2012}[\S~22.1.1]{Ortin2015}{DuffEtAl2026}), 
\item
and its dimensional reduction to 10D IIA SuGra (\parencites{GianiPernici1984}{HuqNamazie1985}{CampbellWest1984}, cf. \parencites[\S~12.1]{Polchinski1998-Vol2})
that we are concerned with here (following \parencites{GSS24-SuGra}{GiotopoulosSati2026}, cf. \parencites{GSS26-SuGra}). 
\end{enumerate}

Beyond that, the \emph{IR-completion} of the field content---capturing and identifying the nature of topological brane charges---involves a choice of generalized \emph{flux/charge quantization} law%
\footnote{%
\label{MeaningOfQuantization}%
  Beware that the term ``quantization'' in ``flux/charge quantization'' is in the sense of \emph{discretization} (as in \emph{flux quanta} and \emph{charge units}), not in the sense of \emph{quantum physics} (even if both are related: flux quantization of background fluxes tends to provide quantum anomaly cancellation for probe brane sigma-models coupled to these fluxes, cf. \parencites[\S~5]{WenWitten1985}[\S~1]{Alvarez1985} and \parencites{FSS21-Hopf}): All charges discussed here are classical observables; their possible quantum observable algebra is a separate discussion (for aspects of which see instead \parencites{FreedMooreSegal2007}{FreedMooreSegal2007b}{SS24-Obs}).
}
(cf. \parencites{Freed2002}{SS25-Flux}{SS26-HigherGauge}{GSS26-SuGra}) that describes how the gauge potentials may be glued by (higher) gauge transformations across charts (cf. again \cref{TransitionData}).  

\begin{figure}[htb]
\caption{\label{TransitionData}
The data of (higher) gauge field configurations $\widehat{A}$ are only locally given by \emph{gauge potential} differential forms $A$ (here short for tuples of forms in various degrees). Beyond that, the IR-completed field content involves a tower of further (\v Cech cocycle) data that enforce \emph{flux quantization laws} and thereby identify the topological charges of the branes that sourced these higher gauge fields (cf. \parencites[\S~4.2]{GSS26-SuGra}).  Conversely, the entire tower of field data turns out to be encoded by the \emph{classifying space} $\mathcal{A}$ of corresponding brane charges (cf. \cref{ProperFluxQuantization}).
}
\centering
\adjustbox
{rndfbox=4pt}{
  \includegraphics[width=16.1cm]{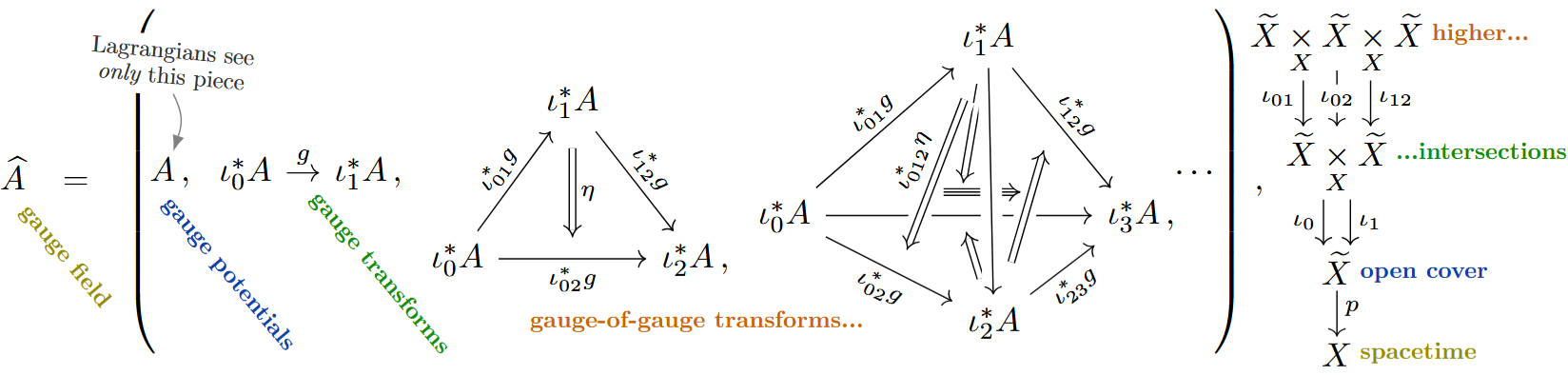}%
}
\end{figure}

\paragraph
{Motivating examples}
This issue is widely appreciated for simple examples:
\begin{enumerate}
  \item 
  The classical choice for IR-completion of 4D Maxwell theory, going back to Dirac 1931 (\cite{Dirac1931}, cf. \parencites{Alvarez1985}[\S~2]{Freed2002}[\S~16.4e]{Frankel2011}[\S~3.1]{SS24-Phase}[\S~2.1]{SS25-Flux}), declares that the electromagnetic field is a cocycle in differential ordinary cohomology in degree 2, with classifying space $B^2 \mathbb{Z}$
  (cf. \cite[\S~2.1]{FSS23-Char}).%
  \footnote{%
    Following \cite{SS25-Flux}, for discrete abelian groups $A$ we abbreviate by ``$B^n A$'' the traditional notation ``$K(A,n)$'' for Eilenberg--MacLane spaces, classifying ordinary cohomology, $H^n\bracket({-; A}) \simeq \pi_0\, \mathrm{Map}\bracket({-,B^n A})$.
    In this notation, we have in particular $B^{n+1}\mathbb{Z} \simeq B^n \mathrm{U}(1)$, which the reader may take as the definition of the latter expression.
  } 

  \item 
  \label{PlainBFieldFluxQuantization}
  The common choice of IR-completion for the NS B-field in 10D IIA supergravity declares it to be a cocycle in differential ordinary cohomology in degree 3 (\parencites{Gawedzki1988}{FreedWitten1999}{Gawedzki2002}{CareyJohnsonMurray2004}{BonoraFerrariSavelli2008}), with classifying space $B^3 \mathbb{Z}$ (cf. \cite[Ex. 3.10]{SS25-Flux}).
\end{enumerate}

It is only after such IR-completion (which is a non-unique choice of physical input, beyond the usual Lagrangian density) that topological brane charges are generally defined; in the above two examples: the charge of magnetic monopoles and of black NS5-branes, respectively.

\paragraph
{Including Electric Charge Quantization}
The above two examples concern only the quantization of \emph{magnetic} charges. More generally one may and should consider quantizing also the \emph{electric} charges. For instance:
\begin{enumerate}
\item 
Duality-symmetric vacuum Maxwell theory admits an IR-completion by joint magnetic and electric flux quantization, cf. \cref{GaussLawOfVacuumEM},
classified by $B^2 \mathbb{Z}_{\mathrm{mg}} \times B^2 \mathbb{Z}_{\mathrm{el}}$
(this is the classical sector \parencites[(1.26)]{FreedMooreSegal2007}, as per \cref{MeaningOfQuantization}, of \parencites[(3.27)]{FreedMooreSegal2007b}, 
also considered in \parencites[Rem. 2.3]{BeckerBeniniSchenkelSzabo2017}{LazaroiuShahbazi2018}{LazaroiuShahbazi2022a}{LazaroiuShahbazi2022b}).
\item 
\label{DualitySymmetricBFieldFluxQuantization}
Analogously, the pure NS sector of 10D IIA SuGra admits IR-completion by electromagnetic flux/charge quantization classified by $B^3 \mathbb{Z}_{\mathrm{mg}} \times B^7 \mathbb{Z}_{\mathrm{el}}$, cf. \cref{GaussLawOfBField,NSBraneChargesDerived}, quantizing not just the magnetic NS5-brane but also the electric NS1-brane (string) charges (cf. \parencites[Ex. 3.8]{Freed2002}[Ex. 2.18]{FreedMooreSegal2007}).

This brings in a quantization law for the flux $H_7 = \star H_3$ in IIA SuGra, which has not received much attention before (exceptions are \parencites{GiotopoulosSati2026}{BaSS26-UnstableK}), but which is the 10D shadow of the famous analogous non-linearity \cref{ElectricGaussLawIn11D}, and which plays a key role below in \cref{KTheoryEmerging}.
\end{enumerate}

Of course, there are further higher gauge fields in IIA SuGra, beyond the NS sector:

\paragraph
{Hypothesis K: RR-Flux quantization in K-theory}
A famous proposal beyond these ordinary examples is what we want to call \emph{Hypothesis K} (following \parencites[\S~4.1]{SS23-Defect}): that the RR-field in 10D IIA SuGra is globally a cocycle in twisted differential (complex topological rank-zero) K-cohomology%
\footnote{%
  By ``K-theory'' here we mean the cohomological version of complex topological K-theory (\cite{Atiyah1967}, cf. \parencites{Hilton1971}[\S~9]{AguilarGitlerPrieto2002}), and we will focus on degree zero (physically: type IIA, cf. \cref{MeasuringDBraneCharge}) and rank zero (physically: without Romans mass), where the classifying space may be denoted $B \mathrm{U} := \bigcup_n B \mathrm{U}(n)$, cf. \parencites[Def. 9.2.8]{AguilarGitlerPrieto2002}: this is \emph{2-connective K-theory} $B \mathrm{U} \simeq \Omega^\infty \Sigma^2 \mathrm{ku}$.
  
  Moreover, by ``twisted K-theory'' we here mean specifically the degree-3 twist of this situation, where the classifying space is the homotopy quotient (Borel construction) of the form $\mathrm{Fred}\bracket({\HilbertSpace})_0 \sslash \mathrm{PU}(\HilbertSpace) \simeq  B \mathrm{U} \sslash B \mathrm{U}(1)$, fibered over $B \mathrm{PU}(\HilbertSpace) \simeq  B^2 \mathrm{U}(1)$ (\parencites{AtiyahSegal2004}, cf. \parencites[Ex. 6.2.5]{SS26-Bun}[\S~5]{SS26-Orb} and \cref{TwistedKTheoryAsNonabelianCohomology} below), where on the right $B \mathrm{U}(1)$ acts on $B \mathrm{U}$ by the universal tensoring of complex vector bundles by complex line bundles \parencites[p. 3]{AntieauGepnerGomez2014}[Thm. 1.5, Prop. 7.2, Thm. 7.9]{HebestreitSagave2020}.
}
(\parencites{MinasianMoore1997}{Witten1998}{FreedWitten1999}{MooreWitten2000}{FreedHopkins2000}{BouwknegtMathai2000}, cf. \parencites[\S~3.7]{Freed2001}{Freed2002}{Szabo2002}{Fredenhagen2008}{DistlerFreedMoore2011}{GS22-KTheory}[\S~4.1]{SS25-Flux})---with relevant classifying space $B \mathrm{U} \sslash B \mathrm{U}(1)$, cf. \cref{NSRRFluxGaussLaw,MeasuringDBraneCharge} below.

Here $B \mathrm{U} := \bigcup_{n \in \mathbb{N}} B \mathrm{U}(n)$ by itself classifies both magnetic and electric RR-charges (cf. \cref{11DAnd10DFluxSpecies,DBraneChargeInOmegaBUCohomology}), while the canonical projection $\inlinetikzcd{ B \mathrm{U} \sslash B \mathrm{U}(1) \ar[r] \& B^3 \mathbb{Z} }$ extracts the magnetic B-field charges as in the above \cref{PlainBFieldFluxQuantization}. In contrast, the \emph{electric} NS1-brane charges of the above  \cref{DualitySymmetricBFieldFluxQuantization} are \emph{not} reflected in twisted K-theory (nor in \emph{any abelian} twisted generalized cohomology theory, cf. \cref{TheH7GaussLaw} below).

This particular proposal of \emph{Hypothesis K}  (one among infinitely many other admissible choices of IR-completing IIA SuGra \cite{GSS26-SuGra}) is of course motivated by phenomena expected in type IIA string theory (cf. \parencites{Polchinski1998-Vol2}{Ortin2015}). But this raises the following question:

\paragraph!
{How does this lift to M-Theory?}
The \emph{M-theory conjecture}  (\parencites{DuffHoweInamiStelle1987}{Townsend1995}[\S~2.2--2.3]{Witten1995}, cf. \parencites{Duff1996-Formerly}{Schwarz1998}{Witten1998}{Duff1999World}) --- that type IIA string theory is the dimensional reduction on a small circle fiber of a completion of 11D SuGra --- suggests that:
\begin{standout}
  The appropriate flux quantization law in 10D IIA should be the dimensional reduction of an M-theoretic flux quantization law in 11D. 
\end{standout}
In other words, once an M-theoretic IR-completion in 11D is agreed upon, the IR-completion in 10D ought to be \emph{implied} and \emph{derived} by a systematic process of dimensional reduction. In the jargon of phenomenologists, the question is:
\begin{standout}
Is there a \emph{top-down} (beyond the historical \emph{bottom-up}) justification of \emph{Hypothesis K}?
\end{standout}

\subsection
{The Observation by DMW}
\label{TheObservationByDMW}

The desire for such a  ``derivation of K-theory from M-theory'' motivated the title of \cite{DMW2000,DMW2003} (in the following ``DMW'', for short). The result presented there is a consistency check that such a derivation might exist, of which we briefly outline the part relevant to us:

\paragraph
{The Subsector}

DMW restrict attention to the sector of trivial twist (vanishing NS5-brane charge, $H_3$-flux) and initially to a trivial M-circle bundle (vanishing D6-brane charge, $F_2$-flux),
leaving open the 11D origin of D2-brane charge and its $F_6$-flux and disregarding NS1-brane charge and its $H_7$-flux.
(Beware that the latter are the electric partners to $F_4$ and $H_3$ and that we are trying to derive in 10D a quantization law for both magnetic and electric charges/fluxes.)

Under these assumptions, the only flux density in 10D obtained from 11D is $F_4$.

\paragraph
{The Conceptual Issue}

This highlights a notorious general issue with lifting \emph{Hypothesis K} to 11D (cf. \cref{11DAnd10DFluxSpecies,ReductionTo9D}): 
\begin{standout}
Naive dimensional reduction does not produce all the required flux species in 10D: notably, $F_8$ has no plain 11D origin.
\end{standout}

Taken at face value, this means: \emph{There is no natural lift through ordinary dim-reduction of Hypothesis K to 11D.}  We resolve this in \cref{MorseBottRegularity} by showing that the missing fluxes appear when imposing a topological smallness condition on the circle reduction.

\begin{table}[htb]
\caption{%
\label{11DAnd10DFluxSpecies}
Under ordinary dimensional reduction, not all the flux species subject to \emph{Hypothesis K} in 10D have lifts to 11D (cf. also \cref{ReductionTo9D}). 
In particular, the electric $F_8$-flux of non-massive type IIA does not.
(Here ``$[S^1_{\mathrm{M}}]$'' is for M-circle \emph{fibration}: $F_2$ in 10D is the Chern form of the M-circle bundle.) 
For the black brane species sourcing these fluxes cf. \cref{DBraneAndSourcedFluxes}.
}
\centering
\adjustbox{rndfbox=4pt}{
\def\arraystretch{1.2}
\begin{tabular}{c||ccc|c|c|ccc}
  \textbf{Fluxes on $X^{D-1}$}
  &
  \multicolumn{4}{c|}
    {\textbf{magnetic}}
  &
  \multicolumn{4}{c}
    {\textbf{electric}}
  \\
  \hline
  \hline
  11D SuGra
  &
  ---
  &
  $[S^1_{\mathrm{M}}]$
  & 
  $(G_4)_{\mathrm{bas}}$
  &
  $(G_4)_{\mathrm{fib}}$
  &
  $(G_7)_{\mathrm{bas}}$
  &
  $(G_7)_{\mathrm{fib}}$
  &
  ---
  &
  ---
  \\
  \hline
  IIA 10D
  & ---
  & $F_2$
  & $F_4$
  & $H_3$
  & $H_7$
  & $F_6$
  & $F_8$
  & ---
  \\
  \hline
  Massive 10D
  & $F_0$
  & $F_2$
  & $F_4$
  & $H_3$
  & $H_7$
  & $F_6$
  & $F_8$
  & $F_{10}$
  \\
  \hline\hline
  \textbf{Sector}
  & \multicolumn{3}{c|}{\textbf{RR}}
  & \multicolumn{2}{c|}{\textbf{NS}}
  & \multicolumn{3}{c}{\textbf{RR}}
\end{tabular}
}
\end{table}

DMW introduce the missing flux species by \emph{fiat}, as follows:

\paragraph
{The DMW Lift}

DMW consider the equivalence classes $Q$ of field content in 10D to be lifts of the D4-brane charge $f_4 \in H^4\bracket({X^{9};\mathbb{Z}})$ through the Atiyah--Hirzebruch spectral sequence (AHSS, cf. \cite[Thm. 4.2.5]{Kochman1996}) for K-theory.
By the filtration underlying the AHSS, this is the step that tacitly postulates the appearance of the missing flux species $F_8$ (and brings in the previously neglected $F_6$) in the image of the Chern character%
\footnote{%
  The string theory literature traditionally includes a factor $\sqrt{\widehat{A}}$ with the Chern character (\parencites[(2.11)]{GreenHarveyMoore1996}[(5.1)]{CheungYin1997}[(1.1)]{MinasianMoore1997}). This is a historical normalization choice \parencites[ftn. 12]{FreedHopkins2000}[(2.8)]{Freed2002}[(1.8)]{BrodzkiMathaiRosenbergSzabo2008}[ftn. 8]{DistlerFreedMoore2011} and, as such, may even be absorbed \cite[(1.11)]{BrodzkiMathaiRosenbergSzabo2008} into the cohomological pairing that plays no role for our purpose here. Therefore, we need not and will not include this factor. Remarkably, without it we obtain integrality of the Chern character components in \cref{UniversalChernChInDMWSector} of \cref{ComparisonInDMWSector} below (in the DMW sector, on globally hyperbolic 10D spacetimes, under \emph{Hypothesis H}), which may be of more intrinsic physical relevance, cf. \cref{OnIntegralityOfChernChComponents} in \cref{PhysicalImpactOfDMWSectorResult} below.
}
$\inlinetikzcd{ \mathrm{ch} : \mathrm{KU}(-) \ar[r] \& \oplus_{k} H^{2k}({-; \mathbb{R}}) }$ (\parencites[\S~1.10]{AtiyahHirzebruch1961}[Thm. 5.8]{Hilton1971}, cf. \parencites[Ex. 7.2 \& \S~2]{FreedHopkinsTeleman2002}[Prop. 10.1]{FSS23-Char}):
\begin{equation}
\label{TheDMWLift}
  \begin{tikzcd}[
    column sep=-5.5pt,
    row sep=25pt,
    /tikz/column 3/.append style={anchor=base west},
  ]
    Q
    \ar[
      r,
      |->,
      "{
        \substack{
          \text{Chern}
          \\
          \text{character}
        }
      }",
      "{
        \mathrm{ch}
      }"'
    ]
    &[60pt]
    {[F_4} 
    &
    {+ F_6 + F_8]} 
    \\
    f_4
    \ar[
      u,
      |->,
      "{
        \substack{
          \text{ad hoc lift}
          \\
          \text{through AHSS}
          \\
          \text{to 4-connective}
          \\
          \text{K-theory}
        }
      }"{xshift=-2pt}
    ]
    \ar[
      r,
      |->,
      "{
          \text{de Rham}
      }",
      "{
          \text{map}
      }"'
    ]
    &
    {[F_4}
    &
    {]} 
    \,.
  \end{tikzcd}
\end{equation}

\paragraph
{The Remaining Gap}
This is where the DMW account falls short of being a ``derivation'': 
K-theory is assumed at this point \cref{TheDMWLift}, and the matching of the partition functions becomes a consistency check verifying the compatibility condition of this ansatz. It shows that, under the assumption of \emph{Hypothesis K}, the partition function in 10D is compatible with the 11D M-theoretic perspective.

However, beyond that check, the situation has remained unresolved: Even assuming \emph{Hypothesis K}, DMW lift deformation classes to deformation classes without explaining where the actual field configurations come from. 
Moreover, the generalization to the twisted situation with $h_3 \neq 0$ has remained less developed (cf. \parencites[p. 79]{DMW2003}[p. 13]{MathaiSati04-E8}[p. 44]{DFM2007}).

\medskip

Our aim is to fill this gap. Our method (cf. \parencites{SS26-HigherGauge}{GSS26-SuGra}) is to:
\begin{description}
\item[\cref{ProperFluxQuantization}:] properly formulate electromagnetic flux-quantization applicable to SuGras,

\item[\cref{TheCFieldIn11D}:] choose in 11D the initial proper flux quantization law  (\emph{Hypothesis H}),

\item[\cref{MCircleReduction}:] make precise the dimensional reduction of flux-quantization laws,

\end{description}
and then discover twisted K-theory emerging (\cref{Results}).

\section
{Method}

\paragraph
{Homotopy Theory}
\begin{enumerate}
\item
We make free use of basic concepts of \emph{homotopy theory} (cf. \parencites{AguilarGitlerPrieto2002}{Fomenko2016}; for an introduction aimed at physicists cf. \parencites{Schwarz1994}, and for a digest specific to our context cf. \cite[\S~1]{FSS23-Char}), especially in \cref{Results} below. 

\item
We also allude to basics of \emph{dg-algebraic rational} homotopy theory (cf. \parencites{Hess2007}{FelixHalperin2017} for general background, while for exposition and digest specific to our context cf. \parencites{FSS2019}[\S~3.1]{SS25-Flux} and \cite[\S~5]{FSS23-Char}, respectively).

\item
In the proofs (and only there) we appeal profusely and without further ado to power tools of \emph{algebraic topology} (cf. \parencites{Switzer1975}{Whitehead1978}{Kochman1996}{tomDieck2008}), notably to various spectral sequences (cf. \parencites{McCleary2010}).
\end{enumerate}

As usual, where we speak of \emph{spaces} (as in \emph{classifying spaces}, in \cref{ProperFluxQuantization}) the reader may read ``topological spaces'' in any convenient (cartesian closed) category of topological spaces (\parencites{Steenrod1967}, cf. \parencites[\S~2.1.2]{SS26-Bun}), but only their homotopy type (their \emph{shape} $\infty$-groupoid) ever matters. In particular, all commuting diagrams in the following are understood as either resolved or filled by homotopies (such as in \cref{SlicedMappingSpace}). Notably,   
the fiber products appearing in \cref{SlicedMappingSpace,KTheoryEmerging} are \emph{homotopy} fiber products (cf. \cite[Def. 1.15]{FSS23-Char}).

\paragraph
{Mapping Spaces}

Throughout, we make substantial use of understanding cohomology classes as homotopy classes of classifying maps (recalled in \cref{ProperFluxQuantization} and heavily applied in \cref{KTheoryEmerging}).
For a pair of spaces $\mathcal{A}_1$, $\mathcal{A}_2$, we write
\begin{equation}
\label{MappingSpace}
  \mathrm{Map}\bracket({
    \mathcal{A}_1
    ,\,
    \mathcal{A}_2
  })
  \,\defneq\,
  \{
  \begin{tikzcd}[
    column sep=15pt
  ]
    \mathcal{A}_1
    \ar[r, dashed]
    &
    \mathcal{A}_2
  \end{tikzcd}
  \}
\end{equation}
for the space of (continuous) maps between them (cf. \parencites[\S~1.2, 1.3]{AguilarGitlerPrieto2002}[(1.79)]{FSS23-Char}[\S~2.1.2]{SS26-Bun}). We use dashed arrows to highlight maps that remain unspecified. By $\inlinetikzcd{\mathcal{A}_1 \ar[r, "{ \sim }"] \& \mathcal{A}_2}$ we denote maps that are (weak homotopy) equivalences (cf. p. \pageref{Notation}).

More generally, for a pair of fibered spaces, hence equipped with maps  $\inlinetikzcd{ \mathcal{A}_i \ar[r, "{ p_i }"] \& \mathcal{B} }$ to a joint base space $\mathcal{B}$, we write
\begin{equation}
\label
{SlicedMappingSpace}
  \mathrm{Map}\bracket({
    \mathcal{A}_1
    ,\,
    \mathcal{A}_2
  })_{\mathcal{B}}
  :=
  \mathrm{Map}\bracket({
    \mathcal{A}_1
    ,\,
    \mathcal{A}_2
  })
  \underset
    {
      \mathrm{Map}
      \scaledbracket({
        \mathcal{A}_1
        ,\,
        \mathcal{B}
      })
    }
    {\times}
  \{ p_1 \}
  \,\defneq\,
  \left\{
  \begin{tikzcd}[
    column sep=10pt,
    row sep=0pt
  ]
    \mathcal{A}_1
    \ar[
      dr,
      "{ p_1 }"{swap,pos=.35}
    ]
    \ar[
      rr,
      dashed
    ]
    &&
    \mathcal{A}_2
    \ar[
      dl,
      "{ p_2 }"{pos=.45}
    ]
    \\
    &
    \mathcal{B}
  \end{tikzcd}
  \right\}
\end{equation}
for their \emph{slice} or \emph{relative} mapping space (cf. \parencites[\S~4.1.8]{SS26-Orb}[Def. 4.2.66]{SS26-Bun}).

\subsection
{Flux Quantization}
\label
{ProperFluxQuantization}

A key observation (\parencites{SS25-Flux}{SS26-HigherGauge}, cf. \cref{AbelianCohDoesNotQuantizeNonlinear,NonLinearGaussMeansNonLinearCoh}): one needs to take flux quantization laws beyond \emph{abelian} differential generalized cohomology (\parencites{HopkinsSinger2005}, cf. \parencites{Freed2002}{Bunke2012}) such as (differential) K-theory, to (differential) \emph{nonabelian} cohomology theories (cf. \cite[\S~2]{FSS23-Char}) in order to describe electromagnetic flux quantization with non-linear Bianchi/Gauss laws, such as those that appear in 11D SuGra \cref{S4ValuedGaussLaws}.

\begin{SCtable}[1.2][htb]
\caption{%
\label{NonLinearGaussMeansNonLinearCoh}
  The Gauss laws on fluxes quantized by a cohomology theory $H^1({-; \Omega \mathcal{A}})$ \cref{NonabelianCohomology} are the differential relations in the \emph{minimal Sullivan model} of $\mathcal{A}$ \parencites[\S~3.2]{SS25-Flux}[\S~4.1.1]{SS26-HigherGauge}, see \cref{AdmissibleFluxQuantizationLaw}. When these Gauss laws are non-linear (like \cref{TheH7GaussLaw,ElectricGaussLawIn11D} and many in \cref{ReductionTo9D}), then this \emph{cannot} (by \cref{AbelianCohDoesNotQuantizeNonlinear}) be an abelian generalized cohomology theory \cref{WhiteheadGeneralizedCohomology} like K-theory. 
}
\adjustbox{rndfbox=4pt}{
\begin{tabular}{c|c}
  \textbf{Fluxes} & \textbf{Charges}
  \\
  \multicolumn{2}{c}
    {\text{subject to}\phantom{--------------}}
  \\
  \textbf{Gauss laws} 
    & 
  \textbf{cohomology theory}
  \\
  \hline\hline
  linear  &  possibly abelian
  \\
  non-linear & necessarily non-abelian
\end{tabular}
}
\end{SCtable}

\paragraph
{Nonabelian Cohomology: Brane Charges}
 Abelian \textbf{Whitehead generalized cohomology} theories $E$  (\parencites{KonoTamaki2006}, cf. \parencites[\S~3.4]{Kochman1996}[Ex. 2.10]{FSS23-Char}) are those whose pointed classifying spaces $E_n$, with
 \begin{equation}
 \label
 {WhiteheadGeneralizedCohomology}
  E^n\bracket({
    X
  })
  \,\simeq\,
  \pi_0\, \mathrm{Map}\bracket({
    X, 
    E_n
  })
  \mathrlap{\,,}
 \end{equation}
 form a \emph{spectrum of spaces}, $\inlinetikzcd{ E_n \ar[r, "{ \sigma_n }", "{\sim}"'] \& \Omega E_{n+1} }$, making them the homotopy-theoretic analogues of abelian groups. This class includes:
 \begin{enumerate}
 \item 
 \textbf{Ordinary cohomology} $H^n(-;A)$ 
 
 with $(HA)_n \defneq B^n A$.
 
 \item
 \textbf{Complex topological K-theory} $\mathrm{KU}^n$, 
 
 with $\mathrm{KU}_{2k} \simeq B \mathrm{U} \times \mathbb{Z}$ and $\mathrm{KU}_{2k+1} \simeq \mathrm{U}$.
 
 \item The \textbf{2-connective complex K-theory} $\Sigma^2 \mathrm{ku}$ with which we will be mostly concerned, 
 
 where $\bracket({ \Sigma^2\mathrm{ku} })_0 \simeq B \mathrm{U}$ and $\bracket({\Sigma^2 \mathrm{ku}})_1 \simeq \mathrm{SU}$.
\end{enumerate}

More generally, for \emph{any} (say connected, just for brevity) space $\mathcal{A}$, it still makes sense to regard the connected components of the mapping space \cref{MappingSpace} into it,
\begin{equation}
\label{NonabelianCohomology}
  H^1\bracket({
    X; \Omega\mathcal{A}
  })
  \,\defneq\,
  \pi_0\, \mathrm{Map}\bracket({  
    X, \mathcal{A}
  })
  \mathrlap{\,,}
\end{equation}
as the \textbf{nonabelian cohomology} of $X$ (\parencites[Def. 6.0.6]{Toen2002}[Def. 2.3]{Schreiber2009OWR}[Def. 6]{Lurie2014}[\S~2]{FSS23-Char}[\S~1]{SS25-TEC}; under suitable conditions this is the \emph{non-abelian Poincar{\'e} dual} of \emph{factorization homology} \parencites{Lurie2014}[\S~4]{AyalaFrancis2015})
with coefficients in the loop-$\infty$-group $\Omega \mathcal{A}$.
This general notion subsumes all abelian cohomology theories; for instance: 
\begin{enumerate}
\item 
Ordinary \emph{integral cohomology} is recovered as (cf. \parencites[\S~7.1]{AguilarGitlerPrieto2002}[Ex. 2.1]{FSS23-Char}):
\begin{equation}
  H^{n+1}\bracket({
    X
    ;\,
    \mathbb{Z}
  })
  \simeq
  H^1\bracket({
    X
    ;\,
    B^n \mathbb{Z}
  })
  \defneq
  \{
  \inlinetikzcd{
    X
    \ar[r, dashed]
    \&
    B^{n+1} \mathbb{Z}
  }
  \}_{/\mathrm{hmtp}}\,.
\end{equation}
\item 
Rank-zero \emph{K-theory} is recovered as (\cite[\S~1.3]{AtiyahHirzebruch1961}, cf. \parencites[Ex. 2.11]{FSS23-Char})
\begin{equation}
  \mathrm{KU}^0_{\mathrm{rk=0}}(X)
  \simeq
  H^1\bracket({
    X
    ;\,
    \Omega
    \,
    B \mathrm{U}
  })
  \defneq
  \{
  \inlinetikzcd{
    X
    \ar[r, dashed]
    \&
    B \mathrm{U}
  }
  \}_{/\mathrm{hmtp}}
  \mathrlap{\,.}
\end{equation}
\end{enumerate}
Beyond these, there are genuinely nonabelian cohomology theories, the most fundamental one being:
\begin{enumerate}[resume]
  \item
  \emph{Cohomotopy}, the cohomology theory whose classifying space is a sphere (\parencites{Borsuk1936}{Pontrjagin1938}{Spanier1949}, cf. \parencites[\S~VII]{STHu59}[Ex. 2.7]{FSS23-Char}):
  \begin{equation}
  \label{Cohomotopy}
    \pi^n(X)
    \,\defneq\,
    H^1\bracket({
      X
      ;\,
      \Omega S^n
    })
    \defneq
    \{
    \inlinetikzcd{
      X
      \ar[r, dashed]
      \&
      S^n
    }
    \}_{/\mathrm{hmtp}}
    \mathrlap{\,.}
  \end{equation}
  This is the unstable refinement of \emph{stable Cohomotopy}  (cf. \parencites[p. 245]{Adams1974}), whose classifying spaces are the stages of the sphere spectrum $\mathbb{S}$, the stabilized spheres:
  \begin{equation}
  \label{StableCohomotopy}
   \mathbb{S}^n(X)
   \simeq
   H^1\bracket({
    X
    ;\,
    \Omega^{\infty+1} \Sigma^\infty S^n
   })
   \defneq
    \{
    \inlinetikzcd{
      X
      \ar[r, dashed]
      \&
      \Omega^{\infty} \Sigma^\infty S^n
    }
    \}_{/\mathrm{hmtp}}
    \mathrlap{\,.}
  \end{equation}
\end{enumerate}

Below we make crucial use of the fact that \emph{nonabelian cohomology operations} (\cite[Def. 2.3]{FSS23-Char}, generalizing familiar cohomology operations between abelian Whitehead-generalized cohomology theories, cf. \parencites{MosherTangora2008}{Boardman1995})---namely natural transformations between nonabelian cohomology theories \cref{NonabelianCohomology},
\begin{equation}
\label
{NonabelianCohomologyOperation}
  \begin{tikzcd}
  H^1\bracket({
    -;
    \Omega
    \, 
    \mathcal{A}_1
  })
  \ar[r, "{ \Phi_\ast }"]
  &
  H^1\bracket({
    -;
    \Omega
    \, 
    \mathcal{A}_2
  })
  \mathrlap{\,,}
  \end{tikzcd}
\end{equation}
---are in natural correspondence with homotopy classes of maps between their classifying spaces,
\begin{equation}
  \label{MapsRepresentingCohomologyOperations}
  \begin{tikzcd}
    \mathcal{A}_1
    \ar[
      r,
      "{ \Phi }"
    ]
    &
    \mathcal{A}_2
    \mathrlap{\,,}
  \end{tikzcd}
\end{equation}
and hence with elements of the target cohomology theory applied to the source classifying space:
\begin{equation}
  [\Phi]
  \in
  H^1\bracket({
    \mathcal{A}_1
    ;\,
    \Omega\, \mathcal{A}_2
  })
  \mathrlap{\,.}
\end{equation}

\paragraph
{Twisted Nonabelian Cohomology: Background Charges}
More generally, consider a classifying space sitting in a homotopy fiber sequence (of connected spaces, for brevity)
\begin{equation}
  \begin{tikzcd}[row sep=small, column sep=large]
    \mathcal{A}
    \ar[d]
    \ar[r]
    \ar[
      dr,
      phantom,
      "{
        \lrcorner
      }"{pos=.1}
    ]
    &
    \mathcal{P}
    \ar[
      d,
      "{ \wp }"
    ]
    \\
    \ast 
    \ar[r]
    &
    \mathcal{B}
    \mathrlap{\,.}
  \end{tikzcd}
\end{equation}
Then for domain spaces $\mathcal{X}$ equipped with a  map $\inlinetikzcd{ \mathcal{X} \ar[r, "{ \tau }"] \& \mathcal{B}}$ to be called a \emph{twist}, the connected components of the slice mapping space \cref{SlicedMappingSpace} form the  \emph{$\tau$-twisted nonabelian cohomology} \parencites[\S~3]{FSS23-Char}[\S~4.1.8]{SS26-Orb}:
\begin{equation}
\label
{TwistedNonabelianCohomology}
  H^1_\tau\bracket({
    \mathcal{X}
    ;\,
    \Omega \mathcal{A}
  })
  :=
  \pi_0
  \,
  \mathrm{Map}\bracket({
    \mathcal{X}
    ,\,
    \mathcal{P}
  })_{\mathcal{B}}
  \defneq
  \left\{\;
  \begin{tikzcd}[
    row sep=20pt, 
    column sep=28pt
  ]
    &[-10pt] 
    \tau^\ast \mathcal{P}
    \ar[r]
    \ar[d]
    \ar[
      dr,
      phantom,
      "{ \lrcorner }"{pos=.1}
    ]
    &
    \mathcal{P}
    \ar[
      d,
      "{ \wp }"
    ]
    \\
    \mathcal{X}
    \ar[
      r,
      equals
    ]
    \ar[
      ur,
      dashed
    ]
    &
    \mathcal{X}
    \ar[
      r,
      "{ \tau }"
    ]
    &
    \mathcal{B}
  \end{tikzcd}
  \right\}
  \mathrlap{\,.}
\end{equation}
Again, this general notion subsumes all twisted abelian cohomology; for instance:
\begin{itemize}
  \item
  twisted K-theory at rank zero, for given twist $\big[{ \inlinetikzcd{ X \ar[r, "{ h_3 }"] \& B^3 \mathbb{Z} } }\big]$, is recovered as (cf. \cite[Prop. 3.5]{FSS23-Char}):
  \begin{equation}
  \label
  {TwistedKTheoryAsNonabelianCohomology}
    \mathrm{KU}^{h_3}_{\mathrm{rk}=0}
      (X)
    \simeq
    H^1_{h_3}\bracket({
      X
      ;\,
      \Omega B \mathrm{U}
    })
    \defneq
    \left\{
    \begin{tikzcd}[
      row sep=20pt,
      column sep=23pt
    ]
      &
      &[-15pt]
      B \mathrm{U}
      \sslash 
      B \mathrm{U}(1)
      \ar[
        d
      ]
      \\
      X
      \ar[
        urr,
        dashed,
        bend left=20,
        start anchor={[xshift=-5pt]}
      ]
      \ar[r, "{ h_3 }"]
      &
      B^3 \mathbb{Z}
      \ar[r, "{ \sim }"]
      &
      B^2 \mathrm{U}(1)
    \end{tikzcd}
    \right\}
    \mathrlap{\,.}
  \end{equation}
\end{itemize}

In an evident generalization of \cref{NonabelianCohomologyOperation}, cohomology operations on twisted nonabelian cohomology
\begin{equation}
\label
{TwistedCohomologyOperation}
  \begin{tikzcd}
    H^1_{\tau}\bracket({
      -
      ;\,
      \Omega \mathcal{A}_1
    })
    \ar[
      r,
      "{ \Phi_\ast }"
    ]
    &
    H^1_{\Phi_\ast\tau}\bracket({
      -
      ;\,
      \Omega \mathcal{A}_2
    })
  \end{tikzcd}
\end{equation}
are given by postcomposition with slice maps between classifying fibrations
\begin{equation}
  \begin{tikzcd}[
    sep=15pt
  ]
    \mathcal{A}_1
    \ar[d]
    \ar[
      rr,
      "{
        \Phi\vert_{\mathcal{A}_1}
      }"
    ]
    &&
    \mathcal{A}_2
    \ar[d]
    \\
    \mathcal{P}_1
    \ar[
      rr,
      "{ \Phi }"
    ]
    \ar[
      dr,
      "{ \wp_1 }"{swap,pos=.35}
    ]
    &&
    \mathcal{P}_2
    \ar[
      dl,
      "{ \wp_2 }"
    ]
    \\
    &
    \mathcal{B}
  \end{tikzcd}
\end{equation}
and hence by elements of the target twisted cohomology theory applied to the domain fibration:
\begin{equation}
  [\Phi]
  \in
  H^1_{\wp_1}\bracket({
    \mathcal{P}_1
    ;\,
    \Omega
    \mathcal{A}_2
  })_{\mathcal{B}}
  \mathrlap{\,.}
\end{equation}

\paragraph
{Nonabelian Character: Higher Gauss Laws}
Now, the rational (meaning: non-torsion) part of a nonabelian cohomology theory $\mathcal{A}$ \cref{NonabelianCohomology}%
\footnote{%
  Technically, we are assuming throughout that our classifying spaces $\mathcal{A}$, $\mathcal{B}$ are simply connected of finite rational type, so that the fundamental theorem of dg-algebraic rational homotopy theory applies (\cite{BousfieldGugenheim1976}, cf. \parencites{FHT2000}{Hess2007}[\S~5]{FSS23-Char}). The assumption of rational finite type is crucial for the relation to higher Gauss laws, while connectedness is assumed just for brevity, and nontrivial fundamental groups may also be allowed, though at a considerable cost of technicalities ($\pi_1$-fiberwise rationalization) when the group is not nilpotent and does not act nilpotently on the higher homotopy groups.
}
is encoded in an $L_\infty$-algebra $\mathfrak{l}\mathcal{A}$ of rational higher \emph{Whitehead brackets} (the Koszul dual to the \emph{minimal Sullivan model} of $\mathcal{A}$, cf. \cite[Prop. 5.11]{FSS23-Char}). 

We say (with \parencites[Rem. 2.4]{SS24-Phase}{SS25-Flux}) that $\mathcal{A}$ is an \emph{admissible} flux-quantization law for a given higher gauge theory if the solutions to the latter's higher Gauss laws on a Cauchy surface $X^d$ are equivalently the closed (meaning: flat, Maurer-Cartan) $\mathfrak{l}\mathcal{A}$-valued differential forms (cf. \parencites[Def. 6.1]{FSS23-Char}[\S~2.1]{SS25-Flux}[(329)]{GSS26-SuGra}):
\begin{equation}
\label
{AdmissibleFluxQuantizationLaw}
\adjustbox{rndfbox=4pt}{
  \text{$\mathcal{A}$ is admissible flux quantization}
  \;\;\;
  $\Leftrightarrow
  \;\;\;
  \left\{
  \substack{
    \text{fluxes on $X^d$ satisfying}
    \\
    \text{their higher Gauss laws}
  }
  \right\}
  \simeq
  \Omega^1_{\mathrm{cl}}\bracket({
    X^d; 
    \mathfrak{l}
    \mathcal{A}
  })
  \;
  $
  }
\end{equation}

\begin{remark}
\label[remark]{CEGenerators}
On the right of \cref{AdmissibleFluxQuantizationLaw}, the Chevalley--Eilenberg algebra of $\mathfrak{l}\mathcal{A}$ is generated as a graded-commutative algebra by abstract graded generators $\vec F \defneq \bracket\{{ F^{(i)} }\}_{i \in I}$ subject to differential relations
\begin{equation}
  \mathllap{
  \forall_{i \in I}
  \;\;\;
  }
  \mathrm{d} F^{(i)}
  =
  P^{(i)}
  \defneq
  P^{(i)}\bracket({ \vec F })
  \,,
\end{equation}
where the $P^{(i)}$ are graded-symmetric polynomials in the generators; and an element in $\Omega^1_{\mathrm{cl}}\bracket({X^d; \mathfrak{l}\mathcal{A} })$ is an image of this situation in the de Rham complex of $X$, hence a dg-algebra homomorphism $\inlinetikzcd{\mathrm{CE}\bracket({ \mathfrak{l}\mathcal{A} }) \ar[r] \& \Omega^\bullet_{\mathrm{dR}}(X) }$. We will use the same symbols for these differential forms as for the abstract algebra generators that they are images of, as the latter may be regarded as the \emph{universal} such forms, cf. \cref{D0TadpoleGeneratorAndCousins} below.  
\end{remark}

For example (cf. \parencites[p. 4 \& \S~3.1--2]{SS25-Flux}): 
\begin{enumerate}
\item 
The magnetic and electric Gauss laws of vacuum Maxwell theory are equivalently the closure condition for differential forms with coefficients in $B^2 \mathbb{Z} \times B^2 \mathbb{Z}$:
\begin{equation}
\label{GaussLawOfVacuumEM}
  \Omega^1_{\mathrm{cl}}\bracket({
    X^3;
    \mathfrak{l}\bracket({
      B^2 \mathbb{Z}
      \times
      B^2 \mathbb{Z}
    })
  })
  \;\simeq\;
  \left\{
  \begin{aligned}
    E_2 & 
    \in \Omega^2_{\mathrm{dR}}
    \bracket({ X^3 })
    \\
    B_2 & 
    \in \Omega^2_{\mathrm{dR}}
    \bracket({ X^3 })
  \end{aligned}
  \,\middle\vert\,
  \begin{aligned}
    \mathrm{d}\, E_2 & = 0
    \\
    \mathrm{d}\, B_2 & = 0
  \end{aligned}
  \right\}
  \mathrlap{.}
\end{equation}

\item 
The magnetic and electric Gauss laws of the NS-flux sector in IIA are equivalently the closure condition for differential forms with coefficients in $B^3 \mathbb{Z} \times B^7 \mathbb{Z}$:
\begin{equation}
\label{GaussLawOfBField}
  \Omega^1_{\mathrm{cl}}\bracket({
    X^9;
    \mathfrak{l}\bracket({
      B^3 \mathbb{Z}
      \times
      B^7 \mathbb{Z}
    })
  })
  \;\simeq\;
  \left\{
  \begin{aligned}
    H_3 & 
    \in \Omega^3_{\mathrm{dR}}
    \bracket({ X^9 })
    \\
    H_7 & 
    \in \Omega^7_{\mathrm{dR}}
    \bracket({ X^9 })
  \end{aligned}
  \,\middle\vert\,
  \begin{aligned}
    \mathrm{d}\, H_3 & = 0
    \\
    \mathrm{d}\, H_7 & = 0
  \end{aligned}
  \right\}
  \mathrlap{.}
\end{equation}

\item 
The Gauss laws of type IIA, \emph{excluding} $H_7$, are equivalently the closure conditions for differential forms with coefficients in the classifying space $\mathcal{A} \,\defneq\,  B \mathrm{U} \sslash B \mathrm{U}(1)$ for rank-0 twisted K-theory \cref{TwistedKTheoryAsNonabelianCohomology}:
\begin{equation}
\label{NSRRFluxGaussLaw}
  \Omega^1_{\mathrm{cl}}\bracket({
    X^{9}
    ;
    \mathfrak{l}
    \bracket({
      B \mathrm{U}
      \sslash
      B \mathrm{U}(1)
    })
  })
  \;\simeq\;
  \left\{
  \begin{aligned}
    H_3 & \in 
    \Omega^3_{\mathrm{dR}}
    \bracket({ X^9 })
    \\
    F_2 & \in 
    \Omega^2_{\mathrm{dR}}
    \bracket({ X^9 })
    \\
    F_4 & \in 
    \Omega^4_{\mathrm{dR}}
    \bracket({ X^9 })
    \\
    F_6 & \in 
    \Omega^6_{\mathrm{dR}}
    \bracket({ X^9 })
    \\
    F_8 & \in 
    \Omega^8_{\mathrm{dR}}
    \bracket({ X^9 })
  \end{aligned}
  \;\middle\vert\;
  \begin{aligned}
    \mathrm{d}\, H_3 & = 0
    \mathclap{\phantom{\bracket({ X^9 })}}
    \\
    \mathrm{d}\, F_2 & = 0
    \mathclap{\phantom{\bracket({ X^9 })}}
    \\
    \mathrm{d}\, F_4 & = H_3 \wedge F_2
    \mathclap{\phantom{\bracket({ X^9 })}}
    \\
    \mathrm{d}\, F_6 & = H_3 \wedge F_4
    \mathclap{\phantom{\bracket({ X^9 })}}
    \\
    \mathrm{d}\, F_8 & = H_3 \wedge F_6
    \mathclap{\phantom{\bracket({ X^9 })}}
  \end{aligned}
  \right\}
  \mathrlap{.}
\end{equation}
\end{enumerate}
From this \cref{NSRRFluxGaussLaw}, one sees at once that twisted K-theory:
\begin{enumerate}
\item 
is an admissible flux quantization law \cref{AdmissibleFluxQuantizationLaw} for IIA SuGra including electric $F_8$-flux, 
\item 
\emph{if} one ignores the electric $H_7$-flux and thus fails to quantize NS1-brane charge.
\end{enumerate}

Remarkably, the dominant \emph{Hypothesis K} about flux quantization in string theory fails to quantize the string charge.
This blind spot is arguably the result of a previous lack of tools,
because the $H_7$-Gauss law is non-linear (cf. \cite[\S~22.1.3]{Ortin2015}),
\begin{equation}
\label{TheH7GaussLaw}
  \mathrm{d}\,
  H_7 \,=\,
  \tfrac{1}{2} F_4 \wedge F_4
  -
  F_2 \wedge F_6
  \mathrlap{\,,}
\end{equation}
and \emph{no abelian twisted generalized cohomology theory is compatible with non-linear Gauss laws} (cf. \cref{NonLinearGaussMeansNonLinearCoh}):

\begin{proposition}[Abelian cohomology cannot quantize nonlinear Gauss laws]
\label[proposition]
{AbelianCohDoesNotQuantizeNonlinear}
\begin{enumerate}
\item
\label
{GaussLawsOfTwistedAbelianGeneralizedCoh}
 The Gauss laws induced, via \cref{AdmissibleFluxQuantizationLaw}, by any abelian twisted generalized cohomology theory are \textup{(up to field redefinition preserving the given twist)} linear in the twisted flux densities relative to the twisting flux densities.
 
 \textup{(For instance $\mathrm{d}F_{4} = H_3 \wedge F_2$ is linear in the $F_{2\bullet}$ relative to the twisting coefficient $H_3$.)}

\item
\label
{GaussLawsOfAbelianGeneralizedCoh}
For vanishing twist, an even stronger statement holds: As soon as the classifying space is a loop space $\mathcal{A} \defneq \Omega \mathcal{B}$ \textup{(let alone an infinite loop space classifying abelian generalized cohomology \cref{WhiteheadGeneralizedCohomology})}, the induced Gauss laws are all closure conditions, of the form \mbox{$\mathrm{d}F \!=\! 0$}.
\end{enumerate}
\end{proposition}
\begin{proof}
\Cref{GaussLawsOfAbelianGeneralizedCoh} is a simple exercise in minimal Sullivan models, cf. \cite[Ex. 5.6]{FSS23-Char}.
For the general case, rational parameterized stable homotopy theory shows that the rational minimal models of underlying total spaces $\Omega^\infty_{\mathcal{B}}\bracket({ E \sslash \Omega \mathcal{B} })$
of parameterized spectra over a base $\mathcal{B}$ are gc-algebras of the form $\mathrm{Sym}_A(N)$, given by dg-modules $N$ over the minimal dg-algebra model $A$ of the base, with differential $\mathrm{d} N \subset A \otimes N$ (\cite[Thm. 2.19, Prop. 2.30]{BMSS2019}, based on \parencites[Thm. 2.7.42]{Braunack-Mayer2018}[Thm. 4.20 \& \S~5.1]{Braunack-Mayer2020II}). 
\end{proof}

This \emph{no-go theorem} (\cref{AbelianCohDoesNotQuantizeNonlinear}) is a substantial issue for attempts to derive IR-completions via electromagnetic flux/charge quantization laws under dimensional reduction, because non-abelian Gauss laws proliferate under dimensional reduction, cf. \cref{ReductionTo9D}. But if we embrace flux/charge quantization in nonabelian cohomology, then the theory runs smoothly:

\paragraph
{Differential Nonabelian Cohomology: IR-Complete Higher Gauge Fields}

Without recalling here the technical details needed (\emph{cohesive higher topos theory} \parencites[\S~3.1]{SSS12}{Sc13-dcct}, cf. \parencites[\S~1]{FSS23-Char}[\S~9.1]{SS26-Orb}[\S~2]{SS26-HigherGauge}) to make the following statement precise, the upshot of choosing an admissible classifying space $\mathcal{A}$ \cref{AdmissibleFluxQuantizationLaw} is that it completely characterizes IR-completed higher gauge fields that may be sourced by brane charges in $\Omega \mathcal{A}$-cohomology. Namely (\parencites[Def. 9.3]{FSS23-Char}, cf. \parencites[\S~3.3]{SS25-Flux}[\S~4.2]{SS26-HigherGauge}):
\begin{standout}
After flux quantization in $\mathcal{A}$, the IR-complete higher gauge fields on the Cauchy surface $X^d$ are the cocycles in \emph{differential} nonabelian $\Omega\mathcal{A}$-cohomology.
\end{standout}

For instance:
\begin{enumerate}
  \item
  For $\mathcal{A} \defneq B^2 \mathbb{Z}$ we obtain differential integral cohomology in degree 2 \cite[Ex. 9.4]{FSS23-Char}, whose cocycles may be identified with connections on $\mathrm{U}(1)$-principal bundles as familiar for the Maxwell field.

  \item 
  For $\mathcal{A} \defneq B^3 \mathbb{Z}$ we obtain differential integral cohomology in degree 3 \cite[Ex. 9.4]{FSS23-Char}, whose cocycles may be identified with connections/curving on $\mathrm{U}(1)$-bundle gerbes, as familiar for the NS B-field.

  \item 
  For $\mathcal{A} \defneq B \mathrm{U} \sslash B \mathrm{U}(1)$ we obtain differential twisted K-theory \cite[Ex. 9.2, Ex. 11.3]{FSS23-Char},  whose cocycles are the IR-completed RR-fields of IIA under \emph{Hypothesis K}.

   Forgetting the differential refinement, the underlying charges are in plain (non-differential) rank-zero twisted K-theory (\cite{AtiyahSegal2004}, cf. \parencites[Ex. 3.4]{FSS23-Char}), which in our ambient nonabelian notation \cref{TwistedNonabelianCohomology} looks like this:
\begin{equation}
\label
{TwistedKTheoryAsTwistedNonabe}
  \mathrm{KU}^{\tau}_{\mathrm{rk=0}}
  \bracket({
    X^9
  })
  \simeq
  H^1_{\tau}\bracket({
    X^9
    ;
    \Omega\,
    B \mathrm{U}
  })
  \defneq 
  \left\{
  \begin{tikzcd}[
    row sep=15pt, 
    column sep=40pt
  ]
    & 
    B \mathrm{U}
    \sslash
    B \mathrm{U}(1)
    \ar[
      d,
      ->>
    ]
    \\
    X^9
    \ar[
      r,
      "{
        \substack{
          \text{NS-charge}
        }
      }"'
    ]
    \ar[
      ur,
      dashed,
      "{
        \substack{
          \text{RR-charge\;\;}
        }
      }"{sloped}
    ]
    &
    B^2 \mathrm{U}(1)
    \ar[
      r,
      phantom,
      "{ \simeq }"
    ]
    &[-40pt]
    B^3 \mathbb{Z}
  \end{tikzcd}
  \right\}_{\!\!\big/\mathrm{hmtp}\mathrlap{\,.}}
\end{equation}

  \item
  Generally, for $\mathcal{A} \defneq E_n$ a stage in a spectrum of spaces, we obtain \cite[Ex. 9.1]{FSS23-Char} differential $E$-cohomology going back to \cite{HopkinsSinger2005}, here in the \emph{canonical} form of \cite[Def. 4.46]{Bunke2012}.

  \item 
  For $\mathcal{A} \defneq S^4$ we obtain differential Cohomotopy \cite[Ex. 9.3]{FSS23-Char}, whose cocycles are the IR-complete C-field configurations of 11D SuGra under \emph{Hypothesis H}. 

  This is what we now turn to in \cref{TheCFieldIn11D}.
\end{enumerate}

\subsection
{Model of the C-Field}
\label
{TheCFieldIn11D}

The non-linear Gauss law \cref{TheH7GaussLaw} is of course the dimensional reduction (cf. \cite[\S~4.2]{MathaiSati04-E8}) of the more widely appreciated non-linear Gauss law of the electric C-field flux (cf. \parencites[(3.23)]{MiemiecSchnakenburg2006}[(147)]{GSS24-SuGra}):
\begin{equation}
\label{ElectricGaussLawIn11D}
  \mathrm{d}\, G_7 
  \,=\,
  \tfrac{1}{2} G_4 \wedge G_4
  \mathrlap{\,.}
\end{equation}

While this electric C-field Gauss law, too, 
is impossible to quantize in an abelian generalized cohomology theory (cf. again \cref{AbelianCohDoesNotQuantizeNonlinear}), it is in fact quantizable in the most fundamental example of a nonabelian generalized cohomology theory: \emph{Cohomotopy} \cref{Cohomotopy}, represented by a sphere, 
\begin{equation}
  \pi^n(X)
  \,:=\,
  H^1\bracket({
    X;
    \Omega S^n
  })
  \defneq
  \pi_0
  \mathrm{Map}\bracket({
    X,
    S^n
  })
  \mathrlap{\,,}
\end{equation}
here specifically by the 4-sphere (cf. \parencites{FSS17-Sphere}{FSS2019}):
\begin{equation}
\label
{S4ValuedGaussLaws}
  \Omega^1_{\mathrm{cl}}\bracket({
    X^{10};
    \mathfrak{l}S^4
  })
  \;\simeq\;
  \left\{
  \begin{aligned}
    G_4 & \in
    \Omega^4_{\mathrm{dR}}\bracket({X^{10}})
    \\
    G_7 & \in
    \Omega^7_{\mathrm{dR}}\bracket({X^{10}})
  \end{aligned}
  \,\middle\vert\,
  \begin{aligned}
  \mathrm{d}\, G_4 & = 0
  \\
  \mathrm{d}\, G_7 & = 
  \tfrac{1}{2} G_4 \wedge G_4
  \end{aligned}
  \right\}
  \mathrlap{.}
\end{equation}

There are (always infinitely many) other admissible electromagnetic flux quantization laws $\mathcal{A}$ \cref{AdmissibleFluxQuantizationLaw} for 11D SuGra \cref{ElectricGaussLawIn11D}; see for instance \cite{BaSS26-UnstableK}. But the most fundamental choice (equipped with comparison operations into any other choice) is 4-Cohomotopy $\mathcal{A} \defneq S^4$ \cref{S4ValuedGaussLaws}, originally considered in \parencites[\S~2.5]{Sati2018}[\S~4]{FSS15-M5WZW}. 

The hypothesis that \emph{cohomotopical} flux quantization gives the correct IR-completion of 11D SuGra in view of M-theory (or rather its tangentially twisted version, if the first higher curvature correction is included, cf. \cite{Tsimpis2004-Ell3}), has been called \emph{Hypothesis H} in \parencites{FSS20-H}{FSS21-Hopf}{SS23-Mf}[\S~12]{FSS23-Char}.

This entails that the C-field is globally a cocycle in (tangentially twisted) \emph{differential 4-Cohomotopy} (cf. \parencites[\S~4]{FSS15-M5WZW}[\S~3.1]{GS21}[Ex. 9.3]{FSS23-Char}[Ex. 4.7]{GSS26-SuGra}), and that the M-brane charges, the sources of this field, are in (tangentially twisted) 4-Cohomotopy $\pi^4(-) \defneq H^1\bracket({-; \Omega S^4})$ \cref{Cohomotopy}. As a quick plausibility check, this immediately implies (cf. \cref{MeasuringDBraneCharge}):
\begin{enumerate}
  \item that magnetic M5-brane charge, as measured near the horizon, is in:
  \begin{equation}
  \label
  {M5ChargeAccordingToHypothesisH}
    H^1\bracket({
      \mathbb{R}^{10}
      \setminus 
      \mathbb{R}^5
      ;
      \Omega S^4
    })
    \simeq
    \pi_0 \mathrm{Map}\bracket({
      \mathbb{R}^5
      \times 
      \mathbb{R}_{> 0}
      \times
      S^4
      ,\,
      S^4
    })    
    \simeq
    \pi_4\bracket({
      S^4
    })
    \simeq
    \mathbb{Z}
    \mathrlap{\,,}
  \end{equation}
  as expected,

  \item but also that electric M2-brane charge (``Page charge'', \parencites[(8)]{Page1983}[(43)]{DuffStelle1990}, cf. \parencites{Moore2004}[p. 21]{BaggerLambertMukhiPapageorgakis2013}) is in:
  \begin{equation}
  \label{M2ChargeAccordingToHypothesisH}
    H^1\bracket({
      \mathbb{R}^{10}
      \setminus 
      \mathbb{R}^2
      ;
      \Omega S^4
    })
    \simeq
    \pi_0 \mathrm{Map}\bracket({
      \mathbb{R}^2
      \times 
      \mathbb{R}_{> 0}
      \times
      S^7
      ,\,
      S^4
    })    
    \simeq
    \pi_7\bracket({
      S^4
    })
    \simeq
    \mathbb{Z}
    \oplus
    \mathbb{Z}_{/12}
    \mathrlap{\,.}
  \end{equation}
\end{enumerate}
A list of further consistency checks against topological effects expected in M-theory has been made in \parencites{FSS20-H}{FSS21-Hopf}{SS21-M5Anomaly}, cf. also \cite[\S~12]{FSS23-Char}.

\begin{remark}[Torsion brane charges]
\label[remark]{TorsionMBraneCharges}
  The second summand on the right of \cref{M2ChargeAccordingToHypothesisH} is a first prediction of \emph{Hypothesis H}: that besides the familiar integrally charged flavor of M2-branes, there may also be \emph{torsion}%
  \footnote{%
    Here \emph{torsion} refers to the notion in group theory (not that in differential geometry or supergravity): The \emph{torsion subgroup} of an abelian group, $\mathrm{Tor}(A) \subset A$, is the subgroup of elements of finite order. For instance, $\mathbb{Z}_{/n}$ is pure torsion while $\mathbb{Z}$ has no torsion.
    When a brane species has charges in a torsion subgroup, this means, in physics imagery, that one such brane annihilates against a ``stack'' of $n-1$ such branes, for some finite number $n$.  
  }
  M2-brane charges taking values in $\mathbb{Z}_{/12}$, already as measured near their horizon. Below we see an induced proliferation of torsion brane species after dimensional reduction along a small circle to 10D, cf. \cref{HomotopyOfCyc1S4,HomotopyOfCyc1S4Table,HomotopyOfTheFiber}. 

  The appearance of torsion brane charges is a topological effect that is invisible on the level of flux densities, but generically exhibited by non-trivial flux quantization. This is a well-known phenomenon: 
  \begin{enumerate}
  \item 
  Torsion M2-charges seen in ordinary cohomology have been discussed in \cite[\S~2.2]{AharonyBergmanJafferis2008} (``fractional M2-branes''). 

  \item
  Torsion D-brane charges seen in K-theory have been discussed in \parencites{Braun2000}{BrunnerDistler2001}{BrunnerDistlerMahajan2001}{Mendez-DiezRosenberg2012}. 

  \item In particular, type I $p$-brane charges quantized in $\mathrm{KO}$-theory exhibit torsion already on Minkowski spacetime (cf. \parencites[p. 23, 36]{Witten1998}[p. 8,9]{Gukov2000}[p. 21]{OlsenSzabo2000}) and for the same reason as \cref{M2ChargeAccordingToHypothesisH} above: that their classifying space, $B \mathrm{O}$, has torsion already in its homotopy groups, cf. \cref{DBraneChargeInOmegaBOCohomology}.

  \end{enumerate}
  
  Nevertheless, the reader unhappy with this prediction \cref{M2ChargeAccordingToHypothesisH} of torsion M2-brane charge is free to IR-complete 11D SuGra using another flux quantization law $\mathcal{A}$ with $\mathfrak{l}\mathcal{A} \simeq \mathfrak{l}S^4$. But it is with the choice $\mathcal{A} \defneq S^4$ in 11D that twisted K-theory emerges in 10D, in \cref{KTheoryEmerging} below.
\end{remark}

\paragraph
{Model of the C-Field quantizing both magnetic M5 and electric M2-brane charges}
As such, \emph{Hypothesis H}  is a ``model of the C-field'' (where the hypothesis is that \emph{it is the right model} for the purpose of completing to M-theory), alternative to a previously proposed class of models \parencites[\S~2.7]{HopkinsSinger2005}{DFM2007}{FSS15-ModuliStack}{DonagiWijnholt2023}. 

A key novelty here is that Cohomotopy quantizes not just the magnetic C-field flux ($G_4$, sourced by black M5-branes) but compatibly also the electric C-field flux ($G_7$, sourced by black M2-branes), whose non-linear Gauss law \cref{ElectricGaussLawIn11D} prevents it from being quantized in any abelian generalized cohomology theory (\cref{AbelianCohDoesNotQuantizeNonlinear}). 

This is relevant for the present purpose, since also \emph{Hypothesis K} in 10D postulates that the electric fluxes are to be quantized (cf. \cref{11DAnd10DFluxSpecies}). To relate the two, we next discuss how electromagnetic flux/charge quantization behaves under dimensional reduction along circle fibers.

\subsection
{Dimensional Reduction}
\label
{MCircleReduction}

With a classifying space $\mathcal{A}$ thus encoding (by \cref{ProperFluxQuantization}) the entire IR-complete higher gauge field content, including the flux densities, the local gauge potentials and the corresponding brane charges, we may ask for the formulation of \emph{double dimensional reduction} (\parencites{DuffHoweInamiStelle1987}{Townsend1995}) at this refined level.

It turns out that this has a beautiful answer. Consider:
\begin{enumerate}
 \item $\;L \mathcal{A} := \mathrm{Map}\bracket({S^1, \mathcal{A}})$---the \emph{free loop space}

  (namely, the space of maps from the circle).
  
  \item 
  $\mathrm{Cyc}\,\mathcal{A} := \bracket({L \mathcal{A}}) \sslash S^1$---the \emph{cyclic loop space}

  (namely, the homotopy quotient by the circle group of rigid rotations of loops).

  \item 
  $c_1^{C} \in H^2\bracket({ \mathrm{Cyc}\,\mathcal{A}; \mathbb{Z} })$---the Chern class of the circle principal bundle $\inlinetikzcd{L \mathcal{A} \ar[r] \& \mathrm{Cyc} \mathcal{A}}$

  (classified by $\inlinetikzcd{ (L \mathcal{A} \to \ast)\sslash S^1 : \mathrm{Cyc}\mathcal{A} \ar[r] \& B S^1 }$).
\end{enumerate}
Then:
\begin{proposition}
\label[proposition]{DimReductionViaCyc}
  For $\inlinetikzcd{X^{d+1} \ar[r] \& X^d}$ a circle principal bundle with Chern class $c_1 \in H^2\bracket({X^{d}; \mathbb{Z}})$, there is a canonical bijection between (twisted) nonabelian cohomology sets \cref{NonabelianCohomology,TwistedNonabelianCohomology} of the form:
  \begin{equation}
  \label{CyclificationAdjointness}
    H^1\bracket({
      X^{d+1};
      \Omega \,\mathcal{A}
    })
    \;\simeq\;
    H^1_{c_1}
    \bracket({
      X^{d};
      \Omega 
      \,
      \mathrm{Cyc}\,\mathcal{A}
    })
    \mathrlap{\,.}
  \end{equation}
\end{proposition}
On the right we have classifying maps $\inlinetikzcd{f : X^{d} \ar[r] \&  \mathrm{Cyc}\mathcal{A}}$ with $f^\ast c_1^{C} \simeq c_1$, hence where the class of the $S^1$-bundle on the left is remembered as a cohomology datum on $X^d$: the \emph{KK-monopole} charge.

\paragraph
{Comparison to the Literature}
\Cref{DimReductionViaCyc} is due to \parencites[\S~4]{Schreiber2017}[\S~2.2]{BMSS2019}[\S~2.2]{SS24-Cyc}, following the analogous rational observation in \parencites[\S~3]{FSS17-Sphere}[\S~3]{FSS18-TD}[\S~2]{FSS18-TDR}[\S~8]{FSS2019}, for which see also \parencites[\S~2.4]{SatiVoronov2024}[\S~2.2]{GSS25-TD}. 
The special case when $\mathcal{A} \simeq B G$ is the classifying space of a connected Lie group $G$ was essentially considered in \cite{BergmanVaradarajan2004}, with focus on $G \defneq E_8$; before that, the further specialization to  $c_1 = 0$ was considered in \parencites{MathaiSati04-E8}, cf. also \parencites{Sati2009}{Sati2010-E8}---all this in our context of M/IIA reduction (to which we come in a moment, \cref{PlainReductionToIIA}). Much earlier, the special case $\mathcal{A} = B G$, $c_1 = 0$ of \cref{DimReductionViaCyc} was called the \emph{caloron correspondence} in \parencites{GarlandMurray1988}, cf. \parencites{MurrayVozzo2009}{HekmatiMurrayVozzo2011}.

\begin{example}
\label[example]{CycOfBnZ}
  For $\mathcal{A} \simeq B^{n+1} \mathbb{Z}$, inducing via \cref{AdmissibleFluxQuantizationLaw} a single closed 
  flux form,
  \[
    \Omega^1_{\mathrm{cl}}\bracket({
      X^{d+1};
      \mathfrak{l}
      B^{n+1} \mathbb{Z}
    })
    \,\simeq\,
    \left\{
    \begin{aligned}
      G_{n+1} & \in 
      \Omega^{n+1}_{\mathrm{dR}}\bracket({
        X^{d+1}
      })
    \end{aligned}
    \,\middle\vert\,
    \begin{aligned}
      \mathrm{d}\, G_{n+1} & = 0
    \end{aligned}
    \right\}
    \mathrlap{\,,}
  \]
  on the total space of a circle bundle 
  $\inlinetikzcd{ X^{d+1} \ar[r] \& X^d }$,
  we have:
  \begin{equation}
    \label{CycOfBnZAsHomotopyFiber}
    \mathrm{Cyc} 
    \bracket({ B^{n+1} \mathbb{Z} })
    \,\simeq\,
    \mathrm{hofib}\bracket({
    \inlinetikzcd{
      B^2 \mathbb{Z} 
      \times
      B^n \mathbb{Z}
      \ar[
        r,
        "{
          \cup
        }"
      ]
      \&
      B^{n+2} \mathbb{Z}
    }
    })
    \mathrlap{\,,}
  \end{equation}
  hence sitting in
  \begin{equation}
    \label{PastingDiagramForCyc}
    \begin{tikzcd}[
      column sep=25pt
    ]
    B^n \mathbb{Z}
    \times
    B^{n+1} \mathbb{Z}
    \ar[r]
    \ar[
      d,
      ->>
    ]
    \ar[
      dr,
      phantom,
      "{ \lrcorner }"{pos=.1}
    ]
    & 
    \mathrm{Cyc}\, B^{n+1} \mathbb{Z}
    \ar[r]
    \ar[d]
    \ar[
      dr,
      phantom,
      "{ \lrcorner }"{pos=.1}
    ]
    &
    \ast
    \ar[d]
    \\
    B^n \mathbb{Z}
    \ar[
      r,
      hook
    ]
    &
    B^2 \mathbb{Z}
    \times
    B^n \mathbb{Z}
    \ar[
      r,
      "{ \cup }"
    ]
    &
    B^{n+2} \mathbb{Z}
    \mathrlap{\,,}
    \end{tikzcd}
  \end{equation}
  whose rational model (\cite{VigueBurghelea1985}, cf. \parencites[Prop. 3.2]{FSS17-Sphere}[\S~2.5]{SatiVoronov2024}) is characterized via \cref{CEGenerators} by:
  \begin{equation}
    \Omega^1_{\mathrm{cl}}\bracket({
      X^{d};
      \mathfrak{l} 
      \mathrm{Cyc} B^{n+1}
      \mathbb{Z}
    })
    \;\simeq\;
    \left\{
    \begin{aligned}
      F_{2} & 
      \in
      \Omega^{2}_{\mathrm{dR}}
      \bracket({
        X^{d}
      })
      \\
      G_{n} & 
      \in
      \Omega^{n}_{\mathrm{dR}}
      \bracket({
        X^{d}
      })
      \\
      G_{n+1} & 
      \in
      \Omega^{n+1}_{\mathrm{dR}}
      \bracket({
        X^{d}
      })
    \end{aligned}
    \,\middle\vert\,
    \begin{aligned}
      \mathrm{d}\, F_2 & = 0
      \mathclap{\phantom{\bracket({ X^d})}}
      \\
      \mathrm{d}\, G_n & = 0
      \mathclap{\phantom{\bracket({ X^d})}}
      \\
      \mathrm{d}\, G_{n+1} & = F_2 \wedge G_n
      \mathclap{\phantom{\bracket({ X^d})}}
    \end{aligned}
    \right\}
    \mathrlap{.}
  \end{equation}
  We see: Under dimensional circle-reduction, the flux density form splits as expected into its basic part and its fiber integral (the \emph{double} dimensional reduction, going back to \cite{DuffHoweInamiStelle1987}), and picks up a companion 2-form whose de Rham class is identified, under \cref{DimReductionViaCyc}, with the real Chern class of the circle fibration.
\end{example}

\begin{example}[Topological T-duality via Cyclification]
\label[example]
{CycB3ZInTopTDuality}
  The special case $n =2$ of \cref{CycOfBnZ}, hence the space $\mathrm{Cyc}\, B^3 \mathbb{Z} \simeq \mathrm{Cyc}\, B^2 \mathrm{U}(1)$, is (\parencites[\S~3]{BunkeRumpfSchick2006}, cf. \parencites[p. 7]{SS24-Cyc}) the classifying space for the NS-sector of  \emph{topological T-duality} (\parencites{BouwknegtEvslinMathai2003}{BunkeSchick2005}) and compatibly so with the structure of type II supergravity (\parencites[Rem. 7.2]{FSS18-TD}[Ex. 2.27 \& Lem. 3.19]{GSS25-TD}): According to \cref{DimReductionViaCyc} it classifies NS5-brane charges after circle dimensional reduction to 9D.

  Concretely, there is an auto-equivalence $\inlinetikzcd{ \mathrm{Cyc}\, B^3 \mathbb{Z} \ar[r, "{ \tau }"] \& \mathrm{Cyc}\, B^3 \mathbb{Z} }$ swapping the two degree-2 classes, covered by an equivalence of cyclifications of the twisted K-theory classifying spaces for type IIA and IIB,%
  \footnote{%
    On the right of \cref{TopTDualityViaCyclification}
    the subscript $\mathrm{wd} = 0$ denotes the 
    0-component in 
    $\pi_0\bracket({
      \mathrm{Cyc}\bracket({
        \mathrm{KU}_1 
          \sslash 
        B \mathrm{U}(1)
      })
    })
      \simeq
    \pi_1\bracket({\mathrm{U}})
      \simeq 
    \mathbb{Z}
    $,
    hence the IIB sector where $m \defneq \int_{S^1_{\mathrm{B}}} F_1^{\mathrm{IIB}} = 0$, cf. \cref{ReductionTo9D}. This is the T-dual condition of masslessness, $m \defneq F_0 = 0$, in IIA on the left.
  }
  \begin{equation}
  \label
  {TopTDualityViaCyclification}
    \begin{tikzcd}[
      column sep=10pt
    ]
      \mathrm{Cyc}\bracket({
        \bracket({
          \mathrm{KU}_0
          \sslash
          B \mathrm{U}(1)
        })_{\mathrm{rk=0}}
      })
      \ar[
        rr,
        <->, 
        "{ \sim }"
      ]
      \ar[
        d,
        hook,
      ]
      &&
      \mathrm{Cyc}\bracket({
        \bracket({
          \mathrm{KU}_1
          \sslash
          B \mathrm{U}(1)
        })
      })_{\mathrm{wd}=0}
      \ar[
        d,
        hook,
      ]
      \\
      \mathrm{Cyc}\bracket({
        \mathrm{KU}_0
        \sslash
        B \mathrm{U}(1)
      })
      \ar[
        rr,
        <->, 
        "{ \sim }"
      ]
      \ar[
        d,
        ->>,
        "{
          \mathrm{Cyc}\, h^K_3
        }"{swap}
      ]
      &&
      \mathrm{Cyc}\bracket({
        \mathrm{KU}_1
        \sslash
        B \mathrm{U}(1)
      })
      \ar[
        d,
        ->>,
        "{
          \mathrm{Cyc}\, h^K_3
        }"
      ]
      \\
      \mathrm{Cyc}\, B^3 \mathbb{Z}
      \ar[
        rr,
        <->,
        "{ \sim }",
        "{ \tau }"{swap}
      ]
      &&
      \mathrm{Cyc}\, B^3 \mathbb{Z}
      \mathrlap{\,,}
    \end{tikzcd}
  \end{equation}
  such that the equivalence between charges that this equivalence of classifying spaces induces in 9D corresponds, under \cref{DimReductionViaCyc}, to topological T-duality in 10D. (This statement is shown at the rational level in \parencites{FSS18-TDR}, following the supergravity analysis in \parencites[\S~5]{FSS18-TD}, cf. \parencites[\S~3.2]{GSS25-TD}, but it holds beyond the rational approximation as stated.)

  In a somewhat analogous but rather more intricate manner, we find in \cref{KTheoryEmerging} that $\mathrm{Cyc}\, B^4 \mathbb{Z}$---the classifying space for the circle dimensional reduction of isolated M5-brane charges---controls the dimensional reduction of M-brane charges to 10D; see \cref{BraneChargeFromCohomotopy,LEMtau4BUModBU1IsCycB4Z,BraneChargesFromTwistedKTheory,GroupStructureOnBraneCharges},
  used in \cref{CycB4ZInProofOfComparisonMap} of the proof of \cref{Comparison_map_to_twisted_K_theory} below.  
\end{example}

After these preliminaries,  our goal now is to analyze this formulation of double dimensional reduction for the more interesting case of the supergravity C-field, and to discover K-theory along the way.

\section
{Results}
\label
{Results}

We focus now on the IR-completion of 11D SuGra (cf. \cite{GSS26-SuGra}) given by C-field flux quantization (\cref{ProperFluxQuantization}) in 4-Cohomotopy, $\mathcal{A} \defneq S^4$, according to \cref{TheCFieldIn11D} (``Hypothesis H''). Our goal is to identify the corresponding IR-completion of 10D IIA and its relation to flux quantization in twisted K-theory (``Hypothesis K'').

\paragraph
{Notation}
\label{Notation}
We use fairly standard notation in algebraic topology, but for clarity we highlight that: 
\begin{enumerate}
\item
$R\langle S\rangle$ denotes the linear space of a set/list $S$ of symbols over a ring $R$.

\item 
\label{nTruncation}
$\truncation{n}{\mathcal{A}}$ denotes the \emph{$n$-truncation} (or 
\emph{$n$th Postnikov stage}, also denoted $\tau_{\leq n}\bracket({\mathcal{A}})$ or similar) of a space $\mathcal{A}$ (cf. \parencites[\S~IV]{Whitehead1978}[Cor. 3.7]{GoerssJardine2009}[Prop. 1.9]{FSS23-Char}). This is in particular the aspect of $\mathcal{A}$ that is detectable by homotopy classes of maps out of $n$-manifolds $X^n$ (cf. \parencites[Prop. 1.20]{FSS23-Char}):
\begin{equation}
  \pi_0\,\mathrm{Map}\bracket({
    X^n
    ,
    \mathcal{A}
  })
  \simeq
  \pi_0\,\mathrm{Map}\bracket({
    X^n
    ,
    \truncation{n}{\mathcal{A}}
  })
  \mathrlap{\,.}
\end{equation}

Readers of M-theory literature may be familiar with an allusion to this operation in the observation that homotopy groups of the exceptional Lie group $E_8$ are those of $B^3 \mathbb{Z}$ up to degree 14 (\parencites[p. 4]{Witten1997Flux}[p. 5]{DMW2003}[(3.2)]{DFM2007}), which means:
$
  \truncation
    {15}
    {B E_8}
  \,\simeq\,
  B^4 \mathbb{Z}
$.

\item
\label{WeakHomotopyEquivalence}
$\inlinetikzcd{f : \mathcal{A} \ar[r, "\sim"] \& \mathcal{B}}$ denotes a weak homotopy equivalence (or just ``equivalence'', for short),
hence a map that induces: 
\begin{enumerate}
\item 
a bijection on connected components, $\inlinetikzcd{\pi_0(f) : \pi_0\bracket({\mathcal{A}}) \ar[r, "{ \sim }"] \& \pi_0\bracket({\mathcal{B}})  }$, 
\item
isomorphisms on all homotopy groups:
$\inlinetikzcd{ \pi_{k+1}(f,a) : \pi_{k+1}\bracket({ \mathcal{A}, a }) \ar[r, "{ \sim }"] \& \pi_{k+1}\bracket({\mathcal{B}, f(a)}) }$, $\forall a \in \mathcal{A}, k \in \mathbb{N}$.
\end{enumerate}

\item
$\inlinetikzcd{ f : \mathcal{A} \ar[r, "\sim_n"] \& \mathcal{B} }$ denotes an \emph{$(n+1)$-connected map} or \emph{$(n+1)$-equivalence} (cf. \cite[p. 144]{tomDieck2008}), a map that:
\begin{enumerate}
\item 
is an equivalence (\cref{WeakHomotopyEquivalence}) under $n$-truncation (\cref{nTruncation}), 
$\inlinetikzcd{ \truncation{n}{f} : \truncation{n}{\mathcal{A}} \ar[r, "{ \sim }"] \& \truncation{n}{\mathcal{B}} }$, 
\item
surjects on $(n+1)$-homotopy: $\inlinetikzcd{\pi_{n+1}(f,a) : \pi_{n+1}\bracket({\mathcal{A},a}) \ar[r, ->>] \& \pi_{n+1}\bracket({\mathcal{B}, f(a)}) }$, $\forall a \in \mathcal{A}$.
\end{enumerate}
In particular, such maps induce bijections on homotopy classes of maps out of $n$-manifolds $X^n$:
\begin{equation}
  \inlinetikzcd{
    f_\ast :
    \pi_0
    \,
    \mathrm{Map}\bracket({
      X^n
      ,
      \mathcal{A}
    })
    \ar[
      r,
      "{ \sim }"
    ]
    \&    
    \pi_0
    \,
    \mathrm{Map}\bracket({
      X^n
      ,
      \mathcal{B}
    })
    \mathrlap{\,.}
  }
\end{equation}
\end{enumerate}

\subsection
{
 \texorpdfstring
  {Plain Reduction M$\to$IIA}
  {Plain Reduction M to IIA}
}
\label
{PlainReductionToIIA}

 Upon \emph{plain} dimensional circle reduction according to \cref{MCircleReduction}, \emph{Hypothesis H} induces an IR-completion of 10D IIA SuGra given by NS/RR flux quantization classified by $\mathrm{Cyc}\, S^4 \defneq \bracket({ L S^4 }) \sslash S^1$ \cref{CyclificationAdjointness}. This construction of IR-complete 10D SuGra has been discussed in detail in \parencites{GiotopoulosSati2026} (see also \parencites{Banerjee2026-D4}).
 \begin{equation}
 \label{MIIAreductionSchematics}
 \begin{tikzcd}
   S^1
   \ar[r]
   &[-15pt]
   X^{10}
   \ar[
     d,
     ->>
   ]
   \ar[
     rr,
     dashed,
     "{
       \text{C-field}
     }",
     "{
       \text{charges}     
     }"{swap,yshift=1pt}
   ]
   &&
   S^4
   \\
   &
   X^{9}
   \ar[
     rr,
     dashed,
     "{
       \text{NS/RR-charges}
     }"{yshift=-1pt},
     "{
       \text{excluding D0}      
     }"{swap, yshift=1pt}
   ]
   \ar[
     dr,
     "{ c_1 }"'
   ]
   &&
   \mathrm{Cyc}\,S^4
   \mathrlap{\,.}
   \ar[
     dl,
     "{
       c_1^{C}
     }"
   ]
   \\[-12pt]
   &&
   B S^1
 \end{tikzcd}
\end{equation}
Using a classical expression for the rational minimal model of cyclic loop spaces (\cite{VigueBurghelea1985}, cf. \parencites[Prop. 3.2]{FSS17-Sphere}[\S~2.5]{SatiVoronov2024}) 
one finds, in particular:
\begin{proposition}[{\parencites[\S~3]{FSS17-Sphere}, cf. \parencites[Ex. 2.26]{GSS25-TD}}]
\label[proposition]{ObtainingGaussLawsForCycS4}
  The Gauss laws corresponding, via  \cref{AdmissibleFluxQuantizationLaw}, to $\mathrm{Cyc}\, S^4$ are:
  \begin{equation}
  \label{GaussLawsForCycS4}
    \Omega^1_{\mathrm{cl}}\bracket({
      X^{9};
      \mathfrak{l}
      \mathrm{Cyc}\,
      S^4
    })
    \;\simeq\;
    \left\{
    \begin{aligned}
      H_3 & \in
      \Omega^3_{\mathrm{dR}}\bracket({X^9})
      \\
      H_7 & \in
      \Omega^7_{\mathrm{dR}}\bracket({X^9})
      \\
      F_2 & \in
      \Omega^2_{\mathrm{dR}}\bracket({X^9})
      \\
      F_4 & \in
      \Omega^4_{\mathrm{dR}}\bracket({X^9})
      \\
      F_6 & \in
      \Omega^6_{\mathrm{dR}}\bracket({X^9})
    \end{aligned}
    \,\middle\vert\,
    \begin{aligned}
      \mathrm{d}\, H_3 & = 0
      \mathclap{\phantom{\bracket({X^9})}}
      \\
      \mathrm{d}\, H_7 & =
      \tfrac{1}{2}F_4 \wedge F_4
      - F_2 \wedge F_6
      \mathclap{\phantom{\bracket({X^9})}}
      \\
      \mathrm{d}\, F_2 & = 0
      \mathclap{\phantom{\bracket({X^9})}}
      \\
      \mathrm{d}\, F_4 & = H_3 \wedge F_2
      \mathclap{\phantom{\bracket({X^9})}}
      \\
      \mathrm{d}\, F_6 & = H_3 \wedge F_4
      \mathclap{\phantom{\bracket({X^9})}}
    \end{aligned}
    \right\}
    \mathrlap{.}
  \end{equation}
\end{proposition}

We observe:
\begin{enumerate}
  \item As far as they go, these data 
  in \cref{GaussLawsForCycS4} are the correct Gauss laws of 10D IIA SuGra, as obtained by dim-reduction from 11D (cf. \cite[(4.13--4.15)]{MathaiSati04-E8}), in particular \emph{including} the nonlinear Gauss law of $H_7$ \cref{TheH7GaussLaw}, which is missed by flux quantization in twisted K-theory (or in any other abelian generalized cohomology theory, cf. \cref{AbelianCohDoesNotQuantizeNonlinear}).

  \item But the electric flux $F_8$ sourced by black D0-branes (cf. \cref{11DAnd10DFluxSpecies,DBraneAndSourcedFluxes,NSRRFluxGaussLaw}) does not appear. Indeed, the charge sourcing $F_8$ is not meant to have an 11D origin via plain circle reduction, and this is confirmed here by the cyclification procedure. Accordingly, flux quantization of 10D IIA in $\mathrm{Cyc}\, S^4$, while a valid IR-completion, does not enforce quantization of black D0-brane charge.
\end{enumerate}

This failure of $F_8$ to originate from 11D under ordinary circle reduction is of course the long-familiar issue (cf. \cref{11DAnd10DFluxSpecies}) --- which around \cref{TheDMWLift} above we noticed was tacitly addressed by DMW by declaring/postulating that the dim-reduced field content should be lifted to K-theory, somehow.

Our next and main step is to give this idea of ``\emph{dimensionally reducing but then adjoining $F_8$}'' a precise and conceptually satisfactory meaning and explanation against the backdrop of IR-completion by nonabelian electromagnetic flux quantization.  (Other approaches to addressing this issue  were considered previously in \parencites{BMSS2019}{BaSS26-UnstableK}, but only the present approach turns out to relate to actual twisted K-theory, in \cref{KTheoryEmerging} below.)

\subsection
{Small Circle Reductions}
\label
{MorseBottRegularity}

To motivate the following \cref{SmallCycSpaces}, we observe:
\begin{enumerate}
\item In dimensional reduction from 11D to 10D IIA, the M-theory circle is supposed to be \emph{small} and higher KK-modes \emph{suppressed}. 

\item After reduction via cyclification (\cref{MCircleReduction}), the variation of the 11D charges over the M-theory circle becomes the trajectory of the loops $\gamma \in L S^4 \defneq \mathrm{Map}\bracket({ S^1, S^4 })$, hence of their classes $[\gamma] \in \mathrm{Cyc}\, S^4 \defneq (L S^4) \sslash S^1$.

\item 
In order to express that these loops are \emph{small} and \emph{suppressed}, we should hence \emph{bound their variation}, the latter expressed by the Dirichlet functional (cf. \cref{MeaningOfDirichletFunctional} below):
\begin{equation}
\label{EnergyFunctional}
  E(\gamma)
  :=
  \tfrac{1}{2}
  \int_{S^1}
  \bracket\vert{ \dot \gamma(s) }\vert^2
  \,
  \mathrm{d}s
  \mathrlap{\,.}
\end{equation}
\end{enumerate}

We proceed to make this idea precise, using classical concepts from infinite-dimensional Morse theory (cf. \parencites{Palais1963}{Palais1966}).

\begin{lemma}
\label[lemma]{SobolevApproximation}
  The inclusion
  of the Sobolev space of loops for which the Dirichlet functional \cref{EnergyFunctional} is well-defined is a homotopy equivalence:
  \begin{equation}
    \inlinetikzcd{
    W^{1,2}\bracket({
      S^1, S^4
    })
    \sslash S^1
    =:
    \mathrm{Cyc}_\infty  S^4
    \ar[
      rr, 
      hook, 
      "{ \sim }"
    ]
    \&\&
    \mathrm{Cyc}\, S^4
    \mathrlap{\,.}
    }
  \end{equation}
\end{lemma}
\begin{proof}
   Before quotienting by $S^1$ this is the statement of \cite[Thm. 1.2.10]{Klingenberg1978}; our statement follows immediately since $(-)\sslash S^1$ is an $\infty$-functor on $S^1$ $\infty$-actions.
\end{proof}

\begin{definition}[{cf. \parencites{Ziller1977}{Oancea2015}}]
\label[definition]{SmallCycSpaces}
For $n \in \mathbb{N}$, consider:
\begin{enumerate}
\item $E_1 \in \mathbb{R}$ --- the value $E(\gamma)$ \cref{EnergyFunctional} of a loop $\gamma$ that is a great circle: a geodesic that goes once around the 4-sphere (the numerical  value depends on how one parameterizes $S^1$ and $S^4$, which is irrelevant as long as it is fixed once and for all),

\item $E_n :=  n^2 E_1$ --- the energy of a geodesic loop that goes $n$ times around the 4-sphere,

\item 
\label{Ln}
$L_n S^4 \subset W^{1,2}\bracket({S^1, S^4}) \subset L S^4$ --- the subspace of loops $\gamma$ with $E_n \geq E(\gamma)$ \cref{EnergyFunctional},

\item 
$\mathrm{Cyc}_n S^4 := \bracket({ L_n S^4 }) \sslash S^1 \subset \mathrm{Cyc}_\infty S^4$ --- its cyclification (using that the Dirichlet functional \cref{EnergyFunctional} is $S^1$-invariant).

\item
The resulting filtration:
\begin{equation}
\label{EnergyFiltration}
  \begin{tikzcd}
    S^4 
    \times
    B S^1
    \simeq
    \mathrm{Cyc}_0 S^4
    \ar[r, hook]
    &
    \mathrm{Cyc}_1 S^4 
    \ar[r, hook]
    &
    \mathrm{Cyc}_2 S^4
    \ar[r, hook, dotted]
    &
    \mathrm{Cyc}_\infty S^4
    \simeq
    \mathrm{Cyc}\, S^4
    \mathrlap{.}
  \end{tikzcd}
\end{equation}
\end{enumerate}
\end{definition}
Here, the first stage has a particularly simple description (cf. \cref{PathOfCircles}):
\begin{lemma}
\label[lemma]{Cyc1AsAPushout}
  The first stage, $\mathrm{Cyc}_1 S^4$, in \cref{EnergyFiltration} is homotopy equivalent to:
  \begin{enumerate}
  \item 
  the space of unparameterized oriented round circles in $S^4$,
  where degenerate circles quotiented by rigid reparameterization are included as copies of $\ast \sslash S^1 \simeq B S^1$,
  
  \item 
  hence to the homotopy pushout
  \begin{equation}
  \label{TheHomotopyPushout}
    \begin{tikzcd}[
      row sep=30pt,
      column sep=70pt
    ]
     \left\{
      \substack{
        \textup{oriented great circles}
        \\
        \textup{around a given pole}
      }
      \right\}
      \ar[
        r,
        "{
          \substack{
            \textup{forget pole}
          }
        }"
      ]
      \ar[
        d,
        "{
          \substack{
            \textup{remember}
            \\
            \textup{pole}
          }
        }"'
      ]
      &
     \left\{
      \substack{
        \textup{oriented great}
        \\
        \textup{circles in $S^4$}
      }
      \right\}
      \ar[d]
      \ar[
        dl,
        Rightarrow,
        shorten=5pt,
        "{
          \substack{ 
            \textup{shrink circles}
            \\
            \textup{to pole}
          }
        }"{sloped, description}
      ]
      \\
      \big\{
      \substack{
        \textup{points} 
        \\
        \textup{in $S^4$}
      }
      \big\}
      \times 
      B S^1
      \ar[
        r,
        hook
      ]
      &
      \mathrm{Cyc}_1 S^4
      \mathrlap{\,,}
    \end{tikzcd}
  \end{equation}
  where the second component of the left map classifies the $S^1$-bundle of rigid parameterizations of the oriented great circles.
  \end{enumerate}
\end{lemma}
\begin{proof}
This statement is a digest of \cite[\S~2.3, esp. Thm. 2.4.10 with Prop. 2.5.3]{Klingenberg1978}: The round circles correspond to loops that traverse these round circles at unit rate, up to rigid reparameterization (up to shift of their basepoint).
\end{proof}

\begin{figure}[htb]
\caption{\label{PathOfCircles}%
  According to \cref{Cyc1AsAPushout}, a generic point in the small-cyclification of the 4-sphere, $\mathrm{Cyc}_1 S^4$  (\cref{SmallCycSpaces,TowardsSmallCyclifiedCohomotopy}), is equivalently an embedded round circle in the 4-sphere, here playing the role of an image of the M-circle in the classifying space of M-brane charges. Indicated is a continuous path of such generic points in $\mathrm{Cyc}_1 S^4$. In addition, the circles may degenerate (shrink to vanishing diameter), in which case the path may continue through a copy of $\ast \sslash S^1 \simeq B S^1$ (the shape of a reparameterization orbifold singularity, which we do not try to visualize here) at the position of the degenerate circle.  
}
  \centering
    \includegraphics[width=12cm]
      {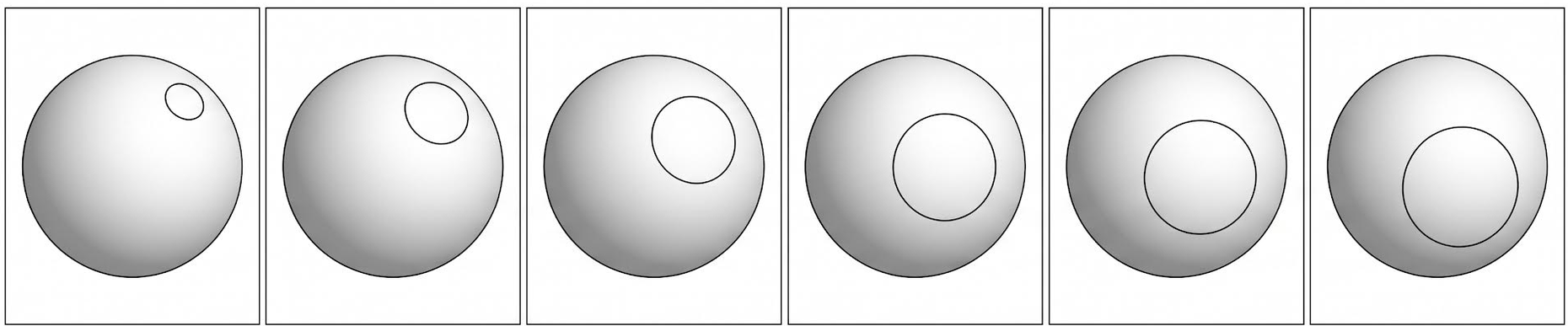}
\end{figure}

This is remarkable in view of the understanding of dimensional reduction of brane charges from \cref{MCircleReduction}:
\begin{remark}[Small-cyclified Cohomotopy]
\label[remark]
{TowardsSmallCyclifiedCohomotopy}
\begin{enumerate}
\item
  The homotopy type of loops $\gamma$ with bound $E(\gamma) \leq E$ for intermediate energy $E_n \leq E < E_{n+1}$ turns out to be equivalent to $L_n\, S^4$, hence its $S^1$-quotient to $\mathrm{Cyc}_n\, S^4$; therefore, the filtration \cref{EnergyFiltration} already shows all intermediate homotopy types of bounded $E(\gamma)$. 

\item
  Therefore, $\mathrm{Cyc}_1 S^4$ is unambiguously the \emph{first stage} after the trivial 0th stage (of vanishing M-circle). 
  
\item With \cref{Cyc1AsAPushout}, we may call $\mathrm{Cyc}_1 S^4 \subset \mathrm{Cyc}\, S^4 $ the \emph{small-cyclification} of $S^4$, equivalently consisting only of those loops that run once through a round circle on $S^4$, cf. \cref{PathOfCircles}.

\item
\label{HigherCycnTruncation}
  All higher $\mathrm{Cyc}_{n > 1} S^4$ already have the same 10-truncation as $\mathrm{Cyc}\, S^4$ itself; therefore $\mathrm{Cyc}_1 S^4$ is the only relevant stage in \cref{EnergyFiltration} for flux quantization in 10D.

\item
  Accordingly, we call the nonabelian cohomology \cref{NonabelianCohomology} with coefficients in $\Omega \mathrm{Cyc}_1 S^4$ the \textbf{small-cyclified 4-Cohomotopy}:
  \begin{equation}
    \label{SmallCyclifiedCohomotopy}
    H^1\bracket({
      -; \Omega\, \mathrm{Cyc}_1 S^4 
    })
    \,\defneq\,
    \pi_0
    \mathrm{Map}\bracket({
      -, \,
      \mathrm{Cyc}_1 S^4
    })
    \mathrlap{\,.}
  \end{equation}
\end{enumerate}
\end{remark}

\begin{remark}[KK-Modes and Small cyclification]
\label[remark]{MeaningOfDirichletFunctional}
  The Dirichlet functional \cref{EnergyFunctional} would be a physical \emph{energy functional} if $S^4$ here played the role of a physical space and $\gamma$ played the role of a physical field.
  Correspondingly, the filtration \cref{EnergyFiltration} would then be reminiscent of a KK-tower.
  But since $S^4$ here instead plays the role of a \emph{classifying space} and $\gamma$ of a \emph{classifying map}, this suggestive interpretation is at least subtle. 
\end{remark}

In any case, \cref{TowardsSmallCyclifiedCohomotopy} allows us to make precise the idea of dimensional reductions of topological charges along ``small'' circles, where the contribution from the M-theory circle to the 10D brane charges is in some sense bounded: We take these to be those reductions \cref{MIIAreductionSchematics} for which the classifying map takes values in loops of $E_1$-bounded variation, hence where it factors (up to homotopy) through the first filtration stage \cref{EnergyFiltration} of the cyclified classifying space, as shown by the following dashed lift in \cref{LiftingThroughCyc1}:
\begin{equation}%
\label{LiftingThroughCyc1}
  \begin{tikzcd}[
    column sep=35pt
  ]
    \mathrlap{\phantom{C}}
    \\
    X^{10}
    \ar[
      r,
      "{
        \substack{
          \text{M-brane} 
          \\
          \text{charges}
        }
      }"
    ]
    &
    S^4
    \\
    \mathrlap{\phantom{C}}
  \end{tikzcd}
  \;\;\;\;\;
  \begin{tikzcd}
    \mathrlap{\phantom{C}}
    \\[+12pt]
    \underset
      {\text{\cref{CyclificationAdjointness}}}
      {\longmapsto}
    \\
    \mathrlap{\phantom{C}}
  \end{tikzcd}
  \;\;\;\;\;
  \scalebox{.9}{their dim-reduction...}
  \;\;
  \begin{tikzcd}[
   column sep=15pt
  ]
    &\phantom{--}& 
    \mathrm{Cyc}_1 S^4
    \ar[
      d,
      hook
    ]
    \\[-10pt]
    X^9
    \ar[
      rr,
      "{ \text{...on any} }"{yshift=-2pt},
      "{ \text{M-circle} }"{swap, yshift=+1pt}
    ]
    \ar[
      dr,
      "{ c_1 }"'
    ]
    \ar[
      urr,
      dashed,
      "{
        \text{ ...on small}
      }"{sloped}
    ]
    &\phantom{--}&
    \mathrm{Cyc}\, S^4
    \mathrlap{\,.}
    \ar[
      dl,
      "{ c_1^{C} }"
    ]
    \\[-10pt]
    & 
    B S^1
  \end{tikzcd}
\end{equation}

For this notion we prove the following results \cref{Classification-of-lifts,The-full-IIA-EM-Gauss-laws,Comparison_map_to_twisted_K_theory}. Since the proofs are intricate, here we outline enough of their main steps so that the reader may reconstruct them.
\begin{lemma}
\label[lemma]{ScndStageInclusion}
The inclusion $\inlinetikzcd{ \mathrm{Cyc}_1 S^4 \ar[r, hook, "{ i }"] \& \mathrm{Cyc}_2 S^4 }$ \cref{EnergyFiltration}:
  \begin{enumerate}
  \item
  \label
  {ScndStageAsAttachmentOf9DiskBundle}
  is the
  attachment of a disk bundle $D^9_{B}$, of fiber dimension 9, 
  
  along its boundary sphere bundle $\partial D^9_B$, fiberwise over $B S^1$,
  
  \item
  \label
  {ScndStageAttachmentOverBaseB}
  over the space $B \defneq K_2 \sslash S^1$ of parameterized loops traversing great circles \emph{twice} at constant rate 
  
  \textup{(cf. \cref{Cyc1AsAPushout})};

  \item
  \label
  {ScndStageAttachmentEquivariance}
  obtained under $(-)\sslash S^1$ by an $S^1$-equivariant attachment $\inlinetikzcd{L_1 S^4 \ar[r, hook] \& L_2 S^4}$.

  \item
  \label
  {ScndStageAttachmentCofiber}
  The cofiber of $i$ over $B S^1$, the fiberwise Thom space of that disk bundle,
  \begin{equation}
    \mathrm{cof}_{B S^1}(i)
    \simeq
    D^9_B
        \underset
          {\mathclap{B S^1}}
          {/}
    \partial D^9_B
    \mathrlap{\,,}
  \end{equation}
  is fiberwise 8-connected, with the next fiberwise (stable) homotopy groups being:
  \begin{subequations}
    \begin{align}
    \pi_9\bracket({ 
      \bracket({
        D^9_B
        \smash{%
          \underset
            {\mathclap{B S^1}}
            {/}
        }
        \partial D^9_B 
      })_b
    }) 
    &
    \simeq \mathbb{Z}
    \\[+7pt]
    \label
    {pi10OfStabilizedCofiber}
    \pi_{10}\bracket({ 
      \bracket({
        \Sigma^\infty_{B S^1} 
        \bracket({
          D^9_B
          \smash{%
            \underset
              {\mathclap{B S^1}}
              {/}
          }
          \partial D^9_B 
        }) 
      })_b
    }) 
    & \simeq \mathbb{Z}_{/2}    
    \end{align}
  \end{subequations}
  for all $b \in B S^1$.
  \end{enumerate}
\end{lemma}
\begin{proof}
  Before taking the $S^1$-quotient,
  \cref{ScndStageAsAttachmentOf9DiskBundle} is \cite[Thm. 2.4.10 \& Cor. 2.4.11]{Klingenberg1978}, 
  \cref{ScndStageAttachmentOverBaseB} is \cite[Prop. 2.5.3]{Klingenberg1978},
  and \cref{ScndStageAttachmentEquivariance} is a special case of 
  \cite[Lem. 2.4.7]{Klingenberg1978}. 
  
  From this, \cref{ScndStageAttachmentCofiber} follows using that $K_2$ has a CW structure with cells in dimensions $d \in \{0,3,4,7\}$, whence $\bracket({\mathrm{cof}_{B S^1}(i)})_b$ has non-basepoint cells in dimensions $d \in \{9, 12, 13, 16\}$, so that the inclusion of the bottom cell is 11-connected. Therefore $\pi_9\bracket({ \bracket({\mathrm{cof}_{B S^1}(i) })_b }) \simeq \pi_9\bracket({ S^9 }) \simeq \mathbb{Z}$ and $\pi_{10}^s\bracket({ \bracket({\mathrm{cof}_{B S^1}(i) })_b }) \simeq \pi^s_{10}\bracket({ S^9 }) \simeq \pi_1^s\bracket({S^0}) \simeq \mathbb{Z}_{/2}$.
\end{proof}

\begin{theorem}[Classification of lifts from plain to small-cyclified Cohomotopy]
\label[theorem]{Classification-of-lifts}
For connected oriented Cauchy surfaces $X^9$ and a given solid map in \cref{LiftingThroughCyc1}:
\begin{enumerate}
  \item
  \label{TheObstructionClass}
  The obstruction to the existence of a dashed lift in \cref{LiftingThroughCyc1} is the \emph{D0-brane tadpole}, namely the class $[H_3 \wedge F_6] \in H^9_{\mathrm{dR}}\bracket({ X^9 })$. 

  \item 
  \label{TorsorProperty}
  When this obstruction vanishes \textup{($[H_3 \wedge F_6] = 0$)} 
  the set of homotopy classes of lifts is a torsor over:
  \begin{enumerate}
    \item 
    \label{TorsorOfLiftsGenerally}
    generally: stable 8-Cohomotopy $\mathbb{S}^{6 + \rho_2 f_2 }\bracket({ X^9 })$ \cref{StableCohomotopy}, twisted by the M-circle bundle, 

    whose twist class over a 9-manifold is the mod-2 reduction $\rho_2$ of the D6-brane charge $f_2$, regarded as having coefficients in $\pi_2 B \mathrm{GL}_1(\mathbb{S}) \simeq \pi_1(\mathbb{S}) \simeq \mathbb{Z}_{/2}$ \textup{(cf. \cite[\S~1.4]{AndoBlumbergGepnerHopkinsRezk2014})},
    
    \item 
    \label{TorsorOfLiftsOverNoncompactManifolds}
    if $X^9$ is non-compact or has a boundary: $H^8\bracket({ X^9; \mathbb{Z} })$.
  \end{enumerate}
\end{enumerate}
\end{theorem}
\begin{proof}[Proof outline]
  The strategy is to first realize the inclusion of $\mathrm{Cyc}_1 S^4$ as the homotopy fiber, in the relevant degrees, of a map $o$, whence that map is then the universal obstruction class and the loops in its codomain give the torsor of lifts.

  To this end we proceed as follows:
  \begin{enumerate}
  \item
  By \cref{TowardsSmallCyclifiedCohomotopy} \cref{HigherCycnTruncation}, the lift is equivalently from the second Morse--Bott stage, through $\inlinetikzcd{ \mathrm{Cyc}_1 S^4 \ar[r, hook, "{ i }"] \& \mathrm{Cyc}_2\, S^4 }$ \cref{EnergyFiltration}. From 
  \cref{ScndStageInclusion} we know that  
  the stabilization of the cofiber of this inclusion has $\pi_{10} \simeq \mathbb{Z}_{/2}$ \cref{pi10OfStabilizedCofiber}. Hence the 10-truncation of the stabilization is a 2-stage Postnikov system of this form:
  \begin{equation}
  \label{PostnikovForCofi}
   \begin{tikzcd}[
     column sep=35pt
   ]
     \truncation{10}{
       \Sigma^\infty_{B S^1}
       \mathrm{cof}_{B S^1}(i)
     }     
     \ar[
       rr,
     ]
     \ar[d]
     \ar[
       dr,
       phantom,
       "{ \lrcorner }"{pos=.2}
     ]
     &[-15pt]
     &[+50pt]
     \ast
     \ar[d]
     \\
     \truncation{9}{
       \Sigma^\infty_{B S^1}
       \mathrm{cof}_{B S^1}(i)
     }
     \ar[
       r,
       "{
         U
       }",
       "{ \sim }"{swap}
     ]
     &
     \Sigma^9 H \mathbb{Z}
     \ar[
       r,
       "{
         \scaledbracket({
         \mathrm{Sq}^2
         \,+\,
         c \,\cup\, (-)
         }) \circ \rho_2
       }"
      ]
     &
     \Sigma^{11} H \mathbb{Z}_{/2}
     \mathrlap{\,,}
   \end{tikzcd}
  \end{equation}
  where: 
  \begin{enumerate}
  \item 
  the diagram is over $B S^1$, with $\Sigma^k H A$ denoting the constant bundles there,
  \item
  $U$ denotes the Thom class of the Thom space $D^9_B/_{_{B S^1}}\partial D^9_B$, 
  \item
  $c$ is some element of $H^2\bracket({B S^1; \mathbb{Z}_{/2}})$ pulled back to $B$, to be determined next.
\end{enumerate}

\item
In \cref{PostnikovForCofi} the coefficient of $\mathrm{Sq}^2$ is fixed by inspection of the situation over a point of $B S^1$, while the class $c$ is determined by Thom's formula:

$\mathrm{Sq}^2 U = w_2 \cup U$, where the above homotopy pullback \cref{PostnikovForCofi} also enforces $\mathrm{Sq}^2 U + c \cup U = 0$. Therefore $(w_2 - c)\cup U = 0$, and since $(-)\cup U$ is the Thom isomorphism, this implies 
\begin{equation}
  c = w_2
  \mathrlap{\,.}
\end{equation}

\item 
With Morse--Bott theory one finds that $w_2 = \rho_2 f_2^C = \rho_2 c_1(L)$
is (the mod-2 reduction of) the first Chern class of the universal complex line bundle $L$ over $B S^1$. With this, a directly analogous analysis reveals that the 10-truncation of the $(7+L)$-shifted sphere spectrum has the same Postnikov system \cref{PostnikovForCofi}, and hence is equivalent:
\begin{equation}
\label
{IdentifyingThomSpaceWithTwistedSpherical}
  \truncation{10}{
    \Sigma^\infty_{B S^1}
    \mathrm{cof}_{B S^1}(i)
  }
  \simeq
  \truncation{10}{
    \Sigma^{7+L}_{B S^1}
    \,
    \mathbb{S}
  }
  \mathrlap{\,.}
\end{equation}

  \item
   Thereby, we have obtained a cofibration sequence extended to a diagram of this form:
  \begin{equation}
    \begin{tikzcd}[
      column sep=20pt
    ]
      \mathrm{Cyc}_1 S^4
      \ar[r, hook, "{ i }"]
      &
      \mathrm{Cyc}_2 S^4
      \ar[
        rrr,
        downhorup,
        "{ o }"{description}
      ]
      \ar[
        r
      ]
      &
      \mathrm{cof}_{B S^1}(i)
      \ar[
        r,
        "{
          \sim_{10}
        }"
      ]
      &
      \Omega^\infty_{B S^1}
      \truncation{10}{
        \Sigma^\infty_{B S^1}
        \bracket({
          D^9_B
          \smash{
            \underset
              {\mathclap{B S^1}}
              {/}
          }
          \partial D^9_B
        })
      }
      \ar[
        r,
        "{ \sim }"{swap},
        "{        \text{\cref{IdentifyingThomSpaceWithTwistedSpherical}} 
        }"
      ]
      &
      \Omega^\infty_{B S^1}
      \truncation{10}{
        \Sigma^{7 + L}_{B S^1}
        \,
        \mathbb{S}
      }
      \mathrlap{\,.}
    \end{tikzcd}
  \end{equation}
  Dually, by the  Blakers--Massey theorem this exhibits $\mathrm{Cyc}_1 S^4$ as the homotopy fiber of the map $o$ in low dimensions:
  \begin{equation}
  \label{Cyc1S4AsApproximatedByAHomotopyFiber}
    \begin{tikzcd}
      \mathrm{Cyc}_1 S^4
      \ar[r, "{ \sim_{9} }"]
      &
      \mathrm{fib}_{/B S^1}
      \big(
        \mathrm{Cyc}_2 S^4
        \ar[r, "{ o }"]
        &[-10pt]
        \Omega^\infty_{B S^1}
        \truncation
          {10}
          {
           \Sigma^{7+L}_{B S^1}
           \,
           \mathbb{S}
          }
      \big)
      \mathrlap{\,.}
    \end{tikzcd}
  \end{equation}

  With this in hand as advertised, we conclude as intended:
  
  \item 
  The home of the obstruction class in  \cref{TheObstructionClass} follows by analysis of the corresponding Atiyah--Hirzebruch spectral sequence for twisted stable Cohomotopy, which in degree 9 shows the single term $H^9\bracket({X^9; \mathbb{Z}})$ due to 10-truncation of the coefficients. 
  
  This group is $\simeq \mathbb{Z}$ when $X^9$ is closed and oriented (and vanishes for non-compact $X^9$, in which case the obstruction class trivially vanishes, too). Therefore the nontrivial obstruction class is equivalently a generator of  $H^9\bracket({ X^9; \mathbb{R} })$. 
  
  To identify this,  \cref{ObtainingGaussLawsForCycS4} shows that the only candidate closed degree-9 elements (in the Sullivan model, cf. \cref{CEGenerators}) spanning this $\mathbb{R}$-rational cohomology group are:
  \begin{subequations}
  \label{D0TadpoleGeneratorAndCousins}
  \begin{align}
    \label
    {D0TadpoleAsCohomologyGenerator}
    & H_3 \wedge F_6
    \\
    &
    H_3 
      \wedge 
    F_2 
      \wedge 
    F_4
    =
    \mathrm{d}\bracket({
      \tfrac{1}{2}
      F_4 \wedge F_4
    })
    \\
    &
    H_3 
      \wedge 
    F_2 
      \wedge 
    F_2 
      \wedge
    F_2
    =
    \mathrm{d}\bracket({
      F_4 
        \wedge 
      F_2 
        \wedge 
      F_2
    })
  \end{align}
  \end{subequations}
  (while the remaining candidate $F_2 \wedge H_7$ is not closed).
  Of these, only \cref{D0TadpoleAsCohomologyGenerator} is not universally exact, and hence must be the universal obstruction class, whose pullback to $X^9$ is as claimed in \cref{TheObstructionClass}.

  \item 
  Finally, the torsor property of \cref{TorsorProperty} is now a formal consequence. That this is over the group claimed in \cref{TorsorOfLiftsGenerally}, and that it collapses to $H^8$ in \cref{TorsorOfLiftsOverNoncompactManifolds}, follows from the long exact sequence of the $k$-invariant $\bracket({\mathrm{Sq}^2 + \rho_2 f_2 \cup (-)}) \circ \rho_2$ of $\Sigma^{7 + L}\, \mathbb{S}$ discussed above.
  \qedhere
  \end{enumerate}
\end{proof}

\begin{remark}
To highlight some impact of \cref{Classification-of-lifts} on the physics of type IIA:
\begin{enumerate}
  \item
  In the physically realistic situation where space, $X^9$, is non-compact, the  obstruction \cref{TheObstructionClass}
  vanishes automatically, because then $H^9_{\mathrm{dR}}\bracket({ X^9 }) = 0$ (same when $X^9$ has a boundary).

  \item 
  The choice in \cref{TorsorProperty} encodes a topological charge that was missing in 11D: the charge of black D0-branes (sourcing $F_8$)---the electric duals of the more commonly considered magnetic black D6-branes (sourcing $F_2$), which in 11D are incarnated as the circle fibration itself, cf. \cref{DBraneAndSourcedFluxes}.
\end{enumerate}
\end{remark}
Hence, passing to the sector of dimensional reductions along small circles means demanding that dim-reduced classifying maps \emph{and their homotopies} (gauge transformations) factor through the small-cyclification $\mathrm{Cyc}_1 S^4$ \cref{LiftingThroughCyc1}. This makes $\mathrm{Cyc}_1 S^4$ the classifying space for flux quantization of the \emph{small circle reduction} of cohomotopical 11D SuGra. We find: 
\begin{proposition}[The full type IIA Gauss laws from small-cyclified 4-Cohomotopy]
\label[proposition]
{The-full-IIA-EM-Gauss-laws}
The Gauss laws induced, via \cref{AdmissibleFluxQuantizationLaw}, 
by $\mathrm{Cyc}_1 S^4$ \cref{EnergyFiltration} on any $X^9$ are:
\begin{equation}
\label{TheFullGaussLaw}
  \Omega^1_{\mathrm{cl}}
  \bracket({
    X^9;
    \mathfrak{l}
    \mathrm{Cyc}_1 S^4
  })
  \,\simeq\,
  \left\{
  \begin{aligned}
    H_3 & 
    \in \Omega^3_{\mathrm{dR}}\bracket({X^9})
    \\
    H_7 & 
    \in \Omega^7_{\mathrm{dR}}\bracket({X^9})
    \\  
    F_2 & 
    \in \Omega^2_{\mathrm{dR}}\bracket({X^9})
    \\
    F_4 & 
    \in \Omega^4_{\mathrm{dR}}\bracket({X^9})
    \\
    F_6 & 
    \in \Omega^6_{\mathrm{dR}}\bracket({X^9})
    \\
    F_8 & 
    \in \Omega^8_{\mathrm{dR}}\bracket({X^9})
  \end{aligned}
  \,\middle\vert\,
  \def\arraystretch{2}
  \begin{aligned}
    \mathrm{d}\, H_3 & =
    0
    \\
    \mathrm{d}\, H_7 & =
    \tfrac{1}{2} F_4 \wedge F_4 - F_2 \wedge F_6
    \mathclap{\phantom{\bracket({X^9})}}
    \\
    \mathrm{d}\, F_2 & = 0
    \mathclap{\phantom{\bracket({X^9})}}
    \\
    \mathrm{d}\, F_4 & = H_3 \wedge F_2
    \mathclap{\phantom{\bracket({X^9})}}
    \\
    \mathrm{d}\, F_6 & = H_3 \wedge F_4
    \mathclap{\phantom{\bracket({X^9})}}
    \\
    \mathrm{d}\, F_8 & = H_3 \wedge F_6
    \mathclap{\phantom{\bracket({X^9})}}
  \end{aligned}
  \right\}
  \mathrlap{.}
\end{equation}
\end{proposition}
\begin{proof}[Proof outline]
  The Morse attachment description   in \cref{ScndStageInclusion} shows that the inclusion $\inlinetikzcd{ \mathrm{Cyc}_1 S^4 \ar[r, hook] \& \mathrm{Cyc}\, S^4 }$ is 8-connected, hence induces isomorphisms on $\pi_{\leq 7}\bracket({ - }) \otimes \mathbb{Q}$  and on
  $H^{\leq 7}\bracket({-;\mathbb{Q}})$ (the second stage attaches cells of dimensions $\geq 9$).
  At the same time, the Mayer--Vietoris sequence  for the pushout in \cref{Cyc1AsAPushout} shows that in degree 9 we have 
  \begin{equation}
    H^9\bracket({ 
      \mathrm{Cyc}_1 S^4
      ; 
      \mathbb{Q} 
    }) 
      \simeq 
    0
    \mathrlap{\,,}
    \phantom{%
    \mathbb{Q}\bracket\langle{
      H_3 \wedge F_6
     }\rangle    
    }
  \end{equation}
  while \cref{ObtainingGaussLawsForCycS4} gives: 
  \begin{equation}
  \label
  {Rational9CohomologyOfCycS4}
    H^9\bracket({
      \mathrm{Cyc}\, S^4
      ;
      \mathbb{Q}
    })
      \simeq 
    \mathbb{Q}\bracket\langle{
      H_3 \wedge F_6
     }\rangle
    \mathrlap{\,.}
  \end{equation}
  Therefore, the induced morphism of minimal Sullivan models is an isomorphism on dg-subalgebras generated in degrees $\leq 7$. 
  
  But then the model for $\mathrm{Cyc}_1 S^4$ needs to feature an element $\omega_8$ with $\mathrm{d} \, \omega_8 = H_3 \wedge F_6$ \cref{Rational9CohomologyOfCycS4}. Since the differential of the linear span of available monomials of degree 8 does not contain this element:
  \begin{equation}
    \mathrm{d}
    \bracket({
    \mathbb{Q}
    \big\langle
      F_2^4
      ,\;
      F_2^2 F_4
      ,\;
      F_4^2 
      ,\;
      F_2 F_6
    \big\rangle
    })
    =
    \mathbb{Q}
    \big\langle{
      F_2^3 H_3
      ,\;
      F_2 H_3 F_4
    }\big\rangle
    \centernot\ni
    H_3 F_6
    \mathrlap{\,,}
  \end{equation}
  it must be a new generator, $\omega_8 = F_8$, as claimed.
  (Repeating this argument exhibits a further generator $F_{10}$, which becomes visible over 10-manifolds, satisfying $\mathrm{d}\, F_{10} = H_3 \wedge F_8$.)

  Finally, the (independent) \cref{HomotopyOfCyc1S4,HomotopyOfCyc1S4Table} show that 
  \[
    \begin{aligned}
      \pi_8
      \bracket({ 
        \mathrm{Cyc}_1\, S^4
      }) 
      \otimes 
      \mathbb{Q}
      & 
      \simeq 
      \mathbb{Q}
      \\
      \pi_9
      \bracket({ 
        \mathrm{Cyc}_1\, S^4
      }) 
      \otimes 
      \mathbb{Q}
      & 
      \simeq 
      0
      \mathrlap{\,,}
    \end{aligned}
  \] 
  whence $F_8$ must be the only additional Sullivan model generator through degree 9.
\end{proof}
\begin{remark}
\label[remark]{F8Appears}
  Interestingly, this system \cref{TheFullGaussLaw} is finally the \emph{complete} electromagnetic flux content of IIA (cf. \cite[\S~22.1.3]{Ortin2015}), which both:
  \begin{enumerate}
    \item retains the nonlinear $H_7$ from 11D \cref{GaussLawsForCycS4},
    \item includes the expected $F_8$  \cref{NSRRFluxGaussLaw} not obtained from 11D.

   This $F_8$ is the flux sourced by the D0-brane charges which form the torsor of lifts in \cref{Classification-of-lifts} \cref{TorsorProperty}. 
  \end{enumerate}

In fact, \cref{The-full-IIA-EM-Gauss-laws} generalizes to 10-manifolds, where a further flux species $F_{10}$ would appear in \cref{TheFullGaussLaw}, satisfying $\mathrm{d}\, F_{10} = H_3 \wedge F_8$.
\end{remark}

Accordingly, we find that the integral brane charges as seen (cf. \cref{MeasuringDBraneCharge}) in small-cyclified 4-Cohomotopy  subsume all supersymmetric type IIA brane species (cf. \cref{DBraneAndSourcedFluxes}), including the D0, together with a range of torsion charges, in proliferation of the phenomenon \cref{M2ChargeAccordingToHypothesisH}:
\begin{table}[htb]
\caption{%
\label{HomotopyOfCyc1S4Table}%
The low-degree homotopy groups of $\mathrm{Cyc}_1 S^4$ and their generators (\cref{HomotopyOfCyc1S4}). 
Under \emph{Hypothesis H}, $\pi_{(8-p)}({ \mathrm{Cyc}_1 S^4 })$ is the group of near-horizon $p$-brane charges 
(cf. \cref{GeneralFormulaForBraneCharge}) visible after dim-reduction from 11D on a small circle. We see that the {\color{IntBraneColor}non-torsion groups} reflect exactly the expected integral brane species, including the D0 and the NS1, cf.
\cref{RationalHomotopyGroups,DBraneAndSourcedFluxes}. The remaining {\color{FracBraneColor}torsion subgroups} (shown in classical notation for homotopy groups of spheres, cf. \cite{Toda1962}) reflect torsion brane species entailed by charge quantization under \emph{Hypothesis H}. 
}
\centering
\adjustbox{rndfbox=4pt, scale=0.9}{
\begin{tabular}{@{\hspace{0pt}}c@{\hspace{0pt}}}
\adjustbox{rndfbox=4pt}{
$
\begin{tikzcd}[
  column sep=-2pt,
  row sep=0pt,
  /tikz/column 3/.append style={anchor=base west},
  /tikz/column 5/.append style={anchor=base west},
  /tikz/column 7/.append style={anchor=base west},
  /tikz/column 9/.append style={anchor=base west},
  /tikz/column 11/.append style={anchor=base west},
  /tikz/column 13/.append style={anchor=base west},
]
  \pi_1\bracket({\mathrm{Cyc}_1 S^4})
  &\simeq&
  \!\!0
  \\
  \pi_2\bracket({\mathrm{Cyc}_1 S^4})
  &\;\simeq\;&
  \smash{%
  \overbrace{%
  \mathcolor{IntBraneColor}{%
  \mathbb{Z}\langle%
    \gamma_{2}%
  \rangle%
  }%
  }^{\text{D6}}%
  }
  &\phantom{\oplus}&
  \\
  \pi_3\bracket({\mathrm{Cyc}_1 S^4})
  &\simeq&
  &&
  \smash{%
  \overbrace{%
  \mathcolor{IntBraneColor}{%
  \mathbb{Z}_{\phantom{/12}}\langle{%
    \gamma_{3}%
  }\rangle%
  }%
  }^{\text{NS5}}%
  }
  \\
  \pi_4\bracket({\mathrm{Cyc}_1 S^4})
  &\simeq&
  &&
  \mathcolor{FracBraneColor}{%
  \mathbb{Z}_{/2\phantom{1}}\langle
    \gamma_{3}\eta
  \rangle}
  &\oplus&
  \smash{%
  \overbrace{%
  \mathcolor{IntBraneColor}{%
  \mathbb{Z}_{\phantom{/12}}\langle{%
    \gamma_{4}%
  }\rangle%
  }}^{\text{D4}}%
  }
  \\
  \pi_5\bracket({\mathrm{Cyc}_1 S^4})
  &\simeq&
  &&
  \mathcolor{FracBraneColor}{%
  \mathbb{Z}_{/2\phantom{1}}\langle
    \gamma_{3}\eta^2
  \rangle}
  &\oplus&
  \mathcolor{FracBraneColor}{%
  \mathbb{Z}_{/2\phantom{1}}\langle{
    \gamma_{4}\eta
  }\rangle}
  \\
  \pi_6\bracket({\mathrm{Cyc}_1 S^4})
  &\simeq&
  &&
  \mathcolor{FracBraneColor}{%
  \mathbb{Z}_{/12}\langle
    \gamma_{3}\nu'
  \rangle}
  &\oplus&
  \mathcolor{FracBraneColor}{%
  \mathbb{Z}_{/2\phantom{1}}\langle{
    \gamma_{4}\eta^2
  }\rangle}
  &\oplus&
  \smash{%
  \overbrace{%
  \mathcolor{IntBraneColor}{%
  \mathbb{Z}_{\phantom{/24}}\langle{%
    \gamma_{6}%
  }\rangle%
  }}^{\text{D2}}%
  }
  \\
  \pi_7\bracket({\mathrm{Cyc}_1 S^4})
  &\simeq&
  &&
  \mathcolor{FracBraneColor}{%
  \mathbb{Z}_{/2\phantom{1}}\langle
    \gamma_{3} \nu' \eta
  \rangle}
  &\oplus&
  \mathcolor{FracBraneColor}{%
  \mathbb{Z}_{/12}\langle{
    \gamma_{4} \Sigma \nu' 
  }\rangle}
  &\oplus&
  \mathcolor{FracBraneColor}{%
  \mathbb{Z}_{/2\phantom{4}}\langle{
    \gamma_{6}\eta
  }\rangle}
  &\oplus&
  \smash{%
  \overbrace{%
  \mathcolor{IntBraneColor}{%
  \mathbb{Z}_{\phantom{/2}}\langle{%
    \gamma_{7}%
  }\rangle%
  }}^{\text{NS1}}%
  }
  \\
  \pi_8\bracket({\mathrm{Cyc}_1 S^4})
  &\simeq&
  &&
  \mathcolor{FracBraneColor}{%
  \mathbb{Z}_{/2\phantom{1}}\langle
    \gamma_{3} \nu' \eta^2
  \rangle}
  &\oplus&
  \mathcolor{FracBraneColor}{%
  \mathbb{Z}_{/2\phantom{1}}\langle{
    \gamma_{4} \Sigma \nu' \eta
  }\rangle}
  &\oplus&
  \mathcolor{FracBraneColor}{%
  \mathbb{Z}_{/2\phantom{4}}\langle{
    \gamma_{6} \eta^2
  }\rangle}
  &\oplus&
  \mathcolor{FracBraneColor}{%
  \mathbb{Z}_{/2}\langle{
    \gamma_{7}\eta
  }\rangle}
  &\oplus&
  \smash{%
  \overbrace{%
  \mathcolor{IntBraneColor}{%
  \mathbb{Z}_{\phantom{/ }}\langle{%
    \gamma_{8}%
  }\rangle%
  }}^{\text{D0}}%
  }
  \\
  \pi_9\bracket({\mathrm{Cyc}_1 S^4})
  &\simeq&
  &&
  \mathcolor{FracBraneColor}{%
  \mathbb{Z}_{/3\phantom{1}}
  \langle
    \gamma_3 \alpha^2
  \rangle}
  &\oplus&
  \mathcolor{FracBraneColor}{%
  \mathbb{Z}_{/2\phantom{1}}
  \langle
    \gamma_4
    \Sigma \nu' \eta^2
  \rangle}
  &\oplus&
  \mathcolor{FracBraneColor}{%
  \mathbb{Z}_{/24}\langle
    \gamma_6 \nu
  \rangle}
  &\oplus&
  \mathcolor{FracBraneColor}{%
  \mathbb{Z}_{/2}\langle
    \gamma_7 \eta^2
  \rangle}
  &\oplus&
  \mathcolor{FracBraneColor}{%
  \mathbb{Z}_{/2}\langle
    \gamma_8 \eta
  \rangle}
  \\
  {}
  \\
\end{tikzcd}
$
}
\\
\\
where:
$
\left\{
\begin{tikzcd}[
  row sep=0pt,
  column sep=12pt,
  /tikz/column 3/.append style={anchor=base west},
]
  S^n 
    \ar[r, "{ \gamma_n }"]
  &[30pt]
  \mathrm{Cyc}_1 S^4
  &
  \text{a free generator, where it exists,} 
  \\
  S^7 
    \ar[
      r, 
      "{ 
        \gamma_7
        \,=\,
        \gamma_4
        \nu
      }"
    ]
  & 
  \mathrm{Cyc}_1 S^4
  &
  \text{composition with the 
    quaternionic Hopf fibration,
  }
  \\
  S^{3+n} 
    \ar[
      r, 
      "{ 
        \eta \,\defneq\, \eta_{2+n}  
      }"]
  &
  S^{2+n}
  &
  \text{the $n$-fold suspension of the complex Hopf fibration $\eta_2$,}
  \\
  S^{4+n} 
    \ar[
      r, 
      "{ 
        \eta^2 
        \,\defneq\,
        \eta \eta
      }"
    ]
  &
  S^{2+n}
  &
  \text{the composite of two consecutive ones,}
  \\
  S^{7+n} 
    \ar[
      r, 
      "{ 
        \nu 
          \,\defneq\,
        \nu_{4+n}
      }"
    ]
  &
  S^{4+n}
  &
  \text{the $n$-fold suspension of the quaternionic Hopf fibration $\nu_4$}
  \\
  S^{6} 
    \ar[
      r, 
      "{ \nu' }"
    ]
  &
  S^{3}
  &
  \text{a generator of $\pi_6(S^3) \simeq \mathbb{Z}_{/12}$,}
  \\
  S^{7} 
    \ar[r, "{ \Sigma \nu' }"]
  &
  S^{4}
  &
  \text{its suspension,}
  \\
  S^{6+n} 
    \ar[
      r, 
      "{ 
        \alpha 
        \,\defneq\,
        \alpha_1(3+n)
      }"
    ]
  &
  S^{3+n}
  &
  \text{notation for $\Sigma^n(4 \nu')$.}
\end{tikzcd}
\right.
$
\\
\\
For instance:
$
  \begin{tikzcd}
    S^8
    \ar[
      r, "{ \Sigma^5 \eta_2 }"
    ]
    \ar[
      rrr,
      uphordown,
      "{ 
        \gamma_4 \Sigma \nu' \eta 
      }"{description}
    ]
    &
    S^7
      \ar[r, "{\Sigma \nu'}"]
    &
    S^4 
     \ar[r, "{ \gamma_4 }"]
    &
    \mathrm{Cyc}_1 S^4
  \end{tikzcd}
  ,
  \quad
  \begin{tikzcd}
    S^6
    \ar[
      rrr,
      uphordown,
      "{ 
        \gamma_4 \eta^2 
      }"{description}
    ]
    \ar[r, "{ \Sigma^3 \eta_2 }"]
    &
    S^5
    \ar[r, "{ \Sigma^2 \eta_2 }"]
    &
    S^4
    \ar[
      r,
      "{ \gamma_4 }"
    ]
    &
    \mathrm{Cyc}_1 S^4.
  \end{tikzcd}
$
\end{tabular}
}
\end{table}
\begin{proposition}
\label[proposition]{HomotopyOfCyc1S4}
The homotopy groups of $\mathrm{Cyc}_1 S^4$ in degrees $\leq 9$ are those of a product of spheres\footnotemark\ with $B S^1$,
\begin{equation}
  \label
  {LowHomotopyOfCyc1AndProductOfSpheres}
    \forall_{k \leq 9}
    \;\;\;
    \pi_k\bracket({
      \mathrm{Cyc}_1 S^4
    })
    \simeq
    \grayunderbrace{
    \pi_k\bracket({
      B S^1
    })}{\mathrm{D6}}
    \oplus
    \grayunderbrace{
    \pi_k\bracket({
      S^4
    })}{ \mathrm{D4}, \mathrm{NS1} }
    \oplus
    \grayunderbrace{
    \pi_k\bracket({
      S^3
    })}{ \mathrm{NS5} }
    \oplus
    \grayunderbrace{
    \pi_k\bracket({
      S^6
    })
    }{ \mathrm{D2} }
    \oplus
    \grayunderbrace{
    \pi_k\bracket({
      S^8
    })}{ \mathrm{D0} }
    \mathrlap{\,,}
\end{equation}
\textup{as made explicit in \cref{HomotopyOfCyc1S4Table}}.
\end{proposition}
\footnotetext{%
  The lead-in slogan of \cite{SS23-Mf} was that \emph{under Hypothesis H, M-brane charges organize into homotopy groups of spheres}. Curiously, \cref{HomotopyOfCyc1S4} shows that under small circle reduction this remains the case for type IIA brane charges, with a great variety of integral and torsion brane charges falling into the classical tables of homotopy groups of spheres (cf. \cite[p. 186]{Toda1962}). Notice how the D4- and NS1-charges are jointly classified by an $S^4$ in \cref{LowHomotopyOfCyc1AndProductOfSpheres},  the direct dimensional reduction of how the M5 and M2 charges are jointly classified by $S^4$ in \cref{M5ChargeAccordingToHypothesisH,M2ChargeAccordingToHypothesisH}.
  
  Beware that \cref{HomotopyOfCyc1S4} does \emph{not} claim that $\truncation{9}{\mathrm{Cyc}_1 S^4}$ is equivalent to $\truncation{9}{B S^1 \times S^4 \times S^3 \times S^6 \times S^8}$, which is indeed far from the case.}
\begin{proof}
  Since $B S^1$ has no homotopy above degree 2, and since the evaluation map on $L_1 S^4$ has a section given by constant loops, the computation immediately reduces to that of the homotopy of small based loop space $\Omega_1 S^4 := \mathrm{fib}({\inlinetikzcd{ L_1 S^4 \ar[r, "{ \mathrm{ev}_0 }"] \& S^4 }})$. As in the proof of  \cref{Cyc1AsAPushout}, this is equivalently the space of possibly degenerate round circles on $S^4$ which are based at a fixed point $p \in S^4$.  Such a round circle is exactly determined by its direction at $p$ and its center, and is a collapsed circle when the center is at $p$ itself, so that:
  \begin{equation}
    \label{ComputingBasedLoopsOfS4}
    \begin{aligned}
    \Omega_1 S^4
    &
    \simeq
    \frac
      { S^3 \times S^3 }
      { S^3 \times \{p\}  }
    \\
    & \simeq
    \bracket({S^3})_+
    \wedge
    S^3
    \\
    & \simeq
    \Sigma^3\bracket({
      S^3
      \vee
      S^0
    })
    \\
    & \simeq
    S^6 \vee S^3
    \mathrlap{\,.}
    \end{aligned}
  \end{equation}
  From this, the Hilton--Milnor theorem (cf. \cite[Ch. XI \S~6]{Whitehead1978}) gives 
  \begin{equation}
    \forall_{k \leq 9}
    \;\;
    \pi_{k}
    \bracket({
      \Omega_1 S^4
    })
    \simeq
    \pi_k\bracket({
      S^3
    })
    \oplus
    \pi_k\bracket({
      S^6
    })
    \oplus
    \pi_k\bracket({
      S^8
    })
    \mathrlap{\,.}
  \end{equation}
  In summary, this proves the claim \cref{LowHomotopyOfCyc1AndProductOfSpheres}; hence the table \cref{HomotopyOfCyc1S4Table} follows from the classical tabulation of these low degree homotopy groups of spheres (cf. \cite[p. 186]{Toda1962}).
\end{proof}

\begin{table}[htb]
\caption{%
\label{RationalHomotopyGroups}%
The low degree rationalized homotopy groups (discarding torsion) of the small-cyclification of $S^4$ (computed via Hilton--Milnor as in the proof of \cref{HomotopyOfCyc1S4}, cf. \cref{HomotopyOfCyc1S4Table}),
versus those of the classifying space for rank-0 twisted K-theory. 
In the range corresponding to IIA brane charges the only difference is in degree 7, corresponding to the NS1-brane charge which is missed by twisted K-theory (cf. \cref{AbelianCohDoesNotQuantizeNonlinear}). The higher homotopy groups correspond to higher global symmetries \cite[\S~4.2.1, p. 32]{SS26-HigherGauge}, and here the predictions increasingly diverge with $k$, cf. \cref{HigherGlobalSymmetries} in \cref{Vistas}.
}
\adjustbox{
  scale=.95,
  rndfbox=4pt
}
{%
\def\arraystretch{1.4}%
\def\tabcolsep{3pt}%
\begin{tabular}
{@{\hspace{0pt}}c||ccccccccc|ccccccccccc@{\hspace{0pt}}}
  $p = (8-k)$-brane
  & 
  \small{$-$}
  & 
  \small D6
  &
  \small NS5
  &
  \small D4
  &
  \small{$-$}
  &
  \small D2
  &
  \small NS1
  &
  \small D0
  &
  \small $-$
  &
  \multicolumn{11}{c}{
    Higher global symmetries
  }
  \\
  \hline
  $k$ 
  &
  $1$
  &
  $2$
  &
  $3$
  &
  $4$
  &
  $5$
  &
  $6$
  &
  $7$
  &
  $8$
  &
  $9$
  &
  $10$
  &
  $11$
  &
  $12$
  &
  $13$
  &
  $14$
  &
  $15$
  &
  $16$
  &
  $17$
  &
  $18$
  &
  $19$
  &
  $20$
  \\
  \hline
  \hline
  $\pi_k\bracket({
    \mathrm{Cyc}_1 S^4
  })
  \!\otimes\! 
  \mathbb{Q}$
  &
  $0$
  &
  $\mathbb{Q}$
  &
  $\mathbb{Q}$
  &
  $\mathbb{Q}$
  & 
  $0$
  &
  $\mathbb{Q}$
  &
  $\mathcolor{purple}{\mathbb{Q}}$
  &
  $\mathbb{Q}$
  &
  $0$
  &
  $\mathbb{Q}$
  &
  $\mathcolor{purple}{\mathbb{Q}}$
  &
  $\mathbb{Q}$
  &
  $\mathcolor{purple}{\mathbb{Q}}$
  &
  $\mathbb{Q}$
  &
  $\mathcolor{purple}{\mathbb{Q}^2}$
  &
  $\mathbb{Q}$
  &
  $\mathcolor{purple}{\mathbb{Q}^2}$
  &
  $\mathcolor{purple}{\mathbb{Q}^2}$
  &
  $\mathcolor{purple}{\mathbb{Q}^3}$
  &
  $\mathcolor{purple}{\mathbb{Q}^3}$
  \\
  \hline
  $\pi_k\bracket({
    B\mathrm{U}
    \sslash 
    B \mathrm{U}(1)
  })
  \!\otimes\! 
  \mathbb{Q}\,$
  &
  $0$
  &
  $\mathbb{Q}$
  &
  $\mathbb{Q}$
  &
  $\mathbb{Q}$
  &
  $0$
  &
  $\mathbb{Q}$
  &
  $\mathcolor{purple}{0}$
  &
  $\mathbb{Q}$
  &
  $0$
  &
  $\mathbb{Q}$
  &
  $\mathcolor{purple}{0}$
  &
  $\mathbb{Q}$
  &
  $\mathcolor{purple}{0}$
  &
  $\mathbb{Q}$
  &
  $\mathcolor{purple}{0}$
  &
  $\mathbb{Q}$
  &
  $\mathcolor{purple}{0}$
  &
  $\mathcolor{purple}{\mathbb{Q}}$
  &
  $\mathcolor{purple}{0}$
  &
  $\mathcolor{purple}{\mathbb{Q}}$
\end{tabular}
}
\end{table}

This match with the type IIA flux and brane content (\cref{The-full-IIA-EM-Gauss-laws,HomotopyOfCyc1S4})  suggests that over 9-manifolds the nonabelian cohomology \cref{NonabelianCohomology} classified by $\mathrm{Cyc}_1 S^4$ may be an M-theoretic nonabelian deformation of twisted K-theory. We check this next.

\subsection
{K-Theory Emerging}
\label
{KTheoryEmerging}

Before we state the main result (\cref{Comparison_map_to_twisted_K_theory,ComparisonInDMWSector}), we identify the low-degree universal brane charge cohomology classes (cf. \cref{MeasuringDBraneCharge}) on the classifying spaces under consideration (\cref{BraneChargeFromCohomotopy,BraneChargesFromTwistedKTheory}). The close similarity between these constructions on both sides, both controlled by maps to $\mathrm{Cyc}\, B^4 \mathbb{Z}$
(cf. \cref{LEMtau4BUModBU1IsCycB4Z}), is what drives the following \cref{Comparison_map_to_twisted_K_theory}.

\begin{definition}[Extracting brane charges from small-cyclified Cohomotopy]
\label[definition]{BraneChargeFromCohomotopy}
The universal
\begin{enumerate}
\item D4-brane charge, $f_4^C$,
\item NS5-brane charge, $h_3^C$, 
\item D6-brane charge, $f_2^C$,
\end{enumerate}
as seen in small-cyclified 4-Cohomotopy \cref{SmallCyclifiedCohomotopy}, are projected out via the following homotopy-commutative diagram (purple labels indicate composite maps defined thereby):
  \begin{equation}
  \label
  {ProjectingBraneChargesOutOfSmallCycCohomotopy}
    \begin{tikzcd}[
      row sep=20pt,
      column sep=30pt
    ]
      &
      L_1 S^4
      \ar[
         rrr,
         uphordown,
         "{ 
           \mathcolor{purple}{
             \scaledbracket({
               f_4^C,
               \,
               h_3^C
              })
           } 
         }"
      ]
      \ar[d]
      \ar[r, hook]
      &[-30pt]
      L\, S^4
      \ar[d]
      \ar[
        r,
        "{ L\, \iota_4 }"
      ]
      &
      L\, B^4 \mathbb{Z}
      \ar[
        r,
        "{ \sim }"
      ]
      \ar[d]
      &
      B^4 \mathbb{Z}
      \times
      B^3 \mathbb{Z}
      \ar[d]
      \\
      &
      \mathrm{Cyc}_1 S^4
      \ar[
        rr,
        bend right=15,
        "{
          \mathcolor{purple}{\phi^C}
        }"{description}
      ]
      \ar[
        r,
        hook
      ]
      \ar[
        d,
        "{ 
          \mathcolor{purple}{
            \scaledbracket({f_2^C, \, h_3^C})
          }
        }"'
      ]
      &
      \mathrm{Cyc}\, S^4
      \ar[
        r,
        "{
          \mathrm{Cyc}\,
          \iota_4
        }"
      ]
      &
      \mathrm{Cyc}\, B^4 \mathbb{Z}
      \ar[
        d,
        ->>,
      ]
      \ar[
        r
      ]
      \ar[
        dr,
        phantom,
        "{ \lrcorner }"{pos=.1}
      ]
      \ar[
        dr,
        phantom,
        "{
           \text{\cref{CycOfBnZAsHomotopyFiber}}
        }"{scale=.7}
      ]
      &
      \ast
      \ar[d]
      \\
      &
      B^2 \mathbb{Z}
      \times
      B^3 \mathbb{Z}
      \ar[
        rr,
        equals
      ]
      &&
      B^2 \mathbb{Z}
      \times
      B^3 \mathbb{Z}
      \ar[
        r,
        "{ \cup }"{description}
      ]
      &
      B^5 \mathbb{Z}
      \mathrlap{\,,}
    \end{tikzcd}
  \end{equation}
where the top rectangle is formed by homotopy pullback of the bottom rectangle along the map
$
  \inlinetikzcd{
    (0,\mathrm{id})
    :
    B^3 \mathbb{Z}
    \ar[r]
    \& B^2 \mathbb{Z} \times B^3 \mathbb{Z}
  }
$,
imposing vanishing of D6-brane charge, $f_2^C = c_1^C \overset{!}{=} 0$.
In the resulting diagram \cref{ProjectingBraneChargesOutOfSmallCycCohomotopy}:
\begin{enumerate}
\item 
the homotopy filling the bottom rectangle is a trivialization, $f_4^C$, of the cup product $f_2^C \cup h_3^C$ in integral cohomology; this is the integral refinement of the Gauss law $\mathrm{d}\, F_4 = H_3 \wedge F_2$.

\item Consequently, the top row is the integral lift $f_4^C$ of $F_4$ where it is closed due to vanishing D6-brane charge and $F_2 = 0$.
\end{enumerate}  
\end{definition}

We will see in \cref{BraneChargesFromTwistedKTheory} that a directly analogous extraction of these brane charges exists in twisted K-theory. For that we need the following result:
\begin{lemma}
\label[lemma]{LEMtau4BUModBU1IsCycB4Z}
  We have homotopy equivalences $\phi$ of this form:
  \begin{equation}
  \label{tau4BUModBU1IsCycB4Z}
  \begin{tikzcd}[
    column sep=10pt,
    row sep=20pt
  ]
    \truncation
      {4}
      { B \mathrm{U} }
    \sslash 
    B \mathrm{U}(1)
    \ar[
      dr,
      "{
        ({
          c_1
          ,\,
          \tau
        })
      }"{swap, pos=.5}
    ]
    \ar[
      rr,
      dashed,
      "{ \sim }"',
      "{ \phi }"
    ]
    &&
    \mathrm{Cyc}\, B^4 \mathbb{Z}
    \ar[
      dl,
      "{
        \text{\cref{CycOfBnZAsHomotopyFiber}}
      }"{pos=.5, sloped},
      "{
        (f_2, h_3)
      }"{pos=.4}
    ]
    \\
    & 
    B^2 \mathbb{Z}
    \times
    B^3 \mathbb{Z}
    \mathrlap{\,,}
  \end{tikzcd}
  \end{equation}
  whose homotopy classes over $B^2 \mathbb{Z} \times B^3 \mathbb{Z}$, for fixed fiber orientation, form exactly a $\mathbb{Z}$-torsor.
\end{lemma}
For the following discussion we pick once and for all any one of these equivalences \cref{tau4BUModBU1IsCycB4Z}.
\begin{proof}
\begin{enumerate}
\item
  The key is to observe that both sides of the triangle \cref{tau4BUModBU1IsCycB4Z} are $B^4 \mathbb{Z}$-principal bundles, hence classified (\parencites{DrorDwyerKan1980}, cf. \cite{NSS2015a}) by elements in 
  \begin{equation}
    \label{CohomologyOfB2ZTimesB3Z}
    H^1\bracket({
      B^2 \mathbb{Z}
      \times
      B^3 \mathbb{Z}
      ;\,
      B^4 \mathbb{Z}
    })
    \simeq
    H^5\bracket({
      B^2 \mathbb{Z}
      \times
      B^3 \mathbb{Z}
      ;\,
      \mathbb{Z}
    })
    \simeq
    \mathbb{Z}\langle
      \iota_2 \cup \iota_3
    \rangle
    \mathrlap{\,,}
  \end{equation}
  where the isomorphism on the right is the K{\"u}nneth theorem.
  \begin{enumerate}
  \item 
  For the right side this statement is \cref{PastingDiagramForCyc}, which also shows that the equivalence class here is the generator $\iota_2 \cup \iota_3$ in \cref{CohomologyOfB2ZTimesB3Z}.
  \item   
  For the left side we first observe that the homotopy fiber is
  \begin{equation}
    \truncation{4}{B \mathrm{SU}}
    \simeq
    B \truncation{3}{\mathrm{SU}}
    \simeq
    B \truncation{3}{\mathrm{SU}(2)}
    \simeq
    B \truncation{3}{S^3}
    \simeq
    B B^3 \mathbb{Z}
    \simeq
    B^4 \mathbb{Z}
    \mathrlap{\,.}
  \end{equation}

  Now, general $B^4 \mathbb{Z}$-fibrations are classified by maps into $B \mathrm{Aut}\bracket({B^4 \mathbb{Z}})$, but a general argument shows that over simply connected base spaces these factor through $B B^4 \mathbb{Z} \simeq B^5 \mathbb{Z}$: 
  Namely we observe that  
  \begin{equation}
    \pi_0
    \,
    \mathrm{Map}\bracket({ 
      B^4 \mathbb{Z}
      ,\, 
      B^4 \mathbb{Z} 
    })
    \simeq
    H^4\bracket({
      B^4 \mathbb{Z}
      ;\,
      \mathbb{Z}
    })
    \simeq
    \mathbb{Z}
    \mathrlap{\,,}
  \end{equation}
  labeled by the induced endomorphism on $\pi_4\bracket({ B^4 \mathbb{Z} }) \simeq \mathbb{Z}$, and
  that all these components are homotopy-equivalent to each other (by the group structure on $B^4 \mathbb{Z}$). The components corresponding to $\mathrm{Aut}\bracket({
    B^4 \mathbb{Z}
  })$ are hence those labeled by $n = \pm 1$, while the component labeled by $n = 0$ manifestly (and hence any other component $n$, too) has 
  \begin{equation}
    \pi_{k \geq 1} 
    \, 
    \mathrm{Map}_n\bracket({
      B^4 \mathbb{Z}
      ,\,
      B^4 \mathbb{Z}
    })
    \simeq
    H^{4-k}\bracket({
      B^4 \mathbb{Z}
      ;\,
      \mathbb{Z}
    })
    \simeq
    \begin{cases}
      \mathbb{Z}
      &
      \text{ if } k = 4
      \\
      0 & \text{ otherwise, }
    \end{cases}
  \end{equation}
  and is hence a $B^4 \mathbb{Z}$ itself. This means that $B \mathrm{Aut}\bracket({B^4 \mathbb{Z}}) \simeq B \bracket({ \mathbb{Z}_{/2} \ltimes B^4 \mathbb{Z} })$, with the higher connected cover visible over a simply connected base being $B B^4 \mathbb{Z} \simeq B^5 \mathbb{Z}$, classifying $B^4 \mathbb{Z}$-principal fibrations, as promised.
    
  With \cref{CohomologyOfB2ZTimesB3Z} this shows that the equivalence class of the $B^4 \mathbb{Z}$ fibration on the left of \cref{tau4BUModBU1IsCycB4Z} corresponds to a single integer $m \in \mathbb{Z}$ in
  \begin{equation}
    \label{tau4BUModBU1AsHomotopyFiber}
    \begin{tikzcd}[
      column sep=50pt
    ]
      \truncation{4}{B\mathrm{U}}
      \sslash
      B\mathrm{U}(1)
      \ar[
        d,
        "{
          (c_1, \tau)
        }"'
      ]
      \ar[r]
      \ar[
        dr,
        phantom,
        "{ \lrcorner }"{pos=.1}
      ]
      &
      \ast
      \ar[d, "{0}"]
      \\
      B^2 \mathbb{Z}
      \times
      B^3 \mathbb{Z}
      \ar[
        r,
        "{
          m 
            \,\cdot\, 
          \iota_2 \cup \iota_3
        }"
      ]
      &
      B^5 \mathbb{Z}
      \mathrlap{\,.}
    \end{tikzcd}
  \end{equation}
  \end{enumerate}

  \item 
  Thereby the proof is reduced to showing that $m = \pm 1$ in \cref{tau4BUModBU1AsHomotopyFiber}.
  \begin{enumerate}
  \item
  To this end, pull back the left fibration in \cref{tau4BUModBU1IsCycB4Z} to $B^2 \mathbb{Z} \times S^3$ along a generator $\inlinetikzcd{S^3 \ar[r, "{ e }"] \& B^3 \mathbb{Z}}$. Understanding $S^3$ as the gluing of a pair of hemispheres over the equatorial $S^2$, the gerbe $e^\ast \tau$ is encoded by its transition line bundle with Chern class a generator $\ell \in H^2\bracket({S^2; \mathbb{Z}})$. Under the action of $B \mathrm{U}(1)$ on $B \mathrm{U}$ by tensoring of represented rank-0 vector bundles $V$ with represented line bundles $L$, this in turn makes the transition class of the $B^4 \mathbb{Z}$ principal bundle over $B^2 \mathbb{Z} \times S^2$ be
  \begin{equation}
    c_2\bracket({
      V \otimes L
    })
    -
    c_2\bracket({ V })
    =
    - c_1(V) \cup c_1(L)
    =
    (- c_1 \cup \ell)(V)
    \mathrlap{\,.}
  \end{equation}

  \item 
  Finally, the actual class of the fibration is the image of this transition class under the connecting homomorphism $\delta$ in the Mayer--Vietoris sequence:
  \begin{equation}
    \begin{tikzcd}[
      sep=0pt
    ]
      H^4\bracket({
        B^2 \mathbb{Z}
          \times
        S^2
        ;
        \mathbb{Z}
      })
      \ar[
        rr,
        "{ \delta }"
      ]
      &&
      H^5\bracket({
        B^2 \mathbb{Z}
          \times
        S^3
        ;
        \mathbb{Z}
      })
      \\[20pt]
      \mathbb{Z}\bracket\langle{
        c_1^2
      }\rangle
      \oplus
      \mathbb{Z}\bracket\langle{
        c_1 \cup \ell
      }\rangle
      \ar[
        u, 
        "\sim"{sloped}
      ]
      \ar[
        rr
      ]
      &&
      H^5\bracket({
        B^2 \mathbb{Z}
          \times
        B^3 \mathbb{Z}
        ;
        \mathbb{Z}
      }) 
      \ar[
        u,
        "{
          (\mathrm{id},e)^\ast
        }"{swap},
        "{ \sim }"{sloped}
      ]
      \\
      c_1 \cup \ell
      &\longmapsto&
      c_1 \cup \iota_3
      \mathrlap{\,.}
    \end{tikzcd}
  \end{equation}
  This shows that $m = -1$ in \cref{tau4BUModBU1AsHomotopyFiber} and thus completes the proof.
  \qedhere
  \end{enumerate}
  \end{enumerate}
  \end{proof}

With \cref{LEMtau4BUModBU1IsCycB4Z} in hand, we obtain the K-theoretic analog of \cref{BraneChargeFromCohomotopy}:
\begin{definition}
[Brane charges from twisted K-theory]
\label[definition]
{BraneChargesFromTwistedKTheory}
The universal 
\begin{enumerate}
\item D4-brane charge, $f_4^K$,
\item NS5-brane charge, $h_3^K$, 
\item D6-brane charge, $f_2^K$,
\end{enumerate}
as seen in twisted K-theory, are projected out via the following homotopy-commutative diagram, analogous to \cref{ProjectingBraneChargesOutOfSmallCycCohomotopy}:
\begin{equation}
\label{DiagramOfBraneChargesInKTheory}
  \begin{tikzcd}[
    column sep=15pt,
  ]
    B \mathrm{SU}
    \sslash 
    B \mathrm{U}(1)
    \ar[d]
    \ar[
      rrr,
      "{
        \mathcolor{purple}{
          \scaledbracket({
            f_4^K
            ,\,
            h_3^K
          })
        }
      }"
    ]
    &&&
    B^4 \mathbb{Z}
    \times
    B^3 \mathbb{Z}
    \ar[d]
    \\
    B \mathrm{U}
    \sslash 
    B \mathrm{U}(1)
    \ar[
      rr,
      bend right=15pt,
      "{ 
        \mathcolor{purple}{\phi^K} 
      }"{description}
    ]
    \ar[
      d,
      "{
        \mathcolor{purple}{
        \scaledbracket({
          f_2^K
          ,\,
          h_3^K
        })
        }
      }"{swap}
    ]
    \ar[
      r,
      ->>
    ]
    &
    \truncation
      {4}
      { B \mathrm{U} }
    \sslash 
    B \mathrm{U}(1)
    \ar[
      r,
      "{ \sim }"',
      "{
        \text{\cref{tau4BUModBU1IsCycB4Z}}
      }"
    ]
    &
    \mathrm{Cyc}\, B^4 \mathbb{Z}
    \ar[
      d,
      ->>
    ]
    \ar[r]
    \ar[
      dr,
      phantom,
      "{ \lrcorner }"{pos=.1}
    ]
    \ar[
      dr,
      phantom,
      "{
         \text{\cref{CycOfBnZAsHomotopyFiber}}
      }"{scale=.7}
    ]
    &
    \ast
    \ar[d]
    \\
    B^2 \mathbb{Z}
    \times
    B^3 \mathbb{Z}
    \ar[
      rr,
      equals
    ]
    &&
    B^2 \mathbb{Z}
    \times
    B^3 \mathbb{Z}
    \ar[
      r,
      "{ \cup }"{description}
    ]
    &
    B^5 \mathbb{Z}
    \mathrlap{\,,}
  \end{tikzcd}
\end{equation}
where, just as in \cref{ProjectingBraneChargesOutOfSmallCycCohomotopy}, the top rectangle is formed by homotopy pullback of the bottom rectangle along the map
$
  \inlinetikzcd{
    (0,\mathrm{id})
    :
    B^3 \mathbb{Z}
    \ar[r]
    \& B^2 \mathbb{Z} \times B^3 \mathbb{Z}
  }
$,
imposing vanishing of D6-brane charge, $f_2^K  \overset{!}{=} 0$.
\end{definition}

In summary of these constructions, we may identify $\mathrm{Cyc}\, B^4 \mathbb{Z}$ as the classifying object for triples $(f_2, h_3, f_4)$ of brane charges:
\begin{equation}
\label
{BraneChargesClassifiedByCycC4Z}
    \begin{tikzcd}[
      row sep=20pt, column sep=large
    ]
      L\, B^4 \mathbb{Z}
      \ar[
        r,
        "{
          (f_4, h_3)
        }",
        "{ \sim }"{swap}
      ]
      \ar[d]
      &
      B^4 \mathbb{Z}
      \times
      B^3 \mathbb{Z}
      \ar[d]
      \\
      \mathrm{Cyc}\, B^4 \mathbb{Z}
      \ar[
        d,
        ->>,
        "{
          (f_2, h_3)
        }"{swap}
      ]
      \ar[
        r
      ]
      \ar[
        dr,
        phantom,
        "{ \lrcorner }"{pos=.1}
      ]
      \ar[
        dr,
        phantom,
        "{
           \text{\cref{CycOfBnZAsHomotopyFiber}}
        }"{scale=.7}
      ]
      &
      \ast
      \ar[d]
      \\
      B^2 \mathbb{Z}
      \times
      B^3 \mathbb{Z}
      \ar[
        r,
        "{ \cup }"{description}
      ]
      &
      B^5 \mathbb{Z}
      \mathrlap{\,.}
    \end{tikzcd}
\end{equation}

\begin{lemma}[Group structure on brane charges]
\label[lemma]{GroupStructureOnBraneCharges}
  The projection
  $
    \inlinetikzcd{
      \mathrm{Cyc}\, B^4 \mathbb{Z}
      \ar[
        r,
        "{
          h_3
        }"
      ]
      \&
      B^3 \mathbb{Z}
    }
  $
  \cref{BraneChargesClassifiedByCycC4Z}
  is an abelian $\infty$-group object over $B^3 \mathbb{Z}$, with fiberwise $n$-fold multiplication map, for $n \in \mathbb{Z}$, acting on brane charges  by:
  \begin{equation}
    \begin{tikzcd}[
      column sep=5pt,
      row sep=5pt
    ]
      \mathrm{Cyc} B^4 \mathbb{Z}
      \ar[
        dr,
        "{ h_3 }"'
      ]
      \ar[
        rr,
        "{ n \cdot  }"
      ]
      &&
      \mathrm{Cyc} B^4 \mathbb{Z}
      \mathrlap{\,,}
      \ar[
        dl,
        "{ h_3 }"
      ]
      \\
      & B^3 \mathbb{Z}
    \end{tikzcd}
    \qquad 
    \begin{tikzcd}[
      column sep=5pt,
      row sep=5pt
    ]
      \mathrm{Cyc} B^4 \mathbb{Z}
      \ar[
        dr,
        "{ 
          n \cdot f_2 
        }"'
      ]
      \ar[
        rr,
        "{ n \cdot  }"
      ]
      &&
      \mathrm{Cyc}  B^4 \mathbb{Z}
      \mathrlap{\,,}
      \ar[
        dl,
        "{ f_2 }"
      ]
      \\
      & B^2 \mathbb{Z}
    \end{tikzcd}
    \qquad 
    \begin{tikzcd}[
      column sep=5pt,
      row sep=5pt
    ]
      L B^4 \mathbb{Z}
      \ar[
        dr,
        "{ 
          n \cdot f_4 
        }"'
      ]
      \ar[
        rr,
        "{ n \cdot  }"
      ]
      &&
      L B^4 \mathbb{Z}
      \mathrlap{\,.}
      \ar[
        dl,
        "{ f_4 }"
      ]
      \\
      & B^4 \mathbb{Z}
    \end{tikzcd}
  \end{equation}
\end{lemma}
\begin{proof}
  By \cref{CycOfBnZ}, $\mathrm{Cyc}\, B^4 \mathbb{Z}$ is the fiber of the cup product $\inlinetikzcd{ B^2 \mathbb{Z} \times B^3 \mathbb{Z} \ar[r, "{ \cup }"] \& B^5 \mathbb{Z} }$. Regarded in the slice over $B^3 \mathbb{Z}$ this means that
  \begin{equation}
  \label
  {CycB^4ZAsFiberOverB3Z}
    \begin{tikzcd}
      \mathrm{Cyc}\, B^4 \mathbb{Z}
      \ar[
        rr,
        "{
          \mathrm{fib}_{/B^3 \mathbb{Z}}
        }"
      ]
      \ar[
        d,
        "{ h_3 }"
      ]
      &&
      B^2 \mathbb{Z}
        \times
      B^3 \mathbb{Z}
      \ar[
        rr,
        "{
          (f_2, h_3)
          \,\mapsto\,
          ({
            f_2 \cup h_3
            ,\,
            h_3
          })
        }"
      ]
      \ar[
        d,
        "{ \mathrm{pr}_2 }"
      ]
      &{\phantom{---}}&
      B^5 \mathbb{Z}
        \times
      B^3 \mathbb{Z}
      \ar[
        d,
        "{ \mathrm{pr}_2 }"
      ]
      \\
      B^3 \mathbb{Z}
      \ar[rr, equals]
      &&
      B^3 \mathbb{Z}
      \ar[rr, equals]
      &&
      B^3 \mathbb{Z}
      \mathrlap{\,,}
    \end{tikzcd}
  \end{equation}
  where $\mathrm{Cyc}\, B^4 \mathbb{Z}$ is 
  the pullback of the top right map along the inclusion $\inlinetikzcd{ B^3 \mathbb{Z} \ar[rr, "{ h_3 \mapsto (0,h_3) }"] \&\phantom{-}\& B^5 \mathbb{Z} \times B^3 \mathbb{Z} }$.
  But the top right map in \cref{CycB^4ZAsFiberOverB3Z} is also the image under $\Omega^\infty_{B^3 \mathbb{Z}}$ of a map of constantly parameterized $H \mathbb{Z}$-module spectra, and since $\Omega^\infty_{B^3 \mathbb{Z}}$ preserves limits over $B^3 \mathbb{Z}$, hence so is $\mathrm{Cyc}\, B^4 \mathbb{Z}$ as an object over $B^3 \mathbb{Z}$:
  \begin{equation}
    \left(
    \begin{array}{c}
     \mathrm{Cyc}\, B^4 \mathbb{Z}
     \\
     \downarrow
     \\
     B^3 \mathbb{Z}
    \end{array}
    \right)
    \,\simeq\,
    \Omega^\infty_{B^3 \mathbb{Z}}
    \,
    \mathrm{fib}
    \,
    \Big(
    \begin{tikzcd}[sep=40pt]
      \mathrm{cnst}_{B^3 \mathbb{Z}}
      \,
      \Sigma^2 H\mathbb{Z}
      \ar[
        r,
        "{
          (-) \cup h_3
        }"
      ]
      &
      \mathrm{cnst}_{B^3 \mathbb{Z}}
      \,
      \Sigma^5 H\mathbb{Z}
    \end{tikzcd}
    \Big)
    \mathrlap{\,.}
  \end{equation}
  Here the fiber is again a parameterized $H \mathbb{Z}$-module, so its fiberwise infinite-loop space is an abelian $\infty$-group object over $B^3 \mathbb{Z}$, as claimed. Multiplication by $n$ fixes the base coordinate $h_3$ and scales the projection $f_2$ by $n$. On the homotopy fiber of $f_2$, the remaining coordinate $f_4$ comes from $\Sigma^4 H\mathbb{Z}$ and hence is likewise scaled by $n$, as claimed.
\end{proof}

Now we may state our first main result:
\begin{theorem}[Comparison from small-cyclified Cohomotopy to twisted K-theory]
\label[theorem]{Comparison_map_to_twisted_K_theory}
\begin{enumerate}
\item
\label{TheComparisonMapAtEvenN}
A twisted cohomology operation \cref{TwistedCohomologyOperation}
from small-cyclified 4-Cohomotopy \cref{SmallCyclifiedCohomotopy} to twisted rank-0 K-theory \cref{TwistedKTheoryAsNonabelianCohomology},
\begin{equation}
\label{TheComparisonMap}
  \begin{tikzcd}[
    column sep=15pt,
    row sep=11pt
  ]
    \mathrm{Cyc}_1 S^4
    \ar[
      rr,
      dashed,
      "{ \Phi }"
    ]
    \ar[
      dr,
      "{ h_3^C }"'
    ]
    &&
    B \mathrm{U}
    \sslash 
    B\mathrm{U}(1)
    \ar[
      dl,
      "{
        h_3^K
      }"
    ]
    \\
    &
    B^3 \mathbb{Z}
    \mathrlap{\,,}
  \end{tikzcd}
\end{equation}
and acting on the low-degree charges \textup{(\cref{BraneChargeFromCohomotopy,BraneChargesFromTwistedKTheory})} as%
\footnote{%
  In stating \cref{D4ChargePullback}, we are observing that the restriction $\Phi_{\vert L_1 S^4}$ factors uniquely up to homotopy through the homotopy fiber of $f_2^K$, where $f_4^K$ is indeed defined. 
}
\begin{subequations}
\label{TheChargeMultiplication}
  \begin{align}
    \Phi^\ast h_3^K
    & =
    \phantom{n \cdot {}} h_3^C
    \\
    \Phi^\ast f_2^K
    & =
    n \cdot f_2^C
    \\
    \label{D4ChargePullback}
    \bracket({
      \Phi
        _{\vert L_1 S^4}
     })
      ^\ast 
    f_4^K
    & =
    n \cdot f_4^C
  \end{align}
\end{subequations}
exists precisely when the integer $n$ is even, $n = 2k$.

\item 
\label{ComparisonMapLiftsPhiCThroughPhiK}
Such a $\Phi$ \cref{TheComparisonMap} satisfying \cref{TheChargeMultiplication} is equivalently a lift of%
\textup{\footnote{%
The quadratic correction $\lambda f_2^C \cup f_2^C$ in \cref{ComparisonMapAsALift} corresponds to a choice of normalization of the degree-four component of the comparison map. It has no effect on the D4-brane charge, since $f_2^C$ restricts to zero on $L_1 S^4$. A precedent of this kind of quadratic flux redefinition appears in \cite[(8.19)]{DMW2003}.
}}
$n \cdot \phi^C + \lambda \cdot f_2^C \cup f_2^C$ 
\textup{(from \cref{BraneChargeFromCohomotopy,GroupStructureOnBraneCharges,PastingDiagramForCyc})} through $\phi^K$ \textup{(\cref{BraneChargesFromTwistedKTheory})},
\begin{equation}
\label{ComparisonMapAsALift}
  \begin{tikzcd}[
    column sep=35pt
  ]
    \mathrm{Cyc}_1 S^4
    \ar[
      dr,
      "{
        n \cdot \phi^C + 
      }"{sloped},
      "{
        \lambda \cdot f_2^C \cup f_2^C      
      }"{swap, sloped}
    ]
    \ar[
      rr,
      dashed,
      "{ \Phi }"
    ]
    \ar[
      ddr,
      bend right=15,
      "{ h_3^C }"{swap}
    ]
    &&
    B \mathrm{U}
    \sslash 
    B \mathrm{U}(1)
    \ar[
      ddl,
      bend left=15,
      "{
        h_3^K
      }"
    ]
    \ar[
      dl,
      "{
        \phi^K
      }"{description}
    ]
    \\
    &
    \mathrm{Cyc}\, B^4 \mathbb{Z}
    \ar[
      d,
      "{ h_3 }"
    ]
    \\
    &
    B^3 
    \mathbb{Z}
    \mathrlap{\,,}
  \end{tikzcd}
\end{equation}
for 
$\lambda \in n/2 + 2\mathbb{Z}$ uniquely determined by $\Phi$ \textup{(and the choice of $\phi^C$)}. Every such $\lambda$ is realized by some $\Phi$.%
\end{enumerate}
\end{theorem}
\begin{proof}[Proof outline]
\begin{enumerate}
\item 
The Mayer--Vietoris long exact sequence for the homotopy pushout \cref{TheHomotopyPushout} in \cref{Cyc1AsAPushout} gives the cohomology groups of $\mathrm{Cyc}_1 S^4$ (and of $L_1 S^4$) on which the following spectral sequence arguments rely, such as:
\begin{subequations}
  \begin{align}
    \label{2CohomologyOfCyc1}
    H^2\bracket({
      \mathrm{Cyc}_1 S^4
      ;\,
      \mathbb{Z}
    })
    & 
    \simeq
    \mathbb{Z}\bracket\langle{
      f_2^C
    }\rangle
    \\
    \label{4CohomologyOfCyc1}
    H^4\bracket({
      \mathrm{Cyc}_1 S^4
      ;\,
      \mathbb{Z}
    })
    & 
    \simeq
    \mathbb{Z}\bracket\langle{
      \bracket({f_2^C})^2
    }\rangle
    \\
    \label
    {4CohomologyOfL1}
    H^4\bracket({
      L_1 S^4
      ;\,
      \mathbb{Z}
    })
    & 
    \simeq
    \mathbb{Z}\bracket\langle{
      f_4^C
    }\rangle
    \\
    \label{5CohomologyOfCyc1}
    H^5\bracket({
      \mathrm{Cyc}_1 S^4
      ;\,
      \mathbb{Z}
    })
    &
    \simeq
    0
    \\
    \label{7CohomologyOfCyc1}
    H^7\bracket({
      \mathrm{Cyc}_1 S^4
      ;\,
      \mathbb{Z}
    })
    &
    \simeq
    \mathbb{Z}_{/2}
    \\
    \label{9CohomologyOfCyc1}
    H^9\bracket({
      \mathrm{Cyc}_1 S^4
      ;\,
      \mathbb{Z}
    })
    &
    \simeq
    \mathbb{Z}_{/2}\bracket\langle{
      f_2^C \cup \epsilon
    }\rangle
    \\
    \label{OddCohomologyOfCyc1Above9}
    \mathllap{
    \forall_{
      \mathrm{evn} 
      \,\in\, 
      2\mathbb{N}
    }
    \;\;\;
    }
    H^{11 + \mathrm{evn}}\bracket({
      \mathrm{Cyc}_1 S^4
      ;\,
      \mathbb{Z}
    })
    & 
    \simeq 0
    \,,
      \end{align}
\end{subequations}
where the class $\epsilon$ in \cref{9CohomologyOfCyc1} is explained in a moment in \cref{Epsilon}.

\item
  We first consider lifting only to $\truncation{10}{B \mathrm{U}} \sslash B \mathrm{U}(1)$ (from which the full lift will follow readily in 
  \cref{LiftingBeyond10Truncation} of
  \cref{ConstructingTheLift} below).
  
  In this situation, the fiber of $\phi^K$ in \cref{ComparisonMapAsALift} is that of $\inlinetikzcd{ \truncation{10}{B \mathrm{U}} \ar[r] \& \truncation{4}{B \mathrm{U}} }$ (by \cref{LEMtau4BUModBU1IsCycB4Z}), which, by Bott periodicity, is 5-connected.  
  Therefore, the first obstruction to the lift $\Phi$ is encountered in the group \cref{7CohomologyOfCyc1}. 
  Inspection of the edge homomorphism of the Serre spectral sequence for $\inlinetikzcd{\mathrm{Cyc}_1 S^4 \ar[r] \& B S^1 }$ reveals that pullback to the homotopy fiber of $f_2^C$, along $\inlinetikzcd{ j : L_1 S^4 \ar[r] \& \mathrm{Cyc}_1 S^4 }$, is an isomorphism \begin{equation}
  \label{IsoInDegreeSeven}
    \inlinetikzcd{
      H^7\bracket({
        \mathrm{Cyc}_1 S^4
        ;
        \mathbb{Z}
      })
      \ar[
        rr,
        "{ j^\ast }",
        "{ \sim }"'
      ]
      \&\&
      H^7\bracket({
        L_1 S^4
        ;
        \mathbb{Z}
      })
      \mathrlap{\,.}
    }
  \end{equation}
  Therefore the first obstruction is equivalently that of lifting the integral class $n \cdot f_4^C$ \cref{ProjectingBraneChargesOutOfSmallCycCohomotopy} in $H^4\bracket({ L_1 S^4; \mathbb{Z} })$ \cref{4CohomologyOfL1}  to $h_3^C$-twisted K-theory.

  \item 
  \label{TheSteenrodObstructionTerm}
  This latter obstruction is well-known (\cite[(4.1)]{AtiyahSegal2006}, cf. \parencites[(11.5)]{DMW2003}[(2.3)]{EvslinSati2006}) to be given by the expression  
    \begin{equation}
     \label{TwistedSteenrodSquare}
       \bracket({
        \mathrm{Sq}^3_{\mathbb{Z}}
        \, 
         -
        h_3^C 
         \cup
        (-)
       })
       \bracket({  
         n \cdot f^C_4
       })
       \,\in\,
       H^7\bracket({
         L_1 S^4
         ;\,
         \mathbb{Z}
       })
       \simeq
       \mathbb{Z}_{/2}
       \mathrlap{\,.}
     \end{equation}
     
    Crucially, in our case, 
    $f_4^C$ is pulled back from $S^4$ \cref{ProjectingBraneChargesOutOfSmallCycCohomotopy}. Since $\mathrm{Sq}^3$ is natural while $H^7\bracket({ S^4; \mathbb{Z}}) = 0$, this immediately implies that 
    the Steenrod summand in \cref{TwistedSteenrodSquare} vanishes identically:
    \begin{equation}
      \label
      {IntegralEquationOfMotion}
      \mathrm{Sq}^3_{\mathbb{Z}}\, f_4^C 
      \,=\, 
      0
      \mathrlap{\,.}
    \end{equation}
    This equation \cref{IntegralEquationOfMotion} is called the \textbf{integral equation of motion of M-theory} in \cite[\S~5]{DMW2000}, there considered in the untwisted DMW sector (cf. \cref{DMWSector}); here it is implied (cf. \parencites[Prop. 3.15]{FSS20-H}{GS21}{Grady2025}) by M-brane charge quantization in 4-Cohomotopy, \cref{TheCFieldIn11D}.

    \item
    But since here we do not restrict to vanishing twist $h_3^C$, the obstruction \cref{TwistedSteenrodSquare} is nominally still given by the $n$-fold multiple of the class
    \begin{equation}
    \label{Epsilon}
      \inlinetikzcd{
      \epsilon 
      \,:=\,
      { h_3^C \cup f_4^C }
      \,\in\,
      H^7\bracket({
        L_1 S^4
        ;\,
        \mathbb{Z}
      })
      \ar[
        rr,
        "{ 
          \scaledbracket({
            j^\ast 
          })^{-1}
        }",
        "{ \sim }"'
      ]
      \&\&
      H^7\bracket({
        \mathrm{Cyc}_1 S^4
        ;\,
        \mathbb{Z}
      })
      \simeq
      \mathbb{Z}_{/2}
      \mathrlap{\,.}
      }
    \end{equation} 
    Now the multiplicative integral Serre spectral sequence of the evaluation fibration $\inlinetikzcd{ S^3 \vee S^6 \simeq \Omega_1 S^4 \ar[r] \&  L_1 S^4 \ar[r, "{ \mathrm{ev}_0 }"] \& S^4 }$
    \cref{ComputingBasedLoopsOfS4} has 
    \begin{equation}
      \label{d4OfIota6}
      \mathrm{d}_4(\iota_6) 
        = 
      \pm 2 \iota_3 \iota_4
      \mathrlap{\,,}
    \end{equation} whence
    \begin{equation}
      \label{EpsilonIsNontrivial}
      \epsilon = [1] \in \mathbb{Z}_{/2}
    \end{equation}
    is non-trivial.
    This shows that lifts $\Phi$ cannot exist for odd $n$, forced by the $h_3^C$-twist (and that $n = 1$ is possible in the DMW sector, where $h_3 = 0$, cf. \cref{DMWSector} below).

  \item
  \label{ConstructingTheLift}
  Since twisted K-theory classes are additive, the claimed existence of lifts for all even $n$ now follows as soon as we establish a lift at $n = 2$. 
  
  This is obtained by a careful analysis of the fate of the D6-brane charge $f_2^C$ in the $h_3^C$-twisted Atiyah--Hirzebruch spectral sequence:
  \begin{enumerate}
    \item 
    First we note that there are only two possibly non-vanishing differentials of $f_2^C$, $\mathrm{d}_5^{h_3}$ and $\mathrm{d}_7^{h_3}$, because
    \begin{enumerate}
      \item all differentials $\mathrm{d}^{h_3}_{2\bullet}$ vanish since their codomains are odd rows, while the coefficient space has homotopy concentrated in even degree,
      \item 
      $\mathrm{d}_3^{h_3}$
      and all $\mathrm{d}^{h_3}_{9 + 2\bullet}$
      vanish since their codomains are subquotients of
      the vanishing groups \cref{5CohomologyOfCyc1,OddCohomologyOfCyc1Above9}, respectively.
    \end{enumerate}

    \item
    \label{FirstNontrivialDifferential}
    Its first possibly non-vanishing differential is just the class \cref{Epsilon}:
    \begin{equation}
      \label{d5Off2}
      \mathrm{d}^{h_3}_5 \, f_2^C
      =
      \epsilon
      \mathrlap{\,,}
    \end{equation} 
    as in the argument of \cref{TheSteenrodObstructionTerm}.
    But since this class inhabits the 2-torsion group \cref{7CohomologyOfCyc1}, we find: 
    \begin{equation}
      \mathrm{d}^{h_3}_5 
      \bracket({
        2 \cdot f_2^C
      })
      =
      2 \cdot 
      \mathrm{d}^{h_3}_5 f_2^C
      =
      2 \cdot \epsilon
      =
      [0]
      \in
      \mathbb{Z}_{/2}
      \mathrlap{\,.}
    \end{equation}

    \item 
    The second possibly non-vanishing differential is of the form
    $\mathrm{d}^{h_3}_7 \bracket({2 \cdot f_2^C})$ in a subquotient of $H^9\bracket({ \mathrm{Cyc}_1 S^4 ; \mathbb{Z} }) \simeq \mathbb{Z}_{/2}\bracket\langle{ f_2^C \cup \epsilon }\rangle$ \cref{9CohomologyOfCyc1}.

    Beware that with $\mathrm{d}^{h_3}_5 f_2^C = \epsilon \neq [0]$ \cref{EpsilonIsNontrivial}, the class $2 \cdot f_2^C$, which survives to the next page by \cref{FirstNontrivialDifferential}, is no longer divisible there, so that linearity of $\mathrm{d}_7^{h_3}$ alone does not imply vanishing of $\mathrm{d}_7^{h_3}\bracket({2 \cdot f_2^C})$, despite its notational appearance.

    Instead, we invoke the left module structure, to be denoted $(-) \dot \cup (-)$, of the $h_3$-twisted AHSS over the $0$-twisted AHSS, and use that $f_2^C$ may equivalently be regarded as an untwisted class, to deduce:
    \begin{equation}
      \begin{aligned}
      \mathrm{d}_5^{h_3}\bracket({
        f_2^C
        \cup
        f_2^C
      })
      & =
      \mathrm{d}_5^{h_3}\bracket({
        f_2^C
        \mathbin{\dot\cup}
        f_2^C
      })
      \\
      & =
      \underbrace{
      \mathrm{d}_5^{0}\bracket({
        f_2^C
      })
      }_{0}
      \mathbin{\dot\cup}
      f_2^C
      +
      f_2^C
      \mathbin{\dot\cup}
      \underbrace{
      \mathrm{d}_5^{h_3}\bracket({
        f_2^C
      })
      }_{ \epsilon }
      \\
      & =
      f_2^C \mathbin{\dot\cup} \epsilon
      = 
      f_2^C \cup \epsilon
      \mathrlap{\,,}
    \end{aligned}
  \end{equation}
  where under the braces we used \cref{d5Off2} and that $f_2^C$, being the class of the complex line bundle associated with the M-circle bundle, is a permanent AHSS cycle.

  Similarly, to see that $\bracket({f_2^C})^2$ exists on the 5th page in the first place, one computes:
  \begin{equation}
    \begin{aligned}
      \mathrm{d}_3^{h_3}
      \bracket({
        f_2^C 
          \cup 
        f_2^C
      })
      &=
      \mathrm{d}_3^{h_3}
      \bracket({
        f_2^C 
          \mathbin{\dot\cup} 
        f_2^C
      })
      \\
      & =
      \underbrace{
      \mathrm{d}_3^0\bracket({
        f_2^C
      })
      }_{0}
      \mathbin{\dot\cup}
      f_2^C
      +
      f_2^C
      \mathbin{\dot\cup}
      \underbrace{
      \mathrm{d}_3^{h_3}\bracket({
        f_2^C
      })
      }_{ 0 }
      \mathrlap{\,,}
    \end{aligned}
  \end{equation}
  where now under the second brace we used \cref{5CohomologyOfCyc1}.

  This reveals that the generator $f_2^C \cup \epsilon = \mathrm{d}_5^{h_3}\bracket({ \bracket({f_2^C})^2 })$ of $H^9$ \cref{9CohomologyOfCyc1} is a boundary and hence vanishes on the next page. Since $\mathrm{d}_7^{h_3} \bracket({2 \cdot f_2^C})$ must be a multiple of this vanished generator, it vanishes, too.  

  \item
  \label{LiftApproximationOnFiniteModel}
  With the two possibly non-trivial differentials on $2 \cdot f_2^C$ thus vanishing, the class survives to the $\infty$-page, and as such it exhibits an $h_3$-twisted K-theory class $\xi$ on, so far, the finite-dimensional approximation of $\mathrm{Cyc}_1 S^4$ given by $\mathrm{Cyc}_1 S^4 \times_{B S^1} \mathbb{C}P^5$, with first Chern class $2\cdot f_2^C$.

  \item
  \label{LiftingBeyond10Truncation}
  We now observe that the relative cells of $\bracket({ \mathrm{Cyc}_1 S^4,\, \mathrm{Cyc}_1 S^4 \times_{B S^1} \mathbb{C}P^5})$ have dimension $\geq 12$, so that this class $\xi$ from \cref{LiftApproximationOnFiniteModel} actually extends to all of $\mathrm{Cyc}_1 S^4$ as a map to $\truncation{10}{B\mathrm{U}} \sslash B \mathrm{U}(1)$.

  But the further Postnikov obstructions to lifting this further, all the way through $\inlinetikzcd{ B\mathrm{U} \sslash B \mathrm{U}(1) \ar[r] \& \truncation{10}{B \mathrm{U}} \sslash B \mathrm{U}(1) }$, now lie entirely in $H^{13 + \mathrm{evn}}\bracket({ \mathrm{Cyc}_1 S^4; \mathbb{Z} })$, while these groups all vanish by \cref{OddCohomologyOfCyc1Above9}. This establishes the promised full lift $\inlinetikzcd{ \mathrm{Cyc}_1 S^4 \ar[r, "{ \Phi }"] \& B \mathrm{U} \sslash B \mathrm{U}(1) }$ claimed in \cref{TheComparisonMap}.

  \item
  \label
  {CycB4ZInProofOfComparisonMap}
  Finally, to obtain from this the D4-brane charge as claimed in \cref{D4ChargePullback}, one may show that the group of D6/NS5/D4-brane charges on $\mathrm{Cyc}_1 S^4$ (via \cref{GroupStructureOnBraneCharges}), hence of homotopy classes of maps of the form
  \[
    \begin{tikzcd}[
      column sep=20pt,
      row sep=3pt
    ]
      \mathrm{Cyc}_1 S^4
      \ar[
        rr,
        dashed,
      ]
      \ar[
        dr,
        "{
          h_3^C
        }"'
      ]
      &&
      \mathrm{Cyc} 
      \,
      B^4 \mathbb{Z}
      \ar[
        dl,
        "{ h_3 }"
      ]
      \\
      &
      B^3 \mathbb{Z}
      \mathrlap{\,,}
    \end{tikzcd}
  \]
  is an extension of \cref{2CohomologyOfCyc1} by 
  \cref{4CohomologyOfCyc1}, split by $f_2^C \mapsto \phi^C$ \cref{ProjectingBraneChargesOutOfSmallCycCohomotopy}, hence is the middle group in
  \begin{equation}
  \label
  {GroupOfMapsFromCyc1S4ToCycB4Z}
    \begin{tikzcd}[
      column sep=15pt
    ]
      0
      \ar[r]
      &
      \mathbb{Z}\bracket\langle{
        \bracket({ f_2^C })^2
      }\rangle
      \ar[
        rr, 
        hook
      ]
      &&    
      \mathbb{Z}\bracket\langle{
        \bracket({ f_2^C })^2
      }\rangle
      \oplus
      \mathbb{Z}\bracket\langle{
        \phi^C
      }\rangle
      \ar[
        rr,
        "{  
          \phi^C 
            \,\mapsto\,
          f_2^C
        }",
        "{
          \scaledbracket({
            f_2^C
          })^2
            \,\mapsto\,
          0
        }"'
      ]
      &\phantom{--}&
      \mathbb{Z}\bracket\langle{
        f_2^C
      }\rangle
      \ar[
        r
      ]
      &
      0
      \mathrlap{\,.}
    \end{tikzcd}
  \end{equation}
  This reveals that---up to addition of the ambiguity term $\bracket({f_2^C})^2$ which disappears on $L_1 S^4$---these brane charges come in joint multiples of $\phi^C$, which by \cref{ProjectingBraneChargesOutOfSmallCycCohomotopy} means that $n \cdot f_2^C$ goes along with $n \cdot f_4^C$.  
\end{enumerate}
  This way we arrive at the main claim of \cref{TheComparisonMapAtEvenN}. 

\item
  From the last step, \cref{CycB4ZInProofOfComparisonMap} in \cref{ConstructingTheLift}, we immediately obtain the proof of most of the claim \cref{ComparisonMapLiftsPhiCThroughPhiK}:  
  If $\Phi^\ast$ sends $f_2^K \mapsto n \cdot f_2^C$ 
    \cref{TheChargeMultiplication},
    then \cref{GroupOfMapsFromCyc1S4ToCycB4Z} implies that 
    \begin{equation}
      \phi^K \circ \Phi
      -
      n \cdot \phi^C
      =
      \lambda
      \bracket({
        f_2^C
      })^2
      \mathrlap{\,,}
    \end{equation}
    for unique $\lambda \in \mathbb{Z}$.
    This is the claimed commutativity of \cref{ComparisonMapAsALift}.
  \qedhere
\end{enumerate}
\end{proof}

\begin{remark}[Explicit construction]
\label[remark]{ExplicitConstruction}
  The universal twisted K-theory class $\Phi$ on $\mathrm{Cyc}_1 S^4$ at $n = 2k$, whose existence is guaranteed by \cref{Comparison_map_to_twisted_K_theory}, admits a more direct explicit construction relying on the pushout characterization of $\mathrm{Cyc}_1 S^4$ given by \cref{Cyc1AsAPushout}. We briefly indicate this without proof:
  
  Denoting by
  \begin{enumerate}
    \item $L$ the complex line bundle classified by $f_2^C$,
    
    \item $\beta$ 
    the Bott generator on $S^4$ \cref{BottGeneratorOn4Sphere}

    (the tautological quaternionic line minus its complex rank),

    \item $G$ the transition line bundle of the gerbe $h_3^C$
    
    (with respect to chosen trivializations of $h_3$ on the legs of the pushout, which exist since the latter's integral 3-cohomology vanishes)
  \end{enumerate}
  we have the following data on the pushout components (for $n = 2k$ in \cref{TheChargeMultiplication}):
  \begin{equation}
  \label{TheExplicitKClasses}
  \adjustbox{}{
  \begin{tabular}{lll}
    On oriented great circles
    & 
    the K-theory class
    & 
    $A 
      := 
      \tfrac{n}{2}
      \bracket({
        L - L^{-1}
      })
    $,
    \\
    on degenerate circles
    &
    the K-theory class
    &
    $
      B 
        := 
      n(L - 1)
      +
      \tfrac{n}{2}
      \beta(L + 1)
    $,
    \\
    on great circles \& poles
    &
    the K-theory relation
    &
    $
      A
      = 
      G B  
      \mathrlap{\,.}
    $
  \end{tabular}
  }
  \end{equation}
  Here the product operation corresponds to tensoring of vector bundles, hence in the last line corresponds to the $B \mathrm{U}(1)$ action (by $G$) on $B \mathrm{U}$. Therefore these data constitute a cocone under our pushout \cref{TheHomotopyPushout} which by the universal property of the latter establishes a comparison map $\Phi$ (\cref{TheCocone}). This is the map guaranteed by \cref{Comparison_map_to_twisted_K_theory}.
\end{remark}

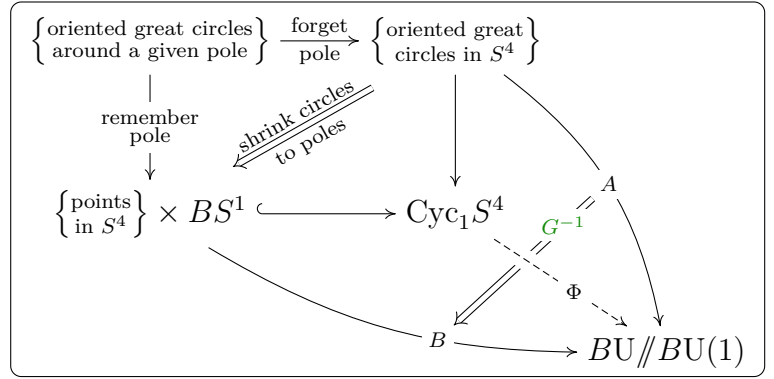
\begin{SCfigure}[.55][htb]
\caption{\label{TheCocone}%
  The cohomology operation from 
  small-cyclified 4-Cohomotopy to 
  twisted K-theory, guaranteed by \cref{Comparison_map_to_twisted_K_theory}, is the amalgamation of a pair of canonical K-theory classes and their relation \cref{TheExplicitKClasses} that exist naturally on the components of the pushout decomposition \cref{TheHomotopyPushout} of the small-cyclified loop space of the 4-sphere, as described in \cref{ExplicitConstruction}.
}
\centering
\adjustbox{rndfbox=4pt}{
$
    \begin{tikzcd}[
      row sep=40pt,
      column sep=30pt
    ]
     \left\{
      \substack{
        \textup{oriented great circles}
        \\
        \textup{around a given pole}
      }
      \right\}
      \ar[
        r,
        "{ 
          \text{forget}
        }"{yshift=-1pt},
        "{ 
          \text{pole}
        }"{swap,yshift=+1pt}
      ]
      \ar[
        d,
        "{ 
          \substack{
            \textup{remember}
            \\
            \textup{pole}
          }
        }"{description}
      ]
      &
     \left\{
      \substack{
        \textup{oriented great}
        \\
        \textup{circles in $S^4$}
      }
      \right\}
      \ar[
        dl,
        shorten=10pt,
        Rightarrow,
        "{ 
          \text{shrink circles} 
        }"{sloped},
        "{ 
          \text{ to poles } 
        }"{swap, sloped}
      ]
      \ar[
        ddr,
        bend left=21,
        "{ A }"{
          description,
          name=A,
          pos=.5
        }
      ]
      \ar[d]
      &[-20pt]
      \\
      \big\{
      \substack{
        \textup{points} 
        \\
        \textup{in $S^4$}
      }
      \big\}
      \times B S^1
      \ar[
        drr,
        bend right=15,
        "{ B }"{
          description,
          name=B,
          pos=.65
        }
      ]
      \ar[
        r,
        hook
      ]
      &
      \mathrm{Cyc}_1 S^4
      \ar[
        Rightarrow,
        from=A, to=B,
        crossing over,
        "{ \color{darkgreen} 
          G^{-1} 
        }"{
          description,
          pos=.21
        }
      ]
      \\[-10pt]
      & 
      &
      \ar[
        from=ul,
        dashed,
        crossing over,
        "{ \Phi }"{
          description,
          pos=.6
        }
      ]
      B \mathrm{U}
      \sslash 
      B \mathrm{U}(1)
    \end{tikzcd}
$
}
\end{SCfigure}

\begin{remark}[\scalebox{.975}{Cohomology operation from small-cyclified Cohomotopy to twisted K-theory}]
\begin{enumerate}
\item
We may read \cref{Comparison_map_to_twisted_K_theory} as establishing 
(for fixed NS5-brane charges $h_3 \in H^3\bracket({X^9; \mathbb{Z}})$) 
K-valued ``character map''-like cohomology operations \cref{TwistedCohomologyOperation} that naturally approximate small-cyclified 4-Cohomotopy of 9-manifolds by twisted rank-zero K-theory 
\cref{TwistedKTheoryAsTwistedNonabe}:
\begin{equation}
  \begin{tikzcd}
   H^1_{h_3}\bracket({
     X^9;
     \Omega\,
     \mathrm{Cyc}_1
     S^4
   })
   \ar[
     r,
     "{ \Phi_\ast }"
   ]
   &
   H^1_{h_3}\bracket({
     X^9;
     \Omega\,
     B\mathrm{U}
   })   
   \simeq
   \mathrm{KU}^{h_3}_{\mathrm{rk=0}}\bracket({
     X^9
   })
   \mathrlap{\,.}
   \end{tikzcd}
\end{equation}

\item
In view of \cref{TheChargeMultiplication}, this means that twisted K-theory is an abelian approximation to small-cyclified 4-Cohomotopy over IIA spacetimes, which:
\begin{enumerate}
\item necessarily forgets the NS1-brane charge 

(not quantizable in any abelian cohomology theory), 
\item accurately reflects NS5-brane charge,
\item 
\label{The2TorsionIssue}
reflects D4- and D6-brane charges up to a factor of 2.
\end{enumerate}

\end{enumerate}
\end{remark}

Next we reconsider this situation in the DMW subsector, where the comparison becomes finer still, since there the factor of 2 in \cref{Comparison_map_to_twisted_K_theory} \cref{TheChargeMultiplication} disappears (by \cref{ComparisonInDMWSector} \cref{TheShiftedLift} below).

\subsection
{The DMW Sector}
\label
{DMWSector}

By the \emph{DMW sector} we shall mean the situation of vanishing NS5- and D6-brane charge, as in \cref{TheObservationByDMW}. In terms of classifying spaces, passage to this sector means lifting classifying maps of charges through the homotopy fibers of the maps $\bracket({f_2^C, h_3^C})$ \cref{ProjectingBraneChargesOutOfSmallCycCohomotopy} and $\bracket({f_2^K,\, h_3^K})$ \cref{DiagramOfBraneChargesInKTheory} that extract these charges, respectively:
\begin{equation}
\label
{HomotopyFibersForDMWSector}
  \begin{tikzcd}[column sep=large]
    \widetilde{ L_1 S^4 }
    \ar[
      d,
      "{ \tilde j }"'
    ]
    \ar[r]
    \ar[
      dr,
      phantom,
      "{ \lrcorner }"{pos=.1}
    ]
    & 
    \ast
    \ar[d]
    \\
    \mathrm{Cyc}_1 S^4
    \ar[
      r,
      "{
        \scaledbracket({
          f_2^C
          ,\,
          h_3^C
        })
      }"'{swap}
    ]
    &
    B^2 \mathbb{Z}
    \times
    B^3 \mathbb{Z}
    \mathrlap{\,,}
  \end{tikzcd}
  \;\;\;\;\;
  \begin{tikzcd}[column sep=large]
    B \mathrm{SU}
    \ar[d]
    \ar[r]
    \ar[
      dr,
      phantom,
      "{ \lrcorner }"{pos=.1}
    ]
    & 
    \ast
    \ar[d]
    \\
    B \mathrm{U}
    \sslash
    B \mathrm{U}(1)
    \ar[
      r,
      "{
        \scaledbracket({
          f_2^K
          ,\,
          h_3^K
        })
      }"'{swap}
    ]
    &
    B^2 \mathbb{Z}
    \times
    B^3 \mathbb{Z}
    \mathrlap{\,.}
  \end{tikzcd}
\end{equation}

\begin{lemma}[Rational model of small-cyclified Cohomotopy in DMW sector]
\label[lemma]{RationalModelsOfTildeL1S4}
  \begin{enumerate}
  \item
  \label{RelativeMinimalModelForTildeL1S4}
  The \emph{relative} Whitehead $L_\infty$-algebra $\mathfrak{l}_{_{\mathrm{Cyc}_1 S^4}} \widetilde{ L_1 S^4}$ of $\tilde j$ \cref{HomotopyFibersForDMWSector} \textup{(the Koszul dual to its relative minimal Sullivan model \cite[Prop. 5.16]{FSS23-Char})} induces the fluxes and Gauss laws of \cref{The-full-IIA-EM-Gauss-laws} with global potentials $A_1$ and $B_2$ for $F_2$ and $H_3$ adjoined, respectively:
  \begin{equation}
    \label{ClosedRelativeForms}
    \Omega^1_{\mathrm{cl}}
    \bracket({
      X^9
      ;\,
      \mathfrak{l}_{_{\mathrm{Cyc}_1 S^4}} \widetilde{ L_1 S^4 }
    })
    \,\simeq\,
    \left\{
    \begin{aligned}
      A_1 
      & 
      \in
      \Omega^1_{\mathrm{dR}}
      \bracket({X^9})
      \\
      B_2 
      & 
      \in
      \Omega^2_{\mathrm{dR}}
      \bracket({X^9})
      \\
      \left({
        \begin{aligned}
        H_3, 
        H_7,
        \\
        F_2, 
        F_4,
        F_6,
        F_8
        \end{aligned}
      }\right)
      & \in
      \Omega^1_{\mathrm{cl}}
      \bracket({
        X^9
        ;\,
        \mathfrak{l}
        \mathrm{Cyc}_1 S^4
      })
    \end{aligned}
    \,\middle\vert\,
    \begin{aligned}
      \mathrm{d}\, A_1
      & = 
      F_2
      \\
      \mathrm{d}\, B_2 
      & =
      H_3
      \\
      \phantom{A}
      \\
      \phantom{A}
    \end{aligned}
    \right\}
    \mathrlap{.}
  \end{equation}

  \item 
  \label{AbsoluteMinimalModelForTildeL1S4}
  The absolute Whitehead $L_\infty$-algebra $\mathfrak{l} \widetilde{ L_1 S^4 }$ instead induces
  \begin{equation}
    \Omega^1_{\mathrm{cl}}\bracket({
      X^9
      ;\,
      \mathfrak{l} 
      \widetilde{ L_1 S^4 }
    })
    \simeq
    \left\{
      \bracket({
        \begin{aligned}
        H_3, 
        \widetilde H_7, 
        \\
        F_2, 
        \widetilde F_4, 
        \widetilde F_6,
        \widetilde F_8
        \end{aligned}
      })
      \in
      \Omega^1_{\mathrm{cl}}    
      \bracket({
        X^9
        ;\,
        \mathfrak{l}
        \mathrm{Cyc}_1 S^4
      })
      \,\middle\vert\,
      \begin{aligned}
        H_3 & = 0
        \\
        F_2 & = 0
      \end{aligned}
     \right\}
     \mathrlap{.}
  \end{equation}

  \item
  \label{RelativeToAbsoluteMinimalModel}
  Their weak equivalence $\inlinetikzcd{\mathfrak{l}_{\mathrm{Cyc}_1 S^4}\widetilde{ L_1 S^4 } \ar[r, "{\gamma}", "{ \sim }"'] \& \mathfrak{l} \widetilde{ L_1 S^4 } }$ induces the following transformation of closed $L_\infty$-algebra valued differential forms on $X^9$:
  \begin{equation}
  \label
  {WeakEquivalenceBetweenRelativeAndAbsoluteModel}
    \begin{tikzcd}[
      ampersand replacement=\&,
      row sep=-3pt, 
      column sep=10pt,
      /tikz/column 5/.append style={anchor=base west},
    ]
      \Omega^1_{\mathrm{cl}}
      \bracket({
        X^9
        ;\,
        \mathfrak{l}_{_{\mathrm{Cyc}_1 S^4}}
        \widetilde{ L_1 S^4 }
      })
      \ar[
        rr,
        "{ \gamma_\ast }"
      ]
      \&\&
      \Omega^1_{\mathrm{cl}}\bracket({
        X^9
        ;\,
        \mathfrak{l}\widetilde{ L_1 S^4 }
      })
      \\
      \left(
      \begin{array}{l}
        A_1
        \\
        B_2
        \\
        H_3
        \\
        H_7
        \\
        F_2
        \\
        F_{4 + 2k}
      \end{array}
      \right)
      \&\longmapsto\&
      \left(
      \begin{array}{ll}
        0 &
        \\
        0 &
        \\
        0 &
        \\
        \widetilde H_7
        & \defneq
        H_7 
          + 
        A_1 
          \wedge
        \widetilde F_6
        \\
        0
        \\
        \widetilde{F}_{4+2k}
        & \defneq
        \bracket({
          e^{-B_2}
          \wedge
          \sum_{r}
          F_{2r}
        })_{4+2k}
      \end{array}
      \right)
      \mathrlap{.}
    \end{tikzcd}
  \end{equation}
  \end{enumerate}
\end{lemma}
\begin{proof}
\Cref{RelativeMinimalModelForTildeL1S4,AbsoluteMinimalModelForTildeL1S4} follow by standard arguments in rational homotopy theory.
\Cref{RelativeToAbsoluteMinimalModel}
is mainly a reincarnation of the classical observation in twisted cohomology (cf. \cite[(23) \& app.]{RohmWitten1985}), that 
\begin{equation}
  \begin{aligned}
    \mathrm{d}
    \bracket({
      e^{-B_2}
      \wedge
      \sum_r F_{2r}
    })
    & = 0
    \\
    \Leftrightarrow
    \;\;\;\;\;\;
    \mathrm{d}
    \bracket({
      e^{-B_2}
      \wedge
      \sum_r F_{2r}
    })_{2k}
    & = 0
    \;\;\;
    \forall_k
    \mathrlap{\,.}
  \end{aligned}
\end{equation}
The only extra subtlety here is the presence of $H_7$ with its non-linear Gauss law \cref{TheH7GaussLaw}. For this, first observe that $\mathrm{d} H_7$ is invariant in form under exchanging
$F_{2\bullet} \leftrightarrow \widetilde F_{2\bullet}$, where we include 
\begin{equation}
  \widetilde F_2 
    := 
  \bracket(
    {e^{-B_2} 
    \wedge 
    \sum_r F_{2r} 
  })_2 
    = 
  F_2
  \mathrlap{\,,}
\end{equation}
in that:
\begin{equation}
  \begin{aligned}
    \mathrm{d}\, H_7
    & =
    \tfrac{1}{2}
    F_4 \wedge F_4
    -
    F_2 \wedge F_6
    \\
    & =
    \tfrac{1}{2}
    \widetilde F_4 
      \wedge 
    \widetilde F_4
    -
    \widetilde F_2 
      \wedge 
    \widetilde F_6
    \mathrlap{\,.}
  \end{aligned}
\end{equation}
Namely consider the algebra automorphism $\sigma$ which on even degree generators $B_2$, $F_{2k}$ is $\sigma(\alpha_{2k}) := (-1)^{k} \alpha_{2k}$, then:
  \begin{align*}
    \tfrac{1}{2}
    F_4 \wedge F_4
    -
    F_2 \wedge F_6
    & =
    \tfrac{1}{2}\bracket({
      \sigma\bracket({
        \sum_r F_{2r}
      })
      \wedge
      \sum_r F_{2r}
    })_8
    \\
    & =
    \tfrac{1}{2}\bracket({
      \sigma\bracket({
        e^{B_2}
        \wedge
        \sum_r \widetilde F_{2r}
      })
      \wedge
      e^{B_2}
      \wedge
      \sum_{r} \widetilde F_{2r} 
    })_8
    \\
    & =
    \tfrac{1}{2}\bracket({
        e^{-B_2}
        \wedge
        \sigma\bracket({
          \sum_r \widetilde F_{2r}
        })
      \wedge
      e^{B_2} \sum_{r} \widetilde F_{2r} 
    })_8
    \\
    & =
    \tfrac{1}{2}\bracket({
        \sigma\bracket({
          \sum_r \widetilde F_{2r}
        })
      \wedge
      \sum_{r} \widetilde F_{2r} 
    })_8
    \\
    &
    =
    \tfrac{1}{2}
    \widetilde F_4 
    \wedge
    \widetilde F_4
    -
    \widetilde F_2
    \wedge
    \widetilde F_6
    \mathrlap{\,.}
  \end{align*}
  Therefore the differential on $\widetilde H_7$ in \cref{WeakEquivalenceBetweenRelativeAndAbsoluteModel} indeed has the required form
  \begin{align*}
      \mathrm{d}\,
      \widetilde H_7
      & \defneq
      \mathrm{d}\bracket({
        H_7 
          + 
        A_1 
          \wedge 
        \widetilde F_6
      })
      \\
      & =
      \tfrac{1}{2}
      \widetilde F_4
      \wedge
      \widetilde F_4
      -
      \smash{
      \underbrace{
      \bracket({
        \widetilde F_2
        -
        \mathrm{d}A_1
      })}_{0}
      \wedge \widetilde F_6
      }
      \\
      & =
      \tfrac{1}{2}
      \widetilde F_4
      \wedge
      \widetilde F_4
      \mathrlap{\,.}
      \qedhere
    \end{align*}
\end{proof}

In this DMW sector \cref{HomotopyFibersForDMWSector}, we have our second main result, elaborating here on \cref{Comparison_map_to_twisted_K_theory}:
\begin{theorem}[Comparison map in the DMW sector]
\label[theorem]{ComparisonInDMWSector}
\begin{enumerate}
\item 
\label
{LiftInDMWSector}
  In the DMW sector \cref{HomotopyFibersForDMWSector}, the analogues of the lifts \cref{TheComparisonMap}, namely dashed maps of the form
  \begin{equation}
  \label{LiftInTheDMWSector}
    \begin{tikzcd}[row sep=small]
      \widetilde{ L_1 S^4 }
      \ar[
        rr,
        dashed,
        "{
          \widetilde \Phi
        }"
      ]
      \ar[
        dr,
        "{
          n \cdot 
          \widetilde{\phi^C}
        }"'
      ]
      &&
      B \mathrm{SU}
      \ar[
        dl,
        "{ \widetilde{\phi^K} }"
      ]
      \\
      &
      B^4 \mathbb{Z}
      \mathrlap{\,,}
    \end{tikzcd}
  \end{equation}
  exist for all $n \in \mathbb{N}$.

\item
  A canonical such lift at $n = 1$ is given by the Bott generator $\beta \in \widetilde{\mathrm{KU}}^0\bracket({S^4})$:
  \begin{equation}
  \label
  {BottGeneratorOn4Sphere}
    \begin{tikzcd}
      \widetilde{L_1 S^4}
      \ar[
        r,
        "{ \mathrm{ev}_0 }"
      ]
      &
      S^4 
      \ar[
        r,
        "{ \beta }"
      ]
      &
      B \mathrm{SU}
      \mathrlap{\,,}
    \end{tikzcd}
  \end{equation}
  whose Chern character is
  concentrated in degree 4:
  \begin{equation}
    \mathrm{ch}\bracket({
      \mathrm{ev}_0^\ast \beta
    })
    \defneq
    \widetilde F_4^C
    \;
    \in
    H^4\bracket({
      \widetilde
      { L_1 S^4 }
      ;\,
      \mathbb{Q}
    })
    \mathrlap{\,.}
  \end{equation}

  \item
  \label{TheShiftedLift}
  For every lift $\Phi$ at $n=2$ in \cref{TheComparisonMap}, its shift by the Bott generator \cref{BottGeneratorOn4Sphere},
  \begin{equation}
  \label
  {ShiftedLifts}
    \widetilde{\Phi}
    \,:=\,
    \tilde j^\ast \Phi
    \,-\,
    \mathrm{ev}_0^\ast \beta
    \mathrlap{\,,}
  \end{equation}
  is a lift at $n = 1$ in \cref{LiftInTheDMWSector}.

  \item 
  \label{UniversalChernChInDMWSector}
  The universal Chern character of such a lift in \cref{ShiftedLifts} is up to degree 10 of the form
  \begin{equation}
    \label{ChernCharacterOfTildePhi}
    \mathrm{ch}_{\leq 10}
    \bracket({
      \widetilde \Phi
    })
    \,=\,
      \widetilde F_4^C
    \,+\,
    2 \widetilde F_6^C
    \,+\,
    2 \widetilde F_8^C
    \,+\,
    2 \widetilde F_{10}^C
    +
    \kappa 
    \widetilde F_4^C \wedge \widetilde F_6^C
    \mathrlap{\,,}
    \;\;\;
    \left\{
    \begin{aligned}
      \widetilde F_{2k}^C
      & \in
      H^{2k}\bracket({
        \widetilde{L_1 S^4}
        ;\,
        \mathbb{Q}
      })
      \\
      \kappa 
      & 
      \in
      \mathbb{Q}
    \mathrlap{\,,}
    \end{aligned}
    \right.
  \end{equation}
  where each summand in the relevant degrees $\leq 8$ is integral and is as such an indivisible generator of the lattice $H^{2k}_{\mathrm{int}}(-;\mathbb{Q})$ of rational cohomology classes that have integral preimages:
  \begin{equation}
  \label
  {IntegralityOfChernCharacterComponents}
    \begin{aligned}
      H^{2k}
      _{\mathrm{int}}\bracket({
        \widetilde{L_1 S^4};
        \mathbb{Q}
      })
      \simeq
      \begin{cases}
        \mathbb{Z}
        \cdot
        \phantom{1} \widetilde F_4^C
        & \text{ if $2k = 4$ }
        \\
        \mathbb{Z}
        \cdot
        2 \widetilde F_6^C
        & \text{ if $2k = 6$ }
        \\
        \mathbb{Z}
        \cdot
        2 \widetilde F_8^C
        & \text{ if $2k = 8$}
        \mathrlap{\,.}
      \end{cases}
    \end{aligned}
  \end{equation}
  \item 
  \label{PhiNotDivisible}
  In particular, the lift $\Phi$ \cref{TheComparisonMap} does not become divisible by 2 in the DMW sector; instead it is the shift \cref{BottGeneratorOn4Sphere} that normalizes it.
  \end{enumerate}
\end{theorem}
    \begin{proof}[Proof outline]
    We explain the main steps modulo spectral sequence runs.
    \begin{enumerate}
    
    \item
      Consider the generators $H^K_3$ and  $F_{2\bullet}^K$ of the rational minimal model of $B \mathrm{U} \sslash B \mathrm{U}(1)$ (cf. \parencites[\S~2]{FreedHopkinsTeleman2002}[Ex. 6.6]{FSS23-Char}), satisfying
      \begin{equation}
        \label{SullivanModelOfTwistedK}
        \mathrm{d}\, F_{2k}^K
        =
        H_3^K \wedge F_{2(k-1)}^K
        \mathrlap{\,.}
      \end{equation}
      The universal twisted Chern character on this space is $\mathrm{ch} \defneq \sum_{k \in \mathbb{N}} F_{2k}^K$, with $F_2^K$ 
      normalized to the rational image of the first Chern class $c_1$. 
      
      Therefore, by the multiplicities \cref{TheChargeMultiplication} from \cref{Comparison_map_to_twisted_K_theory}, the twisted Chern character of $\Phi$ starts out as
      \begin{equation}
        \label{RationalChernChOfPhi}
        \begin{aligned}
        \mathrm{ch}\bracket({\Phi})
        & =
        \Phi^\ast\bracket({
          F_2^K 
          + 
          F_4^K
          + 
          \cdots
        })
        \\
        & =
        \Phi^\ast\bracket({
          F_2^K
        })    
          + 
          \cdots
        \\
        & =
        2 F_2^C
        +  
        \cdots
        \\[+3pt]
        \Phi^\ast H_3^K 
        & = H_3^C
        \mathrlap{\,,}
        \end{aligned}
      \end{equation}
      where the ellipses denote higher degree terms. 
    
      Since pullback $\Phi^\ast$ respects wedge products and the differential, this gives
      \begin{equation}
        \begin{alignedat}{2}
          \mathrm{d}\, 
          \bracket({
            \Phi^\ast \, F_4^K 
          })
          & = 
          \Phi^\ast \mathrm{d}\, F_4^K
          \\
          & 
          =
          \Phi^\ast\bracket({
            H_3^K \wedge F_2^K
          })
          &\;\;&
          \text{\small by \cref{SullivanModelOfTwistedK}}
          \\
          & =
          \Phi^\ast\bracket({
            H_3^K
          })
          \wedge
          \Phi^\ast\bracket({
            F_2^K
          })
          \\
          & 
          =
          2 \, H_3^C \wedge F_2^C
          &&
          \text{\small by \cref{RationalChernChOfPhi}}.
        \end{alignedat}
      \end{equation}
      But the only elements in the minimal model \cref{TheFullGaussLaw} of $\mathrm{Cyc}_1 S^4$ with this property are
      \begin{equation}
        \Phi^\ast F_4^K
        =
        2 F_4^C
        +
        \lambda' F_2^C \wedge F_2^C
        \,,
        \;\;\;\;
        \lambda' \in \mathbb{Q}
        \mathrlap{\,.}
      \end{equation}
      Iterating this kind of argument in increasing degrees gives:
      \begin{subequations}
        \begin{align}
          \Phi^\ast F_4^K
          & = 
          2 F_4^C
          + 
          \mathcal{O}\bracket({F_2^C})
          \\
          \Phi^\ast F_6^K
          & = 
          2 F_6^C
          + 
          \mathcal{O}\bracket({F_2^C})
          \\
          \Phi^\ast F_8^K
          & = 
          2 F_8^C
          + 
          \mathcal{O}\bracket({F_2^C})
          +
          \nu \, 
          \mathrm{d} 
          H_7^C
          \\
          \label
          {PullbackOfF10AlongPhi}
          \Phi^\ast F_{10}^K
          & = 
          2 F_{10}^C
          + 
          \mathcal{O}\bracket({
            F_2^C,
            H_3^C
          })
          +
          \kappa 
          \bracket({
            F_4^C \wedge F_6^C
            +
            2
            H_3^C \wedge H_7^C
            -
            3 F_2^C \wedge F_8^C
          })
          \mathrlap{\,,}
        \end{align}
      \end{subequations}
      where $\mathcal{O}({\cdots})$ denotes the ideal generated by the listed generators,
      where $\nu \in \mathbb{Q}$,
      and where $\kappa \in \mathbb{Q}$ in \cref{PullbackOfF10AlongPhi} is a point in the line of those closed elements in the minimal model that do not vanish with $F_2^C$, which is possible due to the nonlinear differential of $H_7$, cf. \cref{TheH7GaussLaw,The-full-IIA-EM-Gauss-laws}.
      
      Since $F_2^C$ and $H_3^C$ do vanish upon pullback $\tilde j^\ast$ \cref{HomotopyFibersForDMWSector} to the DMW sector, 
      we conclude that:
      \begin{equation}
        \label{ChernCharacterOfJAstPhi}
        \begin{aligned}
          \mathrm{ch}_{\leq 10}\bracket({
            \tilde j^\ast \Phi
          })
          &
          =
            2 \widetilde F_4^C 
          + 2 \widetilde F_6^C 
          + 2 \widetilde F_8^C  
          + \bracket({
            2 \widetilde F_{10}^C
            + 
            \kappa 
            \,
            \widetilde F_4^C \wedge \widetilde F_6^C
          })
          \\
          &
          \phantom{{}={}}
          +
          \mathrm{d}\bracket({
            \nu
            \widetilde H_7^C
          })
          \mathrlap{\,,}
        \end{aligned}
      \end{equation}
      where the tilded generators 
      appear via the transformation \cref{WeakEquivalenceBetweenRelativeAndAbsoluteModel} in \cref{RelativeToAbsoluteMinimalModel} of \cref{RationalModelsOfTildeL1S4}. The last summand in \cref{ChernCharacterOfJAstPhi} is exact and has no effect on the following analysis of cohomology classes.
    
    \item 
      The Bott generator \cref{BottGeneratorOn4Sphere} has Chern character concentrated in degree 4 of the form: 
      \begin{equation}
        \mathrm{ch}\bracket({
          \mathrm{ev}_0^\ast
          \beta
        })
        =
        \widetilde F_4^C
        \mathrlap{\,,}
      \end{equation}
      since it is pulled back from $S^4$ and indivisible there. With \cref{ChernCharacterOfJAstPhi} this implies the claim \cref{ChernCharacterOfTildePhi} for $\mathrm{ch}\bracket({ \tilde j^\ast \Phi - \mathrm{ev}_0^\ast \beta })$. Analogously it implies that
      \begin{equation}
        \label{ChernCharOfJastPhiMinusTwoBeta}
        \mathrm{ch}\bracket({
          \tilde j^\ast \Phi 
            - 
          2 \mathrm{ev}_0^\ast \beta
        })
        = 
        2 \widetilde F_6^C + \cdots
        \mathrlap{\,.}
      \end{equation}
    
     \item 
     \label{LeadingCherCharComponentIsIntegral}
     The Chern character induces a morphism of Atiyah--Hirzebruch spectral sequences over $\widetilde {L_1 S^4}$, from K-theory to even-periodic rational cohomology. 
     Inspection of this morphism reveals that here the lowest-degree non-vanishing component of a Chern character is an \emph{integral} form, provided that component has degree at most 10. With \cref{ChernCharOfJastPhiMinusTwoBeta} this implies that $2 \widetilde F_6$ is integral. 

     \item
     Inspection of the Serre spectral sequence of $\inlinetikzcd{ \widetilde{L_1 S^4} \ar[r] \& L_1 S^4 }$ shows that $2 \widetilde F_8$ is also integral and that all three integral components thus obtained are in fact integral generators, hence establishing the claim \cref{IntegralityOfChernCharacterComponents}.

     \item 
     Finally, with $2 \widetilde F_6$ thus established to be integrally indivisible, we see from \cref{ChernCharacterOfJAstPhi} with \cref{LeadingCherCharComponentIsIntegral} that $\tilde j^\ast \Phi - 2\mathrm{ev}_0^\ast \beta$ is not divisible, hence that $\tilde j^\ast \Phi$ is not, and consequently neither is $\Phi$. This proves the final claim \cref{PhiNotDivisible}.
     \qedhere
    \end{enumerate}  
\end{proof}
    
\begin{remark}
\label[remark]
{PhysicalImpactOfDMWSectorResult}
Some physical impact of \cref{ComparisonInDMWSector}:
\begin{enumerate}

\item
The shifted lifts in \cref{ShiftedLifts} refine the lifts considered by DMW  \cref{TheDMWLift} to a  \emph{natural construction} (on those M-brane charges that appear in Cohomotopy), both in the colloquial sense and in the technical sense: $\widetilde \Phi$ is a cohomology operation \cref{NonabelianCohomologyOperation} from small-cyclified Cohomotopy to K-theory---here considered on topological charges, but after passage to \emph{differential} nonabelian cohomology \cite{SS26-HigherGauge} also on actual IR-completed fields.

\item
\label
{OnIntegralityOfChernChComponents}
Equation
\cref{IntegralityOfChernCharacterComponents} entails that K-theory classes approximating small-cyclified 4-Cohomotopy classes on globally hyperbolic 10D spacetimes are guaranteed to have integral Chern character components.%
\footnote{%
An analogous observation was made already in \cite{BuSS2021}, there for the case of fractional D-branes with charges in \emph{equivariant} K-theory of an orbifold singularity: Lifting these to equivariant Cohomotopy forces their equivariant Chern characters to have integral components.}
This may address a worry voiced by early authors, that non-integral Chern character components are a conceptual problem for RR-flux quantization in K-theory, cf. \parencites[(1.1), (2.19)]{MooreWitten2000}[\S~6]{BachasDouglasSchweigert2000}.

\item
The Bott generator \cref{BottGeneratorOn4Sphere} may be regarded as the K-theoretic unit measure of M5-brane charge already up in 11D (cf. \parencites[\S~2.3]{SS23-Mf}[\S~2.2]{SS26-Orb}). Hence the comparison map \cref{ShiftedLifts} measures K-theoretic D-brane charge relative to this backdrop.

\item 
The indivisible integrality of the 10D flux $2 \widetilde F_6$ in \cref{ChernCharacterOfTildePhi}, with its factor of 2, is reminiscent of (and probably equivalent to, under dimensional reduction) the indivisible integrality of the 11D \emph{Page charge} flux $2\bracket({ G_7 - \tfrac{1}{2} C_3 \wedge G_4 })$ found in \cite[Thm. 4.8]{FSS21-Hopf} with the ``same'' factor of 2, again under \emph{Hypothesis H}. For a possible physical interpretation cf. \cite[(3)]{FSS21-Hopf}.

\end{enumerate}

\end{remark}

Finally, we may characterize more explicitly how close the comparison maps \cref{TheComparisonMap,ShiftedLifts} are to equivalences:
\begin{corollary}[Homotopy fiber of the comparison map]
\label[corollary]{HomotopyOfTheFiber}
\begin{enumerate}
  \item
  The homotopy fiber of 
  $\widetilde \Phi$ \cref{ShiftedLifts} followed by 10-truncation,
  \begin{equation}
    \label
    {TildePhiFollowedBy10Truncation}
    \begin{tikzcd}
    \widetilde{L_1 S^4} 
      \ar[
        r, 
        "{
          \widetilde \Phi
        }"{description}
      ] 
      \ar[
        rr,
        uphordown,
        "{
          \widetilde \Phi_{\leq 10}
        }"{description}
      ]
      &
      B \mathrm{SU}
      \ar[
        r,
        ->>
      ]
      &
      \truncation{10}{B \mathrm{SU}}
    \end{tikzcd}
  \end{equation}
  has low homotopy groups 
  corresponding exactly to the NS1-brane charge and all the torsion brane charges in \cref{HomotopyOfCyc1S4Table}:
  \medskip 
  \begin{equation}
  \label
  {LowHomotopyOfFiberOfTildePhi}
    \mathllap{%
    \forall_{k \leq 9}
    \;\;\;\;
    }
    \pi_{k}
    \bracket({
      \mathrm{fib}\bracket({
        \widetilde \Phi_{\leq 10}
      })
    })
    \simeq 
    \overbrace{
    \mathrm{Tor}\bracket({
      \pi_{k}\bracket({
        \mathrm{Cyc}_1 S^4
      })
    })
    }^{ \textup{torsion branes} }
    \oplus
    \begin{cases}
     \smash{%
      \overbrace{%
      \mathbb{Z}\langle \gamma_7 \rangle%
      }^{\textup{NS1}}%
      }
      &
      \textup{ if } k = 7
      \\
      0 &
      \textup{ otherwise. }
    \end{cases}
  \end{equation}
  \item
  The homotopy fiber of 
  $\inlinetikzcd{ 
    \mathrm{Cyc}_1 S^4 
      \ar[
        rr, 
        "{
          \Phi_{ \leq 10 }
        }"
      ] 
      \&\& 
      \truncation
        {10}
        { B \mathrm{U} } 
      \sslash 
      B \mathrm{U}(1)
  }$ \cref{TheComparisonMap} for $n = 2$
  has low homotopy groups equal to those of \cref{LowHomotopyOfFiberOfTildePhi} plus two copies of $\mathbb{Z}_{/2}$ reflecting the multiplication \cref{TheChargeMultiplication}:
  \begin{equation}
  \label{LowHomotopyOfFiberOfPhi}
    \mathllap{%
    \forall_{k \leq 9}
    \;\;\;\;
    }
    \pi_k\bracket({
      \mathrm{fib}\bracket({
        \Phi_{ \leq 10 }
      })
    })
    \simeq
    \pi_k\bracket({
      \mathrm{fib}\bracket({
        \widetilde \Phi_{ \leq 10 }
      })
    })
    \oplus
    \begin{cases}
      \mathbb{Z}_{/2}
      &
      \textup{ if }
      k \in \{1,3\}
      \\
    0
    &
    \textup{ otherwise. }
    \end{cases}
  \end{equation}
\end{enumerate}
In particular, both maps induce isomorphisms on rational homotopy groups in degrees $\leq 10$ away from the degree 7 of the NS1-brane charge \textup{(cf. \cref{RationalHomotopyGroups})}.
\end{corollary}
\begin{proof}[Proof outline]
  Compare the long exact sequence of homotopy groups using \cref{HomotopyOfCyc1S4}. Observe that the Kronecker pairing between (integration over) the Hurewicz image $h(-)$ of the spherical generators in \cref{HomotopyOfCyc1S4Table} against the universal fluxes \cref{WeakEquivalenceBetweenRelativeAndAbsoluteModel} is:
  \begin{subequations}
    \begin{align}
      \label{KroneckerPairingInDegree4}
      \int_{h(\gamma_4)} 
      \widetilde F_4
      & 
      =
      \pm 1
      \\
      \label{KroneckerPairingInDegree6}
      \int_{h(\gamma_6)}
      \widetilde F_6
      & 
      =
      \pm \tfrac{1}{2}
      \\
      \label{KroneckerPairingInDegree8}
      \int_{h(\gamma_8)}
      \widetilde F_8
      & 
      =
      \pm \tfrac{1}{2}
      \\
      \label{KroneckerPairingInDegree10}
      \int_{h(\gamma_{10})}
      \widetilde F_{10}
      & 
      =
      \pm \tfrac{1}{2}
      \mathrlap{\,.}
    \end{align}
  \end{subequations}
  Here \cref{KroneckerPairingInDegree6} follows with the differential \cref{d4OfIota6} and \cref{KroneckerPairingInDegree8,KroneckerPairingInDegree10} via Whitehead products.
  With \cref{ChernCharacterOfTildePhi} this means that $\widetilde \Phi$ is an isomorphism on the free parts of $\pi_4$, $\pi_6$ and $\pi_8$ and surjective on $\pi_{10}$.
\end{proof}

\begin{SCfigure}[.9][htb]
\caption{\label{ComparisonAndFiber}%
  The comparison map $\widetilde \Phi$ \cref{ShiftedLifts} exhibits, in the untwisted DMW sector, K-theory as an abelian approximation to small-cyclified 4-Cohomotopy that misses \cref{LowHomotopyOfFiberOfTildePhi} NS1-brane charge (necessarily by \cref{AbelianCohDoesNotQuantizeNonlinear}, due to the non-linear Gauss law of its $H_7$-flux) as well as torsion branes\protect\footnotemark\ predicted by Hypothesis H. Generally, $\Phi$ \cref{TheComparisonMap} exhibits twisted K-theory as such an approximation, moreover missing \cref{LowHomotopyOfFiberOfPhi} a $\mathbb{Z}_{/2}$ ambiguity \cref{TheChargeMultiplication} in D6- and D4-brane charge.}
\centering
\adjustbox{rndfbox=4pt}{
$
  \begin{tikzcd}[
    row sep=45pt,
    column sep=40pt
  ]
    &[63pt]
    \mathrm{fib}\bracket({
      \widetilde \Phi_{\leq 10}
    })
    \ar[
      d
    ]
    \ar[
      r
    ]
    \ar[
      dr,
      phantom,
      "{ \lrcorner }"{pos=.1}
    ]
    &[-25pt]
    \ast
    \ar[
      d,
      "{ 0 }"
    ]
    \\
    X^9
    \ar[
      r,
      "{ 
        \text{IIA brane charges}
      }"{pos=.4, yshift=-1pt},
      "{
         \text{by Hypothesis H...}      
      }"{swap,pos=.4, yshift=+1pt}
    ]
    \ar[
      rr,
      downhorup,
      "{ 
        \text{...approximated in K-theory}
      }"'
    ]
    \ar[
      ur,
      dashed,
      shift left=3pt,
      "{ 
          \text{
            ...unseen in K-theory
          }
      }"{sloped, pos=.65},
      "{
          \text{(NS1 and torsion)}
      }"{swap, sloped, pos=.65}
    ]
    &
    \widetilde{ L_1 S^4 }
    \ar[
      r,
      "{ 
        \widetilde \Phi_{\leq 10}
      }"
    ]
    &
    \truncation{10}{B \mathrm{SU}}
  \end{tikzcd}
$
}
\end{SCfigure}
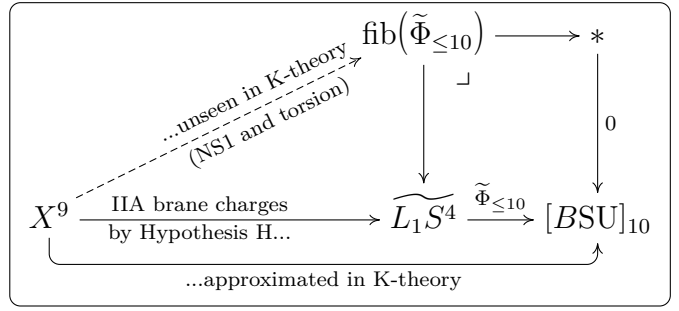

\footnotetext{%
  The torsion brane charges indicated as missed by K-theory in \cref{ComparisonAndFiber} are those visible already on near-horizon spacetimes of the form $\mathbb{R}^{9} \setminus \mathbb{R}^p$ (cf. \cref{MeasuringDBraneCharge}) as listed in \cref{HomotopyOfCyc1S4Table}. On more general spaces, K-theory also sees some torsion charges, cf. \cref{TorsionMBraneCharges}.
}

\section
{Conclusion}

We close by briefly putting our results in perspective. 
To recall, the open problem we are addressing is how to combine two key conjectures/hypotheses about string theory, namely:
\begin{enumerate}
\item 
the \emph{Hypothesis K} that type IIA flux/charge is quantized in twisted K-theory, 

\item 
the \emph{M-theory conjecture}
that type IIA is the dimensional reduction on a small circle of a completion of 11D supergravity, 
\end{enumerate}
and the ambition is to combine these \emph{top-down}: Since flux quantization is IR-completion, we ask for an IR-completion of 11D SuGra whose dimensional reduction on small circle fibers \emph{implies} that type IIA charge is quantized in twisted K-theory.

\subsection
{Observations}
We began with the following initial observations on the general problem of deriving flux quantization in 10D from 11D:
\begin{enumerate}
\item \textbf{Proper electromagnetic flux quantization in nonabelian cohomology.}
These are observations which we have already laid out elsewhere, but they are worth repeating here:
\begin{enumerate}
\item
One should not only quantize the magnetic sector in 11D ($G_4$, M5-brane charge) but also \textbf{include the electric sector} ($G_7 + \cdots$, M2-brane charge). 

(This holds on general grounds, but particularly here since quantization of electric fluxes is already part of the traditional proposal---\emph{Hypothesis K}: that type IIA flux is quantized in K-theory---which is to be derived here from 11D.)

\item
But since the Gauss law constraint on the electric $G_7$-flux in 11D is non-linear, the \textbf{quantized charge must be in nonabelian cohomology}, beyond the traditional realm of abelian generalized cohomology theories like K-theory.

(Again, this issue was already visible all along in 10D: There the electric $H_7$-flux has a non-linear Gauss law (the direct dimensional reduction of the situation with $G_7$ in 11D), and this Gauss law is necessarily dropped by the traditional \emph{Hypothesis K}.)

\item The universal (namely with minimal classifying cell complex) choice of nonabelian cohomology  for electromagnetic flux quantization admissible in 11D gives 
\textbf{11D SuGra quantized in
4-Cohomotopy}. This universal choice is what we consider here---\emph{Hypothesis H}.

(There are other admissible choices in 11D, but substantially different choices will break our derivation of K-theory in 10D. Conversely, readers who already believe in \emph{Hypothesis K} for type IIA may read our result here in the opposite direction as a further consistency check that \emph{Hypothesis H} is the right IR-completion of 11D SuGra for the purpose of completing it to M-theory.)

\end{enumerate}

\item \textbf{Dimensional charge reduction on small M-circles via Morse--Bott theory.}
\label{IdeaOfSmallVariation}
This is our key novel observation, not discussed before:
\begin{enumerate}
\item
The notion that IIA is the dim-reduction of M-theory on a \emph{small} circle needs to be extended to the topological IR-completed sector of brane charges (where metric circumference is not a meaningful size indicator).

\item
Our earlier theorem---that dimensional circle reduction of quantized charges means passage to the \emph{cyclic loop space} of the charge classifying space---allows us to identify such topological \emph{small circle} reductions as factoring charge classifying maps through stages of the Morse--Bott filtration of loops of bounded variation.

(This requires the charge classifying space to admit the auxiliary structure of a Riemannian manifold---a strong constraint, but satisfied by \emph{Hypothesis H}.)

\item 
Under \emph{Hypothesis H} in 11D there is visible exactly one non-trivial stage in the Morse--Bott filtration, namely the first one, thus unambiguously called the \emph{small-cyclified loop space of the 4-sphere}. This is equivalently the space of round embeddings of the M-circle into $S^4$, with degenerate embeddings (collapsed circles) included as shapes of reparameterization orbifold singularities $\ast \sslash S^1 \simeq B S^1$.

\end{enumerate}

\end{enumerate}

\subsection
{Derivations}

From these premises we rigorously derived the following (proofs outlined here, with substantial spectral sequence runs left to the reader):
\begin{enumerate}[resume]

\item 
\textbf
{Completion of 10D charge content.}
  The resulting small circle reduction, from 11D SuGra IR-completed in 4-Cohomotopy, gives \textbf{IIA IR-completed in small-cyclified 4-Cohomotopy},
  which is the first nontrivial 10D flux quantization that reflects strictly all magnetic and electric NS/RR-fluxes:
  \begin{enumerate}
    \item $H_7$ appears with its nonlinear Gauss law quantized, whence 
    \textbf{string (NS1-brane) charge is now well-quantized} (not so under the traditional \emph{Hypothesis K}).
    
    \item $F_8$ appears right away as a quantized flux, as such not ordinarily originating from 11D, but \textbf{completing the Chern character} as demanded by the traditional \emph{Hypothesis K}.

    \item
    The corresponding black \textbf{D0-brane charge is now well defined}, which does not ordinarily descend from 11D. 
    
    (Beware that the D0-branes in BFSS/BMN matrix models are \emph{probe branes} on flat or pp-wave backgrounds, while here we are dealing with their backreacted black cousins, which may be underappreciated.)

    \item 
    On non-compact Cauchy surfaces, this D0-charge
    turns out to be in integral 8-cohomology, but generally \textbf{D0-brane charge is in stable 8-Cohomotopy twisted by D6-brane charge} (the M-circle bundle class) mod 2.

    \item The \textbf{D0-brane tadpole is precisely the obstruction} to restricting an unconstrained circle reduction to a small circle reduction of M-brane charges.
\end{enumerate}

\item 
\textbf
{Twisted K-Theory emerges.}
This is now the key point:

This small-cyclified 4-Cohomotopy theory---not postulated but derived by small circle reduction from 11D---turns out to be a nonlinear deformation (by NS1-brane charge) of rank-zero twisted K-theory, as witnessed by a comparison cohomology operation
\[
  \begin{tikzcd}
    \text{
      small-cyclified 4-Cohomotopy
    }
    \ar[
      r,
      "{ \Phi }"
    ]
    &
    \text{
      rank-0
      twisted K-theory
    }
    \mathrlap{\,,} 
  \end{tikzcd}
\]
which exhibits an abelian approximation, in that:
\begin{enumerate}
\item
in general it:
\begin{enumerate}

  \item 
  bijects on NS5-brane charge and $H_3$-flux,

  \item 
  surjects onto an index-2 sublattice of 
  D6/D4-charges with their $F_2$/$F_4$-fluxes,

  \item 
  forgets $H_7$-flux and NS1-charge (invisible to twisted K-theory)

  as well as some torsion charges,
\end{enumerate}

\item 
in the DMW sector of vanishing NS5- and D6-charges it improves to $\widetilde \Phi$ which:
\begin{enumerate}
  \item 
  has universal untwisted Chern character
  \begin{equation}
  \label{ChernCharacterInConclusion}
    \mathrm{ch}\bracket({
      \widetilde{\Phi}
    })
    \,=\,
    \grayunderbrace
      {\widetilde F_4}
      {1} 
      \,+\, 
    \grayunderbrace
      {2 \widetilde F_6}
      {1} 
      \,+\, 
    \grayunderbrace
      {2 \widetilde F_8}
      {1} 
      \;+\, 
      \substack{
        \text{irrelevant terms}
        \\
        \text{of degrees $\geq 10$}
      }
      \mathrlap{\,,}
  \end{equation}
  where 
  \begin{equation}
    \widetilde F_{2k} 
    = 
    F_{2k}
    \,-\,
    B_2 \wedge F_{2k-2}
    \,+\, 
    \tfrac{1}{2}
    B_2 \wedge B_2 \wedge  F_{2k-4}
    \,+\,
    \cdots
  \end{equation}
  are the untwisted RR-flux densities,

  \item where---as indicated below the braces in \cref{ChernCharacterInConclusion}---the relevant first three summands are generators of the integral lattice inside their cohomology group 
  (so $\widetilde F_6$, $\widetilde F_8$ are actually half-integral),

  \item 
  implying that the K-theoretic Chern characters underlying a small-cyclified 4-Cohomotopy class on a Cauchy surface $X^9$ are guaranteed to have integral components.
  
  (The failure of this property in plain K-theory had been the cause of some head-scratching in the literature on \emph{Hypothesis K}.)
\end{enumerate}

\item This $\widetilde \Phi$ improves the \emph{ad hoc} DMW lifts---of integral D4-brane charge to K-theory---to a natural transformation (on those D4-charges that survive in Cohomotopy), and $\Phi$ generalizes this beyond the DMW sector to twisted K-theory.

\end{enumerate}

\end{enumerate}

\subsection
{Vistas}
\label{Vistas}

A number of follow-up questions arise; we briefly mention a few:
\begin{enumerate}

\item
\textbf
{The Half-Integral Shift of $G_4$.}
For brevity we have disregarded here the tangential twisting of 4-Cohomotopy that is part of \emph{Hypothesis H} proper \parencites{FSS20-H}{FSS21-Hopf}, and which implies \cite[Prop. 3.13]{FSS20-H} the half-integral shift in the $G_4$-flux that is expected in M-theory \cite{Witten1997Flux}:
\begin{equation}
  \inlinetikzcd{
  \bracket[{
    G_4 
      + 
    \tfrac{1}{4}p_1(\omega)
  }]
  \in
  H^4\bracket({
    X^{10}
    ;\,
    \mathbb{Z}
  })
  \ar[r]
  \&
  H^4\bracket({
    X^{10}
    ;\,
    \mathbb{R}
  })
  \mathrlap{\,.}
  }
\end{equation}
We expect that the statements and proofs in \cref{Results} generalize with only mild adjustments to this situation, but it remains to be worked out.

\item
\textbf{The Partition Function.}
The DMW result has two parts: 
\begin{enumerate}
\item 
For every D4-brane charge $f_4$ satisfying $\mathrm{Sq}^3_{\mathbb{Z}} \, f_4 = 0$, and every lift $Q$ to K-theory through the AHSS as in \cref{TheDMWLift}, the corresponding \emph{M-theory phase} of $f_4$ and \emph{K-theory phase} of $Q$ agree \cite[(5.4)]{DMW2000} up to a fixed correction factor.

\item 
The \emph{M-theory partition function} over $f_4 \in H^4\bracket({ X^{9}; \mathbb{Z} })$ localizes onto the constraint subset $\bracketmid\{{ f_4}{\mathrm{Sq}^3_{\mathbb{Z}}\, f_4 = 0 }\}$ (the \emph{integral equation of motion}), and there it coincides with the type IIA partition function over K-theory charges $Q$.
\end{enumerate}

In contrast, under \emph{Hypothesis H} the \emph{integral equation of motion} is absorbed, cf.  \cref{IntegralEquationOfMotion}, into the Cohomotopical flux-quantization (thereby avoiding the notion that M-theory has a path-integral definition). At the same time, this makes $f_4$ be in the image of the Cohomotopical character map \cite[Ex. 6.11, cf. \S~12]{FSS23-Char}
\begin{equation}
\label{CohomotopicalCharacter}
  \inlinetikzcd{
  H^1\bracket({
    X^9
    ;\,
    \Omega S^4
  })
  \defneq 
  \pi^4\bracket({
    X^9
  })
  \ar[
    rr,
    "{
      \mathrm{ch}^{S^4}
    }"
  ]
  \&\&
  H^4\bracket({
    X^{9};
    \mathbb{Z}
  })
  \mathrlap{\,,}
  }
\end{equation}
which is in general not surjective (nor injective, nor with additive image):
The first obstruction to lifting through \cref{CohomotopicalCharacter} is $\mathrm{Sq}^2 \rho_2 f_4$, which is stronger than $\mathrm{Sq}^3_{\mathbb{Z}} f_4 \defneq \beta \mathrm{Sq}^2 \rho_2 f_4$.

Therefore, under \emph{Hypothesis H}, the ``M-theory partition function'' is nevertheless over a different domain than the ``K-theory partition function'' and the two may no longer coincide. Since this image need not be an additive subgroup, it cannot directly replace the charge lattice in DMW's theta-function construction.

However, also under \emph{Hypothesis H} this is now the wrong comparison to make anyway: 
One should instead compare, on the IIA side, to a partition function not over K-theory classes but over small-cyclified 4-Cohomotopy classes.  While these putative modified partition functions (over 4-Cohomotopy in 11D and over small-cyclified 4-Cohomotopy in 10D) may be more complicated, their matching should no longer be a surprising coincidence but a structural consequence, as now they are manifestly related by dimensional reduction (along \emph{small M-circles}).

\item
\textbf{Type IIB and T-Duality.}
Here we have focused entirely on type IIA. Beyond that one wants to ask:
\begin{quote}\emph{%
  What is the nonabelian refinement of flux/charge quantization in type IIB that is to twisted $\mathrm{KU}^1$ as small-cyclified Cohomotopy is to twisted $\mathrm{KU}^0$?
}
\end{quote}
Since type IIB is related to M-theory only via its T-duality  with type IIA, key guidance on this matter should come from addressing the related question:
\begin{quote}\emph{%
How does topological T-duality \textup{(cf. \cref{CycB3ZInTopTDuality})} lift from twisted K-theory to small-cyclified 4-Cohomotopy and whatever its type IIB dual may be?
}
\end{quote}
This will need a new idea, because inspection shows%
\footnote{%
  We are grateful to Grigorios Giotopoulos for highlighting this point.
}
that the ordinary mechanism of topological T-duality from \cref{TopTDualityViaCyclification} in \cref{CycB3ZInTopTDuality} cannot generalize to the presence of NS1-charges and their $H_7$-flux: Already rationally, cyclification cannot relate non-linear Gauss laws like \cref{TheH7GaussLaw} in this way \cref{TopTDualityViaCyclification}.

\item
\label{SDualityAndFTheory}
\textbf
{S-duality and F-theory.}
The result of \cref{DimReductionViaCyc}---that under circle-dimensional reduction charges are classified by the cyclification $\mathrm{Cyc}(\mathcal{A})$ of the original classifying space $\mathcal{A}$---applies, as cited there, more generally to reductions on $n$-tori $T^n$ for all $n \in \mathbb{N}$: If charges in higher dimensions are classified by $\mathcal{A}$ then the same charges seen on the base of a $T^n$-principal bundle are classified by the \emph{$n$-toroidification} $\toroidification^n(\mathcal{A}) := \mathrm{Map}\bracket({ T^n, \mathcal{A} }) \sslash T^n$. For $n = 1$ this reduces to the cyclification; for $n = 2$ we have
\begin{equation}
\label
{TorViaCyCyc}
  \begin{tikzcd}
    \toroidification\bracket({
      \mathcal{A}
    })
    :=
    \toroidification^2\bracket({
      \mathcal{A}
    })
    :=
    \mathrm{Map}\bracket({
      T^2
      ,
      \mathcal{A}
    })
    \sslash T^2
    \,\simeq\,
    ({
      \mathrm{Cyc}
      \,
      \mathrm{Cyc}
      \,
      \mathcal{A}
    })
    \;\;
    \underset
      {\mathclap{
        \mathrm{Cyc}(B S^1)
      }}
      {\times}
    \;\;
    B T^2
    \mathrlap{\,,}
  \end{tikzcd}
\end{equation}
etc.
Hence, we are led to ask:
\begin{quote}\emph{What is the general definition of subspaces
$
  \toroidification^n_1\bracket({S^4})
  \subset
  \toroidification^n\bracket({S^4})
$
that qualify as \emph{small $n$-toroidification}, in generalization of  small cyclification $\mathrm{Cyc}_1\bracket({S^4}) \defneq \toroidification^1_1\bracket({S^4})$?}
\end{quote}

For example, one finds (in the notation of \cref{ReductionTo9D}) that flux quantization of 9D SuGra (cf. \parencites[\S 2.2]{OrtinEtAl2011}) understood as 11D SuGra reduced on a torus (cf. \parencites{Aspinwall1996}{Schwarz1996}) with (only) the M-circle taken to be small, therefore, with \cref{TorViaCyCyc}, classified by
\begin{equation}
  \inlinetikzcd{
  \mathrm{Cyc}\bracket({
    \mathrm{Cyc}_{\mathcolor{purple}{1}}
    \bracket({
      S^4
    })})
    \;\;
    \underset
      {\mathclap{
        \mathrm{Cyc}(B S^1)
      }}
      {\times}
      \;\;
    B T^2
    \ar[
      rr,
      hook
    ]
    \&\&
    \toroidification\bracket({
      S^4
    })
    \mathrlap{\,,}
  }
\end{equation}
\begin{enumerate}
  \item 
  sees the flux $F_8$ and its descendant $F_7^{\mathrm{A}}$, via our results \cref{The-full-IIA-EM-Gauss-laws,DimReductionViaCyc}, cf. \cref{F8Appears},
  \item
  \label{F7MDoesNotAppear}
  but next misses its $\mathrm{SL}(2,\mathbb{Z})$-partner $F_7^{\mathrm{M}}$-flux, in iteration of the way that  $\mathrm{Cyc}\bracket({S^4})$ had missed $F_8$ (by \cref{ObtainingGaussLawsForCycS4}).
\end{enumerate}

The analysis of potential candidates 
$
  \inlinetikzcd{
    \toroidification_1 S^4
    \ar[
      r,
      hook
    ]
    \&
    \toroidification\, S^4
  }
$
for \emph{small toroidification}, 
that would fix this \cref{F7MDoesNotAppear}---in the way that small cyclification $\mathrm{Cyc}_1 S^4$ fixed M$\to$IIA-reduction (in \cref{Results})---turns out to be considerably more involved than the analysis of $\mathrm{Cyc}_1\bracket({S^4})$ that we presented here. 

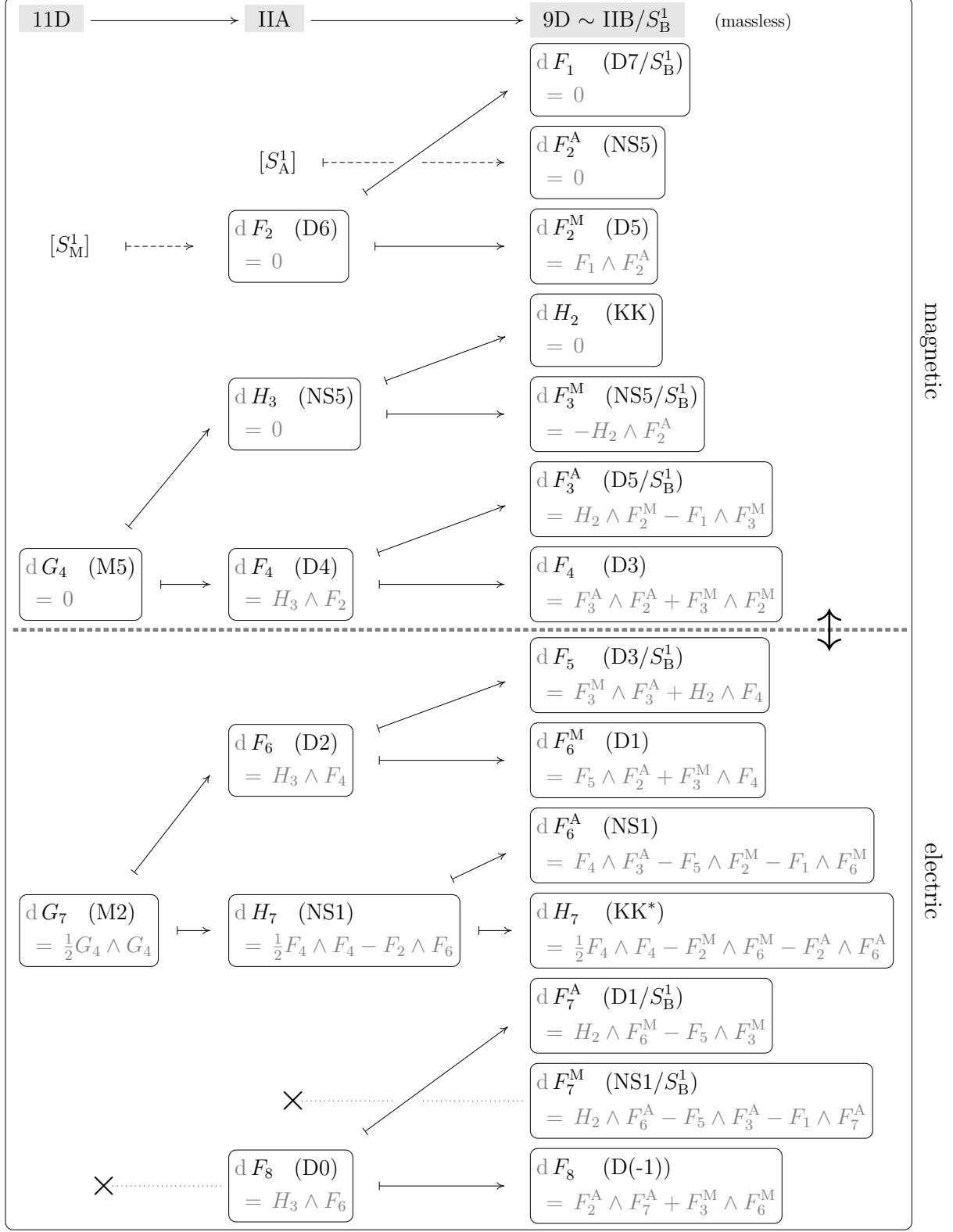
\begin{figure}[p]
\caption{\label{ReductionTo9D}%
\textbf{Flux species, their brane sources and their Bianchi/Gauss laws, across dimensions.} Solid arrows denote (double) dimensional reduction, dashed arrows pick up Chern forms of circle bundles. Reflection---at the middle dashed line, of occupied rows in any column---is $\pm$ (dilated) Hodge duality in that spacetime dimension. The very bottom Hodge duals do \emph{not} ($\times$) originate from naive dimensional reduction; but with \emph{Hypothesis H} they do appear via the \emph{small circle} reduction considered in \cref{MorseBottRegularity} and in \cref{Vistas} \cref{SDualityAndFTheory}.
}
\centering
\adjustbox{
  rndfbox=4pt,
  scale=.9
}{
  \begin{tikzcd}[
    ampersand replacement=\&,
    column sep=30pt,
    row sep=-2pt,
    /tikz/column 1/.append style={anchor=base west},
    /tikz/column 2/.append style={anchor=base west},
    /tikz/column 3/.append style={anchor=base west}
  ]
    \colorbox{lightgray}{ 11D }
    \ar[r]
    \& 
    |[xshift=9pt]|
    \colorbox
      {lightgray}
      { IIA }
    \ar[r]
    \& 
    \colorbox
      {lightgray}
      { 9D $\sim$ IIB/$S^1_{\mathrm{B}}$ }
    \mathrlap{
      \;\;\; \scalebox{.8}{(massless)}
    }
    \\
    \\
    \&\&
    \adjustbox{rndfbox=4pt}{$
    \begin{aligned}
      & 
      \mathcolor{gray}{\mathrm{d}}\, 
      F_1^{\phantom{\mathrm{A}}}
      \;\;
      \text{(D7$/S^1_\mathrm{B}$)}
      \\
      & 
      \mathcolor{gray}{= 0}
      \;\;\;\;\;\;
    \end{aligned}
    $}
    \\
    \&
    \;\;\;
    [S^1_\mathrm{A}]
    \ar[
      r,
      dashed,
      |->,
      shorten=10pt
    ]
    \&
    \adjustbox{rndfbox=4pt}{$
    \begin{aligned}
      & 
      \mathcolor{gray}{\mathrm{d}}\, F^{\mathrm{A}}_2
      \;\;
      \text{(NS5)}
      \\
      & 
      \mathcolor{gray}{= 0} 
      \;\;\;\;\;\;
    \end{aligned}
    $}
    \\
    \;\;\;
    [S^1_\mathrm{M}]
    \ar[
      r,
      |->,
      dashed,
      shorten=15pt
    ]
    \&
    \adjustbox{rndfbox=4pt}{$
    \begin{aligned}
      & 
      \mathcolor{gray}{\mathrm{d}}\, F_2
      \;\;
      \text{(D6)}
      \\
      & \mathcolor{gray}{= 0}
      \;\;\;\;\;\;
    \end{aligned}
    $}
    \ar[
      r,
      |->,
      shorten=10pt
    ]
    \ar[
      uur,
      |->,
      crossing over,
      shorten=10pt,
      end anchor=west
    ]
    \&
    \adjustbox{rndfbox=4pt}{$
    \begin{aligned}
      & 
      \mathcolor{gray}{\mathrm{d}}\, F^{\mathrm{M}}_2
      \;\;
      \text{(D5)}
      \\
      & 
      \mathcolor{gray}{= 
      F_1 \wedge F_2^{\mathrm{A}}}
    \end{aligned}
    $}
    \\
    \&\&
    \adjustbox{rndfbox=4pt}{$
    \begin{aligned}
      & 
      \mathcolor{gray}{\mathrm{d}}\, 
      H_2^{\phantom{\mathrm{A}}}
      \;\;
      \text{(KK)}
      \\
      & 
      \mathcolor{gray}{= 0}
      \;\;\;\;\;\;
    \end{aligned}
    $}
    \\
    \&
    \adjustbox{rndfbox=4pt}{$
    \begin{aligned}
      & 
      \mathcolor{gray}{\mathrm{d}}\, H_3
      \;\;
      \text{(NS5)}
      \\
      & 
      \mathcolor{gray}{= 0}
      \;\;\;\;\;\;
    \end{aligned}
    $}
    \ar[
      r,
      |->,
      shorten=10pt
    ]
    \ar[
      ur,
      |->,
      shorten=10pt,
      end anchor=west
    ]
    \&
    \adjustbox{rndfbox=4pt}{$
    \begin{aligned}
      & 
      \mathcolor{gray}{\mathrm{d}}\, 
      F_3^{\mathrm{M}}
      \;\;
      \text{(NS5/$S^1_{\mathrm{B}}$)}
      \\
      & 
      \mathcolor{gray}{= 
      - H_2 \wedge F_2^{\mathrm{A}}}
    \end{aligned}
    $}
    \\
    \&\&
    \adjustbox{rndfbox=4pt}{$
    \begin{aligned}
      & 
      \mathcolor{gray}{\mathrm{d}}\, F_3^{\mathrm{A}}
      \;\;
      \text{(D5/$S^1_{\mathrm{B}}$)}
      \\
      & 
      \mathcolor{gray}{= 
      H_2 \wedge F_2^{\mathrm{M}}
      -
      F_1 \wedge F_3^{\mathrm{M}}
      }
    \end{aligned}
    $}
    \\
    \adjustbox{rndfbox=4pt}{$
    \begin{aligned}
      & 
      \mathcolor{gray}{\mathrm{d}}\, 
      G_4
      \;\;
      \text{(M5)}
      \\
      & 
      \mathcolor{gray}{= 0}
      \;\;\;\;\;\;
    \end{aligned}
    $}
    \ar[
      uur,
      |->,
      shorten=10pt,
      end anchor=west
    ]
    \ar[
      r,
      |->,
      shorten=6pt
    ]
    \&
    \adjustbox{rndfbox=4pt}{$
    \begin{aligned}
      & 
      \mathcolor{gray}{\mathrm{d}}\, 
      F_4
      \;\;
      \text{(D4)}
      \\
      & 
      \mathcolor{gray}{= H_3 \wedge F_2}
    \end{aligned}
    $}
    \ar[
      r,
      |->,
      shorten=10pt
    ]
    \ar[
      ur,
      |->,
      shorten=10pt,
      end anchor=west
    ]
    \&
    \adjustbox{rndfbox=4pt}{$
    \begin{aligned}
      & 
      \mathcolor{gray}{\mathrm{d}}\, 
      F_4^{\phantom{\mathrm{A}}}
      \;\;
      \text{(D3)}
      \\
      & 
      \mathcolor{gray}{= 
      F_3^{\mathrm{A}}
      \wedge 
      F_2^{\mathrm{A}}
      +
      F_3^{\mathrm{M}}
      \wedge 
      F_2^{\mathrm{M}}
      }
    \end{aligned}
    $}
    \\[+2pt]
    {}
    \ar[
     rrr,
     -,
     dashed,
     gray,
     shorten <=-7pt,
     shorten >=26pt,
     line width=2,
     "{
        \leftrightarrow
     }"{
       black,
       rotate=-90, 
       scale=3,
       pos=.83},
    ]
    \&\&\&
    {}
    \\[-4pt]
    \&\&
    \adjustbox{rndfbox=4pt}{$
    \begin{aligned}
      & 
      \mathcolor{gray}{\mathrm{d}}\, 
      F_5^{\phantom{\mathrm{A}}}
      \;\;
      \text{(D3/$S^1_{\mathrm{B}}$)}
      \\
      & 
      \mathcolor{gray}{= 
      F_3^{\mathrm{M}}
      \wedge
      F_3^{\mathrm{A}}
      +
      H_2 \wedge F_4
      }
    \end{aligned}
    $}
    \\
    \&
    \adjustbox{rndfbox=4pt}{$
    \begin{aligned}
      & 
      \mathcolor{gray}{\mathrm{d}}\, 
      F_6
      \;\;
      \text{(D2)}
      \\
      & 
      \mathcolor{gray}{= 
      H_3 \wedge F_4}
    \end{aligned}
    $}
    \ar[
      ur,
      |->,
      shorten=10pt,
      end anchor=west
    ]
    \ar[
      r,
      |->,
      shorten=10pt
    ]
    \&
    \adjustbox{rndfbox=4pt}{$
    \begin{aligned}
      & 
      \mathcolor{gray}{\mathrm{d}}
      \, 
      F_6^{\mathrm{M}}
      \;\;
      \text{(D1)}
      \\
      & 
      \mathcolor{gray}{= 
      F_5 \wedge F_2^{\mathrm{A}}
      +
      F_3^{\mathrm{M}} \wedge F_4
      }
    \end{aligned}
    $}
    \\
    \&\&
    \adjustbox{rndfbox=4pt}{$
    \begin{aligned}
      & 
      \mathcolor{gray}{\mathrm{d}}\, F_6^{\mathrm{A}}
      \;\;
      \text{(NS1)}
      \\
      &
      \mathcolor{gray}{
      = 
      F_4 
        \wedge 
      F_3^{\mathrm{A}}
      -
      F_5
        \wedge 
      F_2^{\mathrm{M}} 
      -
      F_1 
        \wedge 
      F_6^{\mathrm{M}}
      }
    \end{aligned}    
    $}
    \\
    \adjustbox{rndfbox=4pt}{$
    \begin{aligned}
      & 
      \mathcolor{gray}{\mathrm{d}}\, 
      G_7
      \;\;
      \text{(M2)}
      \\
      & 
      \mathcolor{gray}{= 
      \tfrac{1}{2}
      G_4 \wedge G_4}
    \end{aligned}    
    $}
    \ar[
      uur,
      |->,
      shorten=10pt,
      end anchor=west
    ]
    \ar[
      r,
      |->,
      shorten=6pt
    ]
    \&
    \adjustbox{rndfbox=4pt}{$
    \begin{aligned}
      & 
      \mathcolor{gray}{\mathrm{d}}\, 
      H_7
      \;\;
      \text{(NS1)}
      \\
      & 
      \mathcolor{gray}{= 
      \tfrac{1}{2}
      F_4 \wedge F_4
      -
      F_2 \wedge F_6
      }
    \end{aligned}
    $}
    \ar[
      r,
      |->,
      shorten=6pt
    ]
    \ar[
      ur,
      |->,
      shorten=10pt,
      end anchor=west
    ]
    \&
    \adjustbox{rndfbox=4pt}{$
    \begin{aligned}
      & 
      \mathcolor{gray}{\mathrm{d}}\, 
      H_7^{\phantom{\mathrm{A}}}
      \;\;
      \text{(KK${}^\ast$)}
      \\
      & 
      \mathcolor{gray}{= 
      \tfrac{1}{2}
      F_4 \wedge F_4
      -
      F_2^{\mathrm{M}}
      \wedge 
      F_6^{\mathrm{M}}
      -
      F_2^{\mathrm{A}}
      \wedge
      F_6^{\mathrm{A}}
      }
    \end{aligned}
    $}
    \\
    \&\&
    \adjustbox{rndfbox=4pt}{$
    \begin{aligned}
      & 
      \mathcolor{gray}{\mathrm{d}}\, F_7^{\mathrm{A}}
      \;\;
      \text{(D1/$S^1_{\mathrm{B}}$)}
      \\
      & 
      \mathcolor{gray}{= 
      H_2 
        \wedge 
      F_6^{\mathrm{M}}
      -
      F_5
        \wedge 
      F_3^{\mathrm{M}}
      }
    \end{aligned}
    $}
    \\
    \&
    \;\;\;\;\;
    \adjustbox{
      scale=1.7,
      raise=-2pt
    }{$\times$}
    \ar[
      r,
      -,
      dotted,
      shorten <=-7pt
    ]
    \&
    \adjustbox{rndfbox=4pt}{$
    \begin{aligned}
      & 
      \mathcolor{gray}{\mathrm{d}}\, F_7^{\mathrm{M}}
      \;\;
      \text{(NS1/$S^1_{\mathrm{B}}$)}
      \\
      & 
      \mathcolor{gray}{= 
      H_2 
        \wedge
      F_6^{\mathrm{A}}
        -
      F_5 
        \wedge
      F_3^{\mathrm{A}}
        -
      F_1 
        \wedge
      F_7^{\mathrm{A}}
      }
    \end{aligned}
    $}
    \\
    \;\;\;\;\;\;\;
    \adjustbox{
      scale=1.7,
      raise=-2pt
    }{$\times$}
    \ar[
      r,
      -,
      dotted,
      shorten <=-7pt
    ]
    \&
    \adjustbox{rndfbox=4pt}{$
    \begin{aligned}
      & 
      \mathcolor{gray}{\mathrm{d}}\, 
      F_8
      \;\;
      \text{(D0)}
      \\
      & 
      \mathcolor{gray}{
        = H_3 \wedge F_6
      }
    \end{aligned}
    $}
    \ar[
      r,
      |->,
      shorten=10pt
    ]
    \ar[
      uur,
      |->,
      crossing over,
      shorten=10pt,
      end anchor=west
    ]
    \&
    \adjustbox{rndfbox=4pt}{$
    \begin{aligned}
      & 
      \mathcolor{gray}{\mathrm{d}}\, 
      F_8^{\phantom{\mathrm{A}}}
      \;\;
      \text{(D(-1))}
      \\
      &
      \mathcolor{gray}{
      = 
      F_2^{\mathrm{A}}
        \wedge
      F_7^{\mathrm{A}}
        +
      F_3^{\mathrm{M}}
        \wedge
      F_6^{\mathrm{M}}
      }
    \end{aligned}
    $}
  \end{tikzcd}
  \hspace{-1.35cm}
}
\rotatebox[origin=c]
  {-90}
  {{magnetic}\hspace{7.5cm}{electric}}
\end{figure}

\item 
\label{HigherGlobalSymmetries}
\textbf{Higher Global Symmetries.}
While we have shown that twisted K-theory is an abelian approximation to small-cyclified 4-Cohomotopy up to 10-truncation, beyond that the two cohomology theories diverge more substantially and increasingly so (\cref{RationalHomotopyGroups}). 

This divergence above 10-truncation has no effect on the classification of brane charges in the nonabelian cohomology of a 9-dimensional Cauchy surface, $\pi_0\mathrm{Map}\bracket({ X^9, \mathcal{A} }) \defneq H^1\bracket({ X^9; \Omega\, \mathcal{A} })$ (for $\mathcal{A}$ one of the two classifying spaces), but it does mean that the higher homotopy groups $\pi_n \mathrm{Map}\bracket({ X^9, \mathcal{A} })$ eventually witness a substantial difference. 

The homotopy type of $\mathrm{Map}\bracket({ X^9, \mathcal{A} })$, partially reflected in these higher homotopy groups, is the $\infty$-groupoid of \emph{large} higher gauge transformations of the flux-quantized higher gauge fields (\cite[\S~4.2.1, p. 32]{SS26-HigherGauge}, cf. \parencites[p. 4]{Gwilliam2025}{PerezLona2025}[\S~3.4]{AlfonsiKimLuciani2026}%
\footnote{%
  To compare these references to our discussion, notice that
   \cite[\S~3.4]{AlfonsiKimLuciani2026}
   addresses the (transgressive) 
   \emph{cohomology} (instead of the homotopy) of the mapping space as generalized symmetries, which retains a coarsened image of the homotopy type; while 
  \cite{PerezLona2025} considers the mapping space into $\mathrm{Aut}(\mathcal{A})$ instead of into $\mathcal{A}$, but for given charges $\big[\inlinetikzcd{ X \ar[r, "{ a }"] \& \mathcal{A} }\big]$,   
  evaluation provides a canonical comparison map 
  $\inlinetikzcd{ 
    \mathrm{Map}\bracket({ 
      X, 
      \mathrm{Aut}\bracket({
        \mathcal{A}
      }) 
    }) 
    \ar[
      r, 
      "{ \mathrm{ev}_a }"
    ] 
    \& 
    \mathrm{Map}\bracket({ 
      X, 
      \mathcal{A} 
    }) 
  }$.
}%
), and hence reflects \emph{higher global symmetries} (also known as \emph{categorified symmetries} \cite{SchreiberSkoda2009}, or  \emph{generalized symmetries}).

Therefore, while \emph{Hypothesis H} and twisted K-theory essentially agree on brane charges (away from the NS1-brane and up to torsion charges), their implied higher global symmetries differ substantially and increasingly so as the degree increases (cf. \cref{RationalHomotopyGroups}).
This could in principle serve as a further test of the hypotheses H/K against physical expectations obtained by other means.

\end{enumerate}

We leave these follow-up questions for future work.

\appendix
\section
{Background}

We highlight a couple of pertinent points that tend to receive unclear discussion in the literature.

\subsection
{Measuring D-Brane Charge}
\label
{MeasuringDBraneCharge}

One needs to distinguish (cf. \cite[\S~2.2]{SS25-Flux}) between:%
\footnote{%
  Beware that the meaning of the terminology ``black brane'' and ``solitonic brane'' is not standardized across the literature. Some authors in the past have said ``solitonic'' for what we call ``singular''. The reader in doubt should take \cref{BlackBranes,SolitonicBranes} as the \emph{definition} of how we use terminology for the present purpose (following \cite[\S~2.2]{SS25-Flux}).
}
\begin{enumerate}
\item 
\label{BlackBranes}
backreacted, singular, \textbf{black branes} (cf. \parencites{Marolf2012}) that are spacetime singularities (not necessarily curvature singularities but necessarily flux density singularities for the flux they source) and as such are to be \emph{deleted from spacetime}, as familiar from black hole singularities,

\item 
\label{SolitonicBranes}
nonsingular \textbf{solitonic branes} which are finite bumps of flux density \emph{vanishing at infinity} (as is characteristic for solitons).
\end{enumerate}

For brevity, here we focus on \cref{BlackBranes}, singular black branes whose charge is on their \emph{linking spheres}. (The charge of solitonic branes, such as the Abrikosov vortices of 4D electromagnetism, is instead on the one-point compactification of their transverse space.)

The concept is best nailed down by first recalling its classical template: 

\paragraph
{Black 0-Branes of 4D Einstein--Maxwell theory} 
These are the Reissner--Nordstr{\"o}m spacetimes, topologically the \emph{complement} of the worldline of a point charge (which itself is excised from the spacetime, whose timelike singularity it is, or would be if included, cf. \parencites{Carter1968}{Carter2009}): 
\begin{equation}
  X^{1,3}_{\mathrm{RN}}
  \,\simeq\,
  \mathbb{R}^{1,0} 
  \times
  X^3_{\mathrm{RN}}
  \,\simeq\,
  \smash{%
  \overbrace{%
    \mathbb{R}^{1,0}
  }^{%
    \mathclap{%
      \substack{%
        \text{singular worldline}
      }
    }
  }
  }
  \times 
  \mathbb{R}_{> 0}
  \times
  \smash{
  \overbrace{
    S^2
  }^{
    \mathclap{
      \substack{
        \text{linking sphere}
      }
    }
  }
  }
  \,\simeq\,
  \mathbb{R}^{1,3}
  \setminus
  \mathbb{R}^{1,0}
  \mathrlap{\,.}
\end{equation}
The historical argument of Dirac (\cite{Dirac1931}, cf. \parencites{Alvarez1985}[\S~2.1]{SS25-Flux}) entails, in modern paraphrase, that the quantized charge carried by these black 0-branes is the (ordinary integral) 2-cohomology of their ambient spacetime, hence of their linking sphere, which reveals that they appear in integer multiples of a unit charge brane:
\begin{equation}
  H^2\bracket({
    X^{3}_{\mathrm{RN}}
    ;\,
    \mathbb{Z}
  })
  \,\simeq\,
  H^2\bracket({
    \mathbb{R}_{> 0} \times S^2
    ;\,
    \mathbb{Z}
  })
  \,\simeq\,
  H^2\bracket({
    S^2
    ;\,
    \mathbb{Z}
  })
  \,\simeq\,
  \mathbb{Z}
  \mathrlap{\,.}
\end{equation}
Or rather, with both electric and magnetic charges quantized as in \cref{IRCompletionOfSuGra}, and adopting the general notation of nonabelian cohomology as in \cref{ProperFluxQuantization}, we see the 2-dimensional lattice of electromagnetic charges as:
\begin{equation}
    H^1\bracket({
      X^{3}_{\mathrm{RN}}
      ;\,
      \Omega
      B^2 \mathbb{Z}^2
    })
    \,\simeq\,
    \pi_0
    \,
    \mathrm{Map}\bracket({
      S^2
      ,\,
      B^2 \mathbb{Z}^2
    })
    \,\simeq\,
    \mathbb{Z}^2
  \mathrlap{\,.}
\end{equation}
The analogous argument shows that there are no black $(p > 0)$-branes in 4D Einstein--Maxwell, because locally near their (singular) worldvolume we see charge
\begin{equation}
    H^1\bracket({
      \mathbb{R}^3 \setminus 
      \mathbb{R}^p
      ;\,
      \Omega
      B^2 \mathbb{Z}^2
    })
    \,\simeq\,
    \pi_0
    \,
    \mathrm{Map}\bracket({
      S^{2-p}
      ,\,
      B^2 \mathbb{Z}^2
    })
    \,\simeq\,
    \begin{cases}
      \mathbb{Z}^2
      & 
      \text{ if $p = 0$}
      \\
      0
      &
      \text{ otherwise. }
    \end{cases}
\end{equation}

\paragraph
{Black Brane Charge in General}
Clearly then (cf. \cite{SS25-Flux}) for \emph{any} flux quantization law $\mathcal{A}$ (as in \cref{ProperFluxQuantization}) the brane charge on a Cauchy surface $X^d$ is its nonabelian $\Omega\mathcal{A}$-cohomology \cref{NonabelianCohomology}. Specifically, charges carried by $p$-branes as seen ``near-horizon'' are measured on their near linking sphere (cf. \parencites[\S~1.2.1]{Marolf2012}), since:
\begin{equation}
  \overbrace{
  \mathbb{R}^{1,d}
  \setminus
  \mathbb{R}^{1,p}
  }^{%
    \substack{%
      \text{complement of}
      \\
      \text{singular brane}
    }
  }
  \;\;\simeq\;\;
  \overbrace{
  \mathbb{R}^{1,p}
  }^{%
    \mathclap{%
      \substack{%
        \text{singular brane}
        \\
        \text{worldvolume}
      }
    }
  }
  \;\times\;
  \overbrace{
  \bracket({
    \mathbb{R}^{d-p}
    \setminus \{0\}
  })}^{
    \mathclap{%
      \substack{%
        \text{punctured}
        \\
        \text{transverse space}
      }
    }
  }
  \;\;\simeq\;\;
  \underbrace{
  \overbrace{
  \mathbb{R}^{1,p}
  }^{%
    \mathclap{%
      \substack{%
        \text{singular brane}
        \\
        \text{worldvolume}
      }
      \;
    }
  }
  \;\times\;
  \overbrace{%
    \mathbb{R}^1_{> 0}
  }^{\mathclap{
    \;
    \substack{
      \text{radial}
      \\
      \text{direction}
    }
  }}
  }_{%
    \substack{%
      \text{contractible}
    }
  }
  \;\times\;
  \overbrace{
    S^{d-p-1}
  }^{%
    \substack{%
      \text{linking}
      \\
      \text{sphere}
    }
  }
  \mathrlap{\,,}
\end{equation}
and hence take values in the $\bracket({d-p-1})$-st homotopy group of $\mathcal{A}$ (assuming here for brevity that $\mathcal{A}$ is simply-connected):
\begin{equation}
\label{GeneralFormulaForBraneCharge}
  \mathllap{
  \substack{
    \text{\color{gray}fundamental}
    \\
    \text{\color{gray}$p$-brane charges}
  }
  \;\;
  }
  H^1\bracket({
    \mathbb{R}^d 
    \setminus
    \mathbb{R}^p
    ;\,
    \Omega
    \,
    \mathcal{A}
  })
  \,\simeq\,
  \pi_0
  \mathrm{Map}\bracket({
    S^{d-p-1}
    ,\,
    \mathcal{A}
  })
  \,\simeq\,
  \pi_{d-p-1}\bracket({
    \mathcal{A}
  })
  \mathrlap{\,.}
\end{equation}
For instance:
\begin{enumerate}
\item 
with $\mathcal{A} \defneq B^3 \mathbb{Z} \times B^7 \mathbb{Z}$ in 10D \cref{GaussLawOfBField} this gives a 5-brane and a 1-brane (the NS branes in type IIA, cf. \cref{DBraneAndSourcedFluxes}):
\begin{equation}
\label
{NSBraneChargesDerived}
  H^1\bracket({
    \mathbb{R}^9 \setminus \mathbb{R}^p
    ;\,
    \Omega
    \bracket({
      B^3 \mathbb{Z}
      \times
      B^7 \mathbb{Z}
    })
  })
  \simeq
  \pi_{8-p}\bracket({
    B^3 \mathbb{Z}
    \times
    B^7 \mathbb{Z}
  })
  \simeq
  \begin{cases}
    \mathbb{Z}
    &
    \text{ if $p \in \{5,1\}$}
    \\
    0
    &
    \text{ otherwise}
    \mathrlap{\,,}
  \end{cases}
\end{equation}

\item
with $\mathcal{A} \defneq S^4$ in 11D \cref{S4ValuedGaussLaws} one finds M-brane charges as displayed in \cref{M5ChargeAccordingToHypothesisH,M2ChargeAccordingToHypothesisH}.
\end{enumerate}

\paragraph
{Black D-Brane Charges in IIA}

Consider $\mathcal{A} \defneq B \mathrm{U}$, the classifying space of \emph{rank-zero} K-theory (no Romans mass term) inside the full classifying space $\mathrm{KU}_0 \simeq B \mathrm{U} \times \mathbb{Z}$ of K-theory. This is the right classifying space to reproduce the expected black D-brane charges in non-massive IIA, according to the general formula \cref{GeneralFormulaForBraneCharge}, in that:
\begin{equation}
\label{DBraneChargeInOmegaBUCohomology}
  H^1\bracket({
    \mathbb{R}^{9}
    \setminus
    \mathbb{R}^p
    ;\,
    \Omega\, 
    B \mathrm{U}
  })
  \simeq
  \pi_0
  \,
  \mathrm{Map}\bracket({
    S^{8-p}
    ,\,
    B \mathrm{U}
  })  
  \simeq
  \pi_{8-p}\bracket({
    B \mathrm{U}
  })
  \simeq
  \begin{cases}
    \mathbb{Z}
    & 
    \text{ if } p \in \{0,2,4,6\}
    \\
    0
    & 
    \text{ otherwise, } 
  \end{cases}
\end{equation}
which identifies exactly the sources for the RR-fluxes available in non-massive IIA, cf. \cref{11DAnd10DFluxSpecies,DBraneAndSourcedFluxes}.

\begin{SCtable}[.9][htb]
\caption{%
\label{DBraneAndSourcedFluxes}%
The NS/D-brane species in IIA sourcing, via  \cref{DBraneChargeInOmegaBUCohomology,NSBraneChargesDerived}, the flux species from \cref{11DAnd10DFluxSpecies}.
\newline
(We refrain from showing an entry for ``D(-2)-branes'' and their $F_{10}$ flux in massive IIA, cf. \cref{ConsequenceOfCauchyPrinciple} \cref{NoPhysicalDMinus2Charge}.)}
\adjustbox{rndfbox=4pt, scale=0.9}{
\begin{tblr}{
  colspec = {l||cc|cccc|c},
  row{2} = {bg=gray!30}
}
  & 
  \SetCell[c=6]{c}
    \textbf{--- Type IIA --- }
  & & & & & &
  $\substack{
    \text{massive}
    \\
    \text{IIA}
  }$
  \\
  brane: 
  & 
  NS1
  &
  NS5
  &
  D0
  & 
  D2 
  & 
  D4 
  & 
  D6 
  &
  D8
  \\
  $\substack{
    \text{sourced}
    \\
    \text{flux:}
  }$
  & 
  $H_7$
  & 
  $H_3$
  &
  $F_8$
  & 
  $F_6$ 
  & 
  $F_4$ 
  & 
  $F_2$ 
  &
  $F_0$
  \\
  \hline
  Sector:
  &
  \SetCell[c=2]{c}
    \textbf{NS}
  & &
  \SetCell[c=5]{c}
    \textbf{RR}
  & & & & 
\end{tblr}
}
\end{SCtable}

The analogous statement for type I D-brane charges seen in rank-zero $\mathrm{KO}$-theory is:
\begin{equation}
\label{DBraneChargeInOmegaBOCohomology}
    H^1\bracket({
      \mathbb{R}^9
        \setminus
      \mathbb{R}^p
      ;\,
      \Omega
      \,
      B \mathrm{O}
    })
    \simeq
    \pi_0\,
    \mathrm{Map}\bracket({
      S^{8-p}
      ,\,
      B \mathrm{O}
    })
    \simeq
    \pi_{8-p}\bracket({
      B \mathrm{O}
    })
    \simeq
    \begin{cases}
      \mathbb{Z}
      &
      \text{ if } 
      p \in \{0,4\}
      \\
      \mathbb{Z}_{/2}
      &
      \text{ if }
      p \in \{6,7\}
      \\
      0
      &
      \text{ otherwise. }
    \end{cases}
\end{equation}

\paragraph
{As seen in the traditional literature}
The traditional literature on \emph{Hypothesis K} instead insists on using the full $\mathrm{KU}_0 \simeq \mathbb{Z} \times  B \mathrm{U}$ as the classifying space. This then requires quotienting out by hand the spurious (in non-massive IIA) D8-brane charge (cf. \cite{BergshoeffdeRooGreenPapadopoulosTownsend1996}) reflected in $\pi_0\bracket({ B \mathrm{U} \times \mathbb{Z} }) \simeq \mathbb{Z}$.

One way to write this suggestively is to consider the linking $(8-p)$-sphere in \cref{DBraneChargeInOmegaBUCohomology} as the boundary of the $(9-p)$-ball filling it (which is physical in case the black $p$-brane enclosed by the sphere vanishes),
\begin{equation}
  \begin{tikzcd} 
    S^{8-p}
    \ar[
      r, 
      hook,
      "{ i }"
      ]
    &
    D^{9-p}
    \mathrlap{\,,}
  \end{tikzcd}
\end{equation}
and then to observe that the rank-zero K-cohomology \cref{DBraneChargeInOmegaBUCohomology} may equivalently be written as the following quotient by the spurious (in plain IIA) D8-charge group $\mathbb{Z}_{\text{``D8''}}$: 
\begin{equation}
  H^1\bracket({
    \mathbb{R}^9 \setminus \mathbb{R}^p
    ;\,
    \Omega B \mathrm{U}
  })
  \simeq
  \frac{
    \grayoverbrace{
    \mathrm{KU}^0\bracket({
      S^{8-p}
    })
    }{ 
      \mathbb{Z}_{\text{D$p$}} 
      \times
      \mathbb{Z}_{\text{``D8''}}
    }
  }
  {
    \grayunderbrace{
    i^\ast
    \mathrm{KU}^0\bracket({
      D^{9-p}
    })
    }{
      \mathbb{Z}_{\text{``D8''}}
    }
  }
  \,\simeq\,
  \mathbb{Z}_{\text{D$p$}}
    :=
  \begin{cases}
    \mathbb{Z} & \text{ if  } p \in \{0,2,4,6\}
    \\
    0 & \text{ otherwise. }
  \end{cases}
\end{equation}
This quotient expression is how black D-brane charge in IIA is explained by Seiberg, Maldacena, Moore, Witten and others in
\parencites
[(1.6)]{MaldacenaMooreSeiberg2001}
[(2.2)]{BergmanGimonKol2001}
[(2.10)]{MooreWitten2000}.
Earlier, the idea that IIA D-brane charge is to be measured in $\mathrm{KU}^0$ of the linking sphere, as in \cref{DBraneChargeInOmegaBUCohomology}, goes back to \parencites[(3.13)--(3.16)]{Horava1998}, was stated for the case of orientifolds in \cite[p. 2]{BergmanGimonSugimoto2001}, and was recently reiterated in \parencites[p. 6]{KaidiTachikawaYonekura2024}, perhaps without accounting for the D8-brane issue.


\subsection
{Duality and Cauchy Surfaces}
\label
{IndependenceOfDualFluxesOnCauchySurf}

Throughout, we assume globally hyperbolic spacetimes (cf. \parencites[\S~3.11]{MinguzziSanchez2008}[Thm. 1.1]{BernalSanchez2005}, implicit in most of the literature on brane charges):
\begin{equation}
\label{CauchySurface}
  X^{1,d}
  \,\simeq\,
  \mathbb{R}^{1,0}
  \times X^d
  \,,
  \;\;\;
  \text{ with Cauchy surface }
  \inlinetikzcd{
    X^d 
    \simeq
    \{t = 0\} \times X^d
    \ar[
      r,
      hook,
      "{ i_0 }"
    ]
    \&
    X^{1,d}
    \mathrlap{\,.}
  }
\end{equation}

We highlight two  basic but important points:
  \begin{enumerate}
  \item 
  Brane charges (cf. \cref{ProperFluxQuantization} and \cref{MeasuringDBraneCharge}) are entirely on the Cauchy surface $X^d$, since there is no topology in the time direction:
  \begin{equation}
    H^1\bracket({
      X^{1,d}
      ;\,
      \Omega
      \mathcal{A}
    })
    \,\simeq\,
    H^1\bracket({
      X^d
      ;\,
      \Omega
      \mathcal{A}
    })
    \mathrlap{\,.}
  \end{equation}

  \item Here on the Cauchy surface, the electric-magnetic flux/charge pairs become \emph{independent} of each other
  (\parencites[Thm. 2.2]{SS24-Phase}[Thm. 3.4]{SS26-HigherGauge}): Although on spacetime,
  \begin{equation}
  \label{HodgeDualityRelation}
    G_{d+1-p}
    \,=\,
    \star_{1,d}
    F_p
  \end{equation}
  is the on-shell EM dual flux to $F_p$, and as such fully determined by \cref{HodgeDualityRelation} once $F_p$ is given \emph{on spacetime}, their restrictions (pullbacks) to the Cauchy surface \cref{CauchySurface} are \emph{independent} of each other, in that the pair
  \begin{equation}
    \left(
    \begin{aligned}
      E_{d+1-p}
      & 
      :=
      i_0^\ast G_{d+1-p}
      \mathrlap{\,,}
      \\
      B_{p}
      & 
      :=
      i_0^\ast F_{p}
    \end{aligned}
    \; \right)
  \end{equation}
  is initial value \emph{Cauchy data} for the higher Maxwell equations of motion of $F_p$. The Hodge duality relation \cref{HodgeDualityRelation} is instead what determines the evolution of these data \emph{away from the Cauchy surface} into spacetime.
  \end{enumerate}

  This is familiar in the archetypical example:
  \begin{example}
  For 4D Maxwell theory in temporal gauge, the electric flux density $E_2$ on $X^3$ is the \emph{canonical momentum} to the magnetic gauge potential $A_1$, locally on $X^3$ (cf. \parencites[\S~19.1.1]{HenneauxTeitelboim1992}[\S~5.1]{BlaschkeGieres2021}), as such \emph{independent} from $A_1$ as well as from the magnetic flux density:
  \[
    \begin{tikzcd}[
      row sep=0pt, 
      column sep=20pt
    ]
      {}
      \ar[
        r,
        phantom,
        "{ \text{canonical} }"
      ]
      & 
      {}
      \\
      \text{coordinate}
      &
      \text{momentum}
      \\[10pt]
      \mathllap{
       \substack{
         \text{magnetic}
         \\
         \text{potential}
       }
       \;
      }
      A_1 
      \ar[
        d,
        |->,
        "{ 
          \mathrm{d}_{X^3}
         }"{swap}
      ]
      & 
      E_2
      \mathrlap{
       \;
       \substack{
         \text{electric}
         \\
         \text{flux}
       }
       \;\;
      }
      \\[25pt]
      \mathllap{
       \substack{
         \text{magnetic}
         \\
         \text{flux}
       }
       \;
      }
      B_2
      \ar[
        ur,
        dotted,
        <->,
        "{
          \substack{ 
            \text{independent}
          }
        }"{sloped, swap}
      ]
    \end{tikzcd}
    \;\;\;\;
    \;\;\;\;
    \text{on the Cauchy surface $X^3$.}
  \]
\end{example}

This means \parencites{SS24-Phase}[\S~3.2]{SS26-HigherGauge} that:
\begin{standout}
  Flux/charge quantization takes place on a Cauchy surface, respecting the Gauss laws but independent of the duality constraint.
\end{standout}

\smallskip 
\begin{remark}
\label[remark]{ConsequenceOfCauchyPrinciple}
In the main text, this principle explains notably:
\begin{enumerate}
  \item why classifying spaces like $B \mathrm{U} \sslash B \mathrm{U}(1)$ or $\mathrm{Cyc}_1 S^4$ really classify electromagnetic charges without overcounting,
  
  \item 
  \label{NoPhysicalDMinus2Charge}
  why the would-be universal D(-2) class $\widetilde{F}_{10}$---which survives on classifying spaces in \cref{ChernCharacterOfTildePhi}---does not appear as a physical charge (as it should not in non-massive IIA, cf. \cref{DBraneAndSourcedFluxes}): its pullback to the Cauchy surface $X^9$ vanishes for degree reasons.
\end{enumerate}
\end{remark}


\paragraph{Declarations}
\begin{description}
  \item[]
  The authors declare that they have no conflict of interest.

  \item[]
  No data were generated or analyzed in the course of this theoretical work.
\end{description}

\printbibliography

\end{document}